\documentclass[draftcls, onecolumn]{IEEEtran}
\usepackage[english]{babel}
\usepackage[utf8x]{inputenc}
\usepackage{amssymb,amsfonts,amsthm}
\usepackage{dsfont}
\usepackage{graphicx}
\usepackage{blkarray}
\usepackage[colorinlistoftodos]{todonotes}
\usepackage{amsmath}

\newtheorem{mydef}{Definition}
\newtheorem{proposition}{Proposition}
\newtheorem{Lemma}{Lemma}
\newtheorem{assumption}{Assumption}
\newtheorem{remark}{Remark}
\newtheorem{theorem}{Theorem}
\newtheorem{corollary}{Corollary}

\ifCLASSINFOpdf
\else
\fi
\usepackage{amsthm}
\usepackage{cite}
\usepackage{balance}
\usepackage{latexsym}
\usepackage{graphicx}
\usepackage{caption}
\usepackage{soul,color}
\usepackage{amsmath}
\usepackage{algorithm}
\usepackage{subfigure}
\usepackage{authblk}
\usepackage[noend]{algpseudocode}
\usepackage{alphalph}
\begin{document}

\title{Asymptotically Optimal Delay-aware  Scheduling in Queueing Systems \footnote{ Part of this work has been presented at the IEEE International Symposium of Information Theory \cite{isit2019} }}

%\begin{titlepage}

\author[1]{\small Saad Kriouile}\author[1]{\small Mohamad Assaad }\author[2]{\small Maialen Larranaga}
\affil[1]{TCL Chair on 5G, Laboratoire des Signaux et Syst\`emes CentraleSup\'elec,  91192 Gif sur Yvette, France}\affil[2]{ASML, P.O. Box 324, 5500 AH Veldhoven, The Netherlands}

\maketitle

\begin{abstract}
In this paper, we investigate a general delay-aware channel allocation problem where the number of channels is less than that of users. Due to the proliferation of delay sensitive applications, the objective of our problem is chosen to be the minimization of the total average backlog queues of the network in question. First, we show that our problem falls in the framework of Restless Bandit Problems (RBP), for which obtaining the optimal solution is known to be out of reach. To circumvent this difficulty, we tackle the problem by adopting a Whittle index approach. To that extent, we employ a Lagrangian relaxation of the original problem and prove it to be decomposable into multiple one-dimensional independent subproblems. Afterwards, we provide structural results on the optimal policy of each of the subproblems. More specifically, we prove that a threshold policy is able to achieve the optimal operating point of the considered subproblem. Armed with that, we show the indexability of the subproblems and characterize the Whittle’s indices which are the basis of our proposed heuristic. We then provide a rigorous mathematical proof that our policy is optimal in the infinitely many users regime. Finally, we provide  numerical results that showcase the remarkable good performance of our proposed policy and that corroborate the theoretical findings.
\end{abstract}
\section{Introduction}
This paper deals with user and channel scheduling, which has been widely recognized as a mean to improve the network performance and to meet the service demands of the users. This problem has been widely studied in the past and several allocation policies have been developed for various contexts (e.g. see\cite{deghel2018queueing,destounis2014traffic,destounis2015traffic,tassiulas1992stability,tassiulas1993dynamic,neely2007optimal,georgiadis2006resource}
and the references therein).   
In 5G networks, the problem of channel and user scheduling will be receiving particular interest due to the increase in the number of devices and users. Furthermore, the applications nowadays do not need high data rates only but they are more delay-sensitive, which implies that minimizing the delay is considered as a main design metric in future networks.

In this paper, we consider the problem of scheduling and channel allocation in a discrete time system composed of one central scheduler serving  multiple  users or queues. We consider that the traffic arriving to each queue is time varying, and that the number of users is  higher than the number of channels, which is a quite realistic assumption especially with the growth in density of users in nowadays networks.    At each timeslot, the central scheduler decides to allocate the channels to users, where a channel can be seen as a server in wired networks or a frequency bandwidth in wireless networks. Throughout this paper, we will use the terms "channel" and "server" interchangeably to designate a resource to allocate to users.  Furthermore, we assume that the number of channels is limited and each channel can only be allocated to one user at a time. The objective in this case is to find an allocation policy that minimizes the long-run  average queue length of the users, as a mean to minimize the average delay in the network. Although it is a quite standard scheduling, we provide in this paper  a  rigorous mathematical analysis, leading to a novel scheduling algorithm of which we prove optimality in the many users regime. In fact,   we show in this paper that the considered scheduling problem can be cast as a Restless Bandit Problem (RBP), which is a particular Markov Decision Processes (MDP). However, RBPs are PSPACE-Hard (see Papadimitriou et al. \cite{papadimitriou1999complexity}), and hence their optimal solution is out of reach. One should therefore propose sub-optimal policies when dealing with such problems. In this paper, we approach the considered RBP problem using the Lagrangian relaxation technique, which consists of relaxing the constraint on the available resources. In other words, instead of having the constraint on the number of available channels satisfied in every time slot, we consider that it has to be satisfied on average. This allows us to decompose the large relaxed optimization problem into much simpler one-dimensional problems. Based on the optimal solution of the individual relaxed problems, we develop a heuristic for the original (i.e. non-relaxed) optimization problem. This heuristic is known as the Whittle's index policy (WIP) and we will show that for our particular model, an explicit expression of the Whittle's index can be found. WIP has been proposed as a suboptimal policy for many problems in the literature, see for instance \cite{liu2010indexability,ansell2003Whittle}. It has also been shown to perform near optimally in many scenarios and in the particular case of multiclass M/M/1 queues, WIP which simplifies to the $c\mu$-rule is optimal, see Buyukkoc et al. \cite{buyukkoc1985c}, and Larranaga \cite{larranaga2015dynamic}. In this paper, we will prove that the developed WIP is asymptotically optimal in the many users regime. To that extent, we summarize in the following the key
contributions of this paper:
\begin{itemize}
	\item We provide an analysis of the relaxed optimization problem, which let us obtain the structure of the optimal solution of its dual problem. The optimal solution is shown to be a threshold-based policy by (i) proving that the latter problem is decomposable and (ii) proving that the value function of the Bellman equation that resolves each individual dual problem satisfies both the R-convexity and increasing properties. This part of the analysis is far from trivial and constitutes one the main contributions in this paper. 
	\item We resolve the full balance equations verified by the stationary distribution of the user’s states under a general threshold policy $n$. This step is very crucial and requires a lot of analysis and computations. In fact, unlike the other previous works where the full balance equations give an easy general recurrent relation between the stationary distribution at state $i$ and state $i+1$ under a threshold policy (e.g \cite{ouyang2016downlink}), in our paper the term of the stationary distribution at any given state is linked to a set of terms of the stationary distribution at different states. Moreover, this relation depends on the value of the threshold $n$ as we will see in Section \ref{sec:steady_state}.%\in [-1,R-1]$, $n \in [R,L-R-1]$, $n \in [L-R,L-1]$, $n=L$ (where $R$ and $L$ are defined in Section \ref{sec:modelDes}). We derive for each case, the expression of the steady-state distribution of the system's states.  
	\item We reformulate the individual dual problem of the relaxed problem using the steady state distribution. Afterwards, we provide a general algorithm that allows us to obtain the Whittle index. To reduce even further the complexity, we provide a rigorous proof of the indexability of the classes, along with several lemmas and definitions that allow us to derive simple expressions of the Whittle index. While in previous works the derivation of whittle index policy can be obtained using a standard approach, obtaining Whittle index expressions in our case is much more complex and requires several derivations and lemmas.  
	\item Unlike the previous works, in this paper we provide further characterization of the threshold-based optimal solution of the relaxed optimization problem. The structure of this solution helps us to prove the local asymptotic optimality of our proposed policy as we just need to compare the average cost under the Whittle's Index policy with the optimal cost of the relaxed problem. The reason behind that is the fact that the latter is always less than the optimal cost of the original problem.
	\item We show that the Whittle's Index policy is asymptotically optimal in the infinitely many users regime, that is, when the number of users in the system as well as the available channels grow large.
	\item Finally, we provide numerical performance results of the Whittle's Index policy that corroborate our claims. %in systems with few users and we observe that even in this case WIP is near optimal.
\end{itemize} 
\subsection{Related Work}
%\section{related work}
The problem of resource allocation and scheduling in wireless networks has been widely studied in the literature. In \cite{deghel2018queueing,destounis2014traffic,destounis2015traffic,tassiulas1992stability,tassiulas1993dynamic}, throughput optimal schedulers have been derived for single channel, multi-channel and multi user MIMO contexts. The aforementioned set of work focuses on developing strategies that stabilize the queues of the users using the max weight rule. The classical max weight rule is however known to be not delay optimal. To overcome this issue, many works have been developed in the past to take into account the average delay of the traffic of the users (e.g. see \cite{cui2012survey} and the references therein). Most of the existing works use Markov Decision Process (MDP) frameworks and develop allocation strategies using Bellman equation (e.g. by using value iteration, policy iteration, etc.). However, MDP frameworks and Bellman equation suffer from the curse of dimensionality, which leads to complex resource allocation strategies. In \cite{bettesh2006optimal}\cite{wang2013delay}, the authors try to minimize the average delay of the users' queues using Markov Decision Process (MDP) and stochastic learning tools. The complexity of the developed solutions is however much higher than the Whittle index policy. Stochastic learning is also used in \cite{cui2010distributive} to deal with the problem of power allocation in an OFDM (Orthgonal Frequency Division Multiplexing) system with the goal being to minimize the average delay of the users' packets in the queues. The developed solution requires high memory and computational complexity as compared to the Whittle index policy. 

Whittle index based policies have also been used/developed in wireless networks to deal with the problem of pilot allocation over Markovian channel models. If a pilot is allocated to a user, its CSI can be estimated correctly and the user can hence transmit at a given rate. In \cite{liu2010indexability}\cite{ouyang2016downlink}, a Gilbert-Elliot channel model is considered and the Whittle index is derived. It has been shown in \cite{ouyang2016downlink}\cite{ouyang2011exploiting} that a policy based on Whittle index is asymptotically optimal for their specific problem. The authors in \cite{larranaga2017asymptotically} extended the problem of pilot allocation to the case where the channel evolves according to a Markovian process between K states instead of two states as in the Gilbert-Elliot model. In the aforementioned papers, the queues of the users were not considered. In fact, the focus was on the channel allocation such that the long term total throughput (or equivalent objective function) is maximized without taking into account the dynamic traffic of the users. In this paper, we consider that the traffic arrival is bursty and that the objective of the user/channel allocation is to minimize the long term average queues of the users.  

In \cite{ansell2003Whittle}, a derivation of the Whittle index values for a simple multiclass M/M/1 model has been considered (where only one user can be served). However, the optimality of the obtained Whittle index policy has not been proved in \cite{ansell2003Whittle} and the time was assumed to be continuous in their model. The authors in \cite{weber1990index} considered the problem of project/job scheduling in which an effort is allocated to a fixed number of projects. The performance of a Whittle index based policy was analyzed under a continuous time model. In contrast to these two papers, we consider that the time is slotted and that several users can be scheduled at a given time slot and not only one user. We provide an explicit characterization of the Whittle indices, develop a Whittle index channel allocation policy for our problem and prove the asymptotic optimality of the developed policy in the many users regime.    

The remainder of the paper is organized as follows: In Section~\ref{sec:systemmodel}, we formulate the problem under investigation and we introduce the Lagrangian relaxation. In Section~\ref{sec:threshold policy}, we prove the optimality of threshold/monotone policies for the relaxed problem. In Section~\ref{sec:steady_state}, we compute the steady-state distribution of the system under a general threshold policy. In Section~\ref{sec:WI}, we characterize the Whittle indices explicitly and we lay out our proposed Whittle index based policy. Section~\ref{sec:relaxed} provides further characterization of the optimal solution of the relaxed problem. In Sections \ref{sec:optimality} and \ref{sec:globalOptimality}, we prove the local and global asymptotic optimality of our proposed scheme respectively. In Section~\ref{sec:numerics}, we evaluate the performance of the Whittle index policy numerically. Lastly, the mathematical proofs are provided in the appendices. 

\section{System Model and Problem Formulation}\label{sec:systemmodel}
\subsection{System model description}\label{sec:modelDes}
We consider a time-slotted system with one central scheduler, $N$ users/queues and $M$ uncorrelated channels (or servers) with ($N > M$). The terms "server" and "channel" will be used interchangeably throughout this paper, as well as the terms "user" and "queue".  A channel can be allocated to at most one user, hence only $M$ users will be able to transmit (i.e. send packets) at time slot $t$. We consider $K$ different classes of users and we assume that each user in class-$k$, if scheduled, transmits at most $R_k$ packets per time slot. We will refer to $R_k$ as the maximum transmission rate for every user in class $k$ and we assume that $\min_{k} \{ R_k\}\geq 2$. We denote by $\gamma_k$ the proportion of class-$k$ users in the system.  We further denote by $A_{i}^k(t)\in\{0,\ldots,R_{k}-1\}$ the number of packets that arrive to queue $i$ in class $k$ at time slot $t$. We also let $q^{k,\phi}_{i}(t)$ denote the number of packets in queue $i$ in class $k$. Furthermore, $s^{k,\phi}_{i}(\textbf{q}^\phi(t))$ will denote the transmission action under a decision policy $\phi$ and $\textbf{q}^\phi(t)$ the vector of all queue lengths $ (q^{1,\phi}_1(t),\ldots,q^{1,\phi}_{N \gamma_1}(t),\ldots, q^{K,\phi}_1(t),\ldots,q^{K,\phi}_{N \gamma_K}(t))$. For the sake of clarity, we define $s^{k,\phi}_{i}(t):=s^{k,\phi}_{i}(\textbf{q}^\phi(t))$. If policy $\phi$ prescribes to schedule user $i$ in class $k$ at time $t$, then $s^{k,\phi}_{i}(t)=1$, and $s^{k,\phi}_{i}(t)=0$ otherwise. We denote by $L$ the buffer capacity, which is considered to be the same for all queues and can be very high. The general system model is presented in Figure \ref{fig:system_model}. 
%We make the following assumption that is required for the derivations in Sections \ref{sec:steady_state}, \ref{sec:WI} and \ref{sec:optimality}.
%\begin{assumption}\label{assump:buffer_size_1}
%Let $R_k$ be the maximum transmission rate of class-$k$ users and $L$ the buffer size. We assume that $L \geq  2\max_{k}\{R_k\}$.
%\end{assumption}
%It is worth mentioning that the above assumption is useful to simplify the derivations of the stationary distribution and Whittle index values in Sections IV and V. The analysis in this paper can be easily extended to the case where this assumption is violated. Also, this assumption is a realistic one as the maximum buffer length $L$ is often much higher than the transmission rate $R_k$.

%We make the following assumption that is required for the derivations in Sections \ref{sec:steady_state}, \ref{sec:WI} and \ref{sec:optimality}.
%\begin{assumption}\label{assump:buffer_size_1}
%Let $R_k$ be the maximum transmission rate of class-$k$ users and $L$ the buffer size. We assume that $L \geq  2\max_{k}\{R_k\}$.
%\end{assumption}
%It is worth mentioning that the above assumption is useful to simplify the derivations of the stationary distribution and Whittle index values in Sections IV and V. The analysis in this paper can be easily extended to the case where this assumption is violated. Also, this assumption is a realistic one as the maximum buffer length $L$ is often much higher than the transmission rate $R_k$.

Based on our system model, the number of packets in queue $i$ of class $k$ evolves as follows: 
\begin{equation}\label{eq: qeue_evolution}
q^{k,\phi}_{i}(t+1)=\min\{(q^{k,\phi}_{i}(t)-R_ks^{k,\phi}_{i}(t))^{+}+A^k_i(t),L\},
\end{equation}
where $(x)^+ = \max\{x,0\}$.
\begin{figure}
\centering
\includegraphics[width=0.5\textwidth]{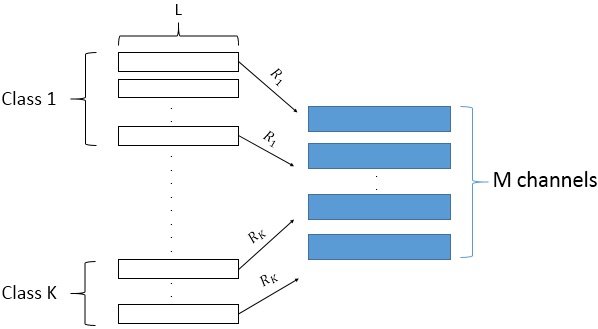}
\caption{System Model}
\end{figure}

%We also assume that at each queue $i$ in class~$k$, packets arrive according to a discrete uniform distribution, that is, $\mathbb{P}(A_{i}^k(t)=x)=\rho_k$ for all $0\leq x\leq R_k-1$ and $0$ otherwise, where $\rho_k=1/R_k$.
 %which are valid for all $R_k$, otherwise we have to separate the classes in tree groups depending whether $L \geq  2R_k$, $2R_k \geq L \geq  R_k$ or $L \leq  R_k$, which makes the computations more complicated and which is not very realistic, since the buffer length $L$ is often much higher than the transmission rate $R_k$.     
The objective of the present work is to find a scheduling policy $\phi$ that minimizes the average queue length of the users which results, according to Little Law, in the minimization of the average delay. 
\subsection{Problem formulation}
 The cost incurred by user $i$ in class $k$, at time $t$ is equal to $a_kq^{k,\phi}_{i}(t)$ for all $i\in\{1,\ldots,\gamma_kN\}$ where $a_k$ is a predefined weight. One can see that the model described in Section~\ref{sec:modelDes} belongs to the family of Restless Bandit Problems (RBP)  \cite{Whittle1988restless}. %where there are $K$ users and each user transmits at some rate $R_k$ with possibly $R_k=R_{k^{'}}$ for $k \neq k^{'}$.
We consider the broad class $\Phi$ of scheduling policies in which a scheduling decision depends on the history of observed queue states and scheduling actions. Our user and channel allocation problem therefore consists of identifying the policy $\phi$ $\in$ $\Phi$ that minimizes
the infinite horizon expected average queues,
subject to the constraint on the number of users selected
at each time slot. Given the initial state $\textbf{q}_0=(q^1_1(0),\ldots,q^1_{N\gamma_1}(0),...,q_{1}^K(0),\ldots,q^K_{N\gamma_K}(0))$, the
problem can be formulated as follows:
%\begin{equation}\label{eq:constraint}
\begin{align}
& \underset{\phi \in \Phi}{\text{min}} \limsup\limits_{T\rightarrow\infty}\frac{1}{T} \mathbb{E}\left[ \sum_{t=0}^{T-1}  \sum_{k=1}^{K} \sum_{i=1}^{{\gamma_k}N} a_kq_i^{k,\phi}(t) \mid \textbf{q}_0\right],\label{obj_function} \\
& \text{s.t.}  \sum_{k=1}^{K} \sum_{i=1}^{{\gamma_k}N} s_i^{k,\phi}(t)\leq {\alpha}N, \hbox{ for all } t,\label{eq:constraint_original}
\end{align}
where $\alpha=M/N$ is the fraction of users that can be scheduled.
\section{Relaxed Problem and Threshold-based Policy}\label{sec:threshold policy}
As it has been discussed in the introduction of this paper, RBPs are PSPACE-Hard (see Papadimitriou et al. \cite{papadimitriou1999complexity}) and therefore one should develop well performing sub-optimal policies to solve these problems. In this paper, the development of our policy is done through several steps. First, we consider a Lagrangian relaxation of our problem and show that it can be decomposed into several one-dimensional problems. We then prove that the optimal solution to each of these relaxed problems is a threshold-based policy. We then compute the stationary distribution of the states of the system under the aforementioned threshold policy. This allows us to obtain a closed form expression of the Whittle index values of the relaxed problem and develop a Whittle index-based scheduling policy for the original RBP. 

In this section, we first formulate the relaxed problem and prove that its optimal policy is a threshold-based one. 
\subsection{Relaxed Problem and Dual Problem}
The Lagrangian relaxation consists of relaxing the constraint on the available resources. Namely, we consider that the constraint in Equation~\eqref{eq:constraint_original}, has to be satisfied on average and not in every decision epoch, that is,
\begin{equation}\label{eq:constraint_relaxed}
\limsup\limits_{T\rightarrow\infty}\frac{1}{T} \mathbb{E}\left[\sum_{k=1}^{K} \sum_{i=1}^{{\gamma_k}N} s_i^{k,\phi}(t)\right] \leq {\alpha}N.
\end{equation}
Note that, contrary to the strict constraint in Equation~\eqref{eq:constraint_original}, the relaxed constraint allows the activation of more than $\alpha$ fraction of users at each time slot.
If we note $W$ the Lagrangian multiplier for the constrained problem, then the Lagrange function equals to:
\begin{equation}
f(W,\phi)=\limsup\limits_{T\rightarrow\infty}\frac{1}{T} \mathbb{E}\left[ \sum_{t=0}^{T-1}  \sum_{k=1}^{K} \sum_{i=1}^{{\gamma_k}N} (a_kq_i^{k,\phi}(t)+Ws_i^{k,\phi}(t)) \mid \textbf{q}_0\right]-W{\alpha}N,
\end{equation}
where $W$ can be seen as a subsidy for not transmitting.
Therefore, the dual problem for a given $W$ is
\begin{equation}
\underset{\phi \in \Phi}{\text{min}} \ f(W,\phi).
\end{equation}
\subsection{Problem Decomposition and Threshold-based Policy}\label{sec:ProbDecomp}
In this section, we show that the relaxed problem can be decomposed into $N$ one-dimensional subproblems, for which the optimal solution is a threshold-based policy. To do that, we first get rid of the constants that do not depend on $\phi$ and reformulate the problem as follows,% and our we consider the problem
\begin{equation}\label{eq:relaxed}
\underset{\phi \in \Phi}{\text{min}}\limsup\limits_{T\rightarrow\infty}\frac{1}{T} \mathbb{E}\left[ \sum_{t=0}^{T-1}  \sum_{k=1}^{K} \sum_{i=1}^{{\gamma_k}N} (a_kq_i^{k,\phi}(t)+Ws_i^{k,\phi}(t)) \mid \textbf{q}_0\right].
\end{equation}
One can see that the solution of this problem can be deduced from the well known Bellman equation (see Ross~\cite{ross2014introduction}). More specifically: 
\begin{equation}\label{eq:Bellman}
\bar{\textbf{V}}(\textbf{q})+ \theta= \underset{\textbf{s}}{\text{min}} \{ \sum_{k=1}^{K} \sum_{i=1}^{{\gamma_k}N} C_k(q_i^k,s_i^k)+ \sum_{\textbf{q}'} Pr(\textbf{q}'|\textbf{q},\textbf{s}) \bar{V}(\textbf{q'})\}, 
\end{equation}
 for all  $\textbf{q} = (q_1^1,\ldots,q_{\gamma_1N}^1,\ldots, q_{1}^K,\ldots,q_{\gamma_kN}^K)$, with $q_{i}^k\in\{1,\ldots,L\}$ being the queue length of class-$k$ user $i$, 
and $\textbf{s} = (s_1^1,\ldots,s_{\gamma_1N}^1,\ldots, s_{1}^K,\ldots,s_{\gamma_kN}^K)$, with $s_i^{k}\in\{0,1\}$ being the action taken with respect to user $i$ in class $k$. In equation~\eqref{eq:Bellman}, $V(\cdot)$ represents the Value Function, $\theta$ is the optimal average cost and $C_k(q_i^k,s_i^k)$ is the holding cost $a_kq_i^k+Ws_i^k$. The optimal decision for each state $\textbf{q}$ can be obtained by minimizing the right hand side of Equation~\eqref{eq:Bellman}. We now show that the problem can be decomposed into $N$ independent subproblems by decomposing  $\bar{V}(\cdot)$ into separate Value Functions for each user $i$ in class $k$, i.e., $V_i^k(\cdot)$. In other words, the optimal decision $\textbf{s}$ to problem~\eqref{eq:Bellman} is a vector composed of elements $s_i^k$, where each $s_i^k$ is nothing but the optimal decision that solves the individual Bellman equations.
\begin{equation} \label{eq:individual}
V_i^k(q_i^k)+ \theta_i^k= \underset{s_i^k}{\text{min}} \{ C_k(q_i^k,s_i^k)+ \sum_{q_i^{'k}} Pr(q_i^{'k}|q_i^k,s_i^k)V_i^k(q_i^{'k})\}.
\end{equation}
\begin{proposition}\label{prop:value_function_decomposition}
Let $V_i^k(\cdot)$ be the optimal value function that solves Equation~\eqref{eq:individual}, and let $\bar{\textbf{V}}(\cdot)$ be the optimal value function that solves Equation~\eqref{eq:Bellman} then:
\begin{align}
\bar{\textbf{V}}(\cdot)=\sum_{k=1}^K\sum_{i=1}^{\gamma_kN} V_i^k(\cdot).
\end{align}
\end{proposition}
\begin{proof}
See appendix \ref{app:value_function_decomposition}.
\end{proof}
In this section, we show that the solution to each individual problem (for each user $i$) follows the structure of a threshold policy. For ease of notation, we drop the indices $k$ and $i$ and consider that $V(\cdot)$ is the value function for a given user. We first provide the definition of threshold policies.
\begin{mydef} 
A threshold policy is a policy $\phi \in \Phi$ for which there exists  $n \in \{-1,0,\cdots,L\}$ such that when the queue of user $i$ is in state $ q \leq n$, the prescribed action is $s^- \in \{0,1\}$, and when the queue $ q > n$, the prescribed action is $s^+ \in \{0,1\}$ while baring in mind that $s^- \neq s^+$.\\
Since we only have two possible actions, a policy is of the form threshold policy if and only if it is monotone in $q$.
\end{mydef}
The solution of the bellman equation~\eqref{eq:individual} $V(\cdot)$ can be obtained by the well known Value iteration algorithm, which consists of updating $V_t(\cdot)$ using the following equation:
\begin{equation} \label{value_iteration}
V_{t+1}(q)= \underset{s}{\text{min}} \{ C(q,s)+ \sum_{q'} Pr(q'|q,s)V_t(q')\} - \theta
\end{equation}
We consider that the initial value function $V_0$ is equal to $0$ for any $q$, (i.e. for all $q$, $V_0(q)=0$). After many iterations, $V_t(\cdot)$ will converge to the unique fixed point of the equation ~\eqref{eq:individual} called $V(\cdot)$. However, the value iteration algorithm is known to have high complexity and can take a long time to converge. Therefore, we will give some structural properties of the value function $V_t(\cdot)$ for any $t$ and conclude that the optimal policy is a threshold-based one.\\   
For that, we consider the operator $TO$ such that for each $(q,s) \in \{0,\ldots,L\}  \times \{0,1\} $ \begin{equation}(TO(V))(q,s)= C(q,s)+ \sum_{q'} Pr(q'|q,s)V(q')-\theta
\end{equation}  
We first provide some useful definitions and preliminary results before proving the desired results.
\begin{mydef}
We say that function $f$ is R-convex in $X=\{0,\ldots,L\}$, if for any $x$ and $y$ in $X$ such that $x < y$, we have: \begin{equation}f(y+R)-f(x+R) \geq f(y)-f(x)
\end{equation} 
\end{mydef}
\begin{Lemma} \label{telescopique}
If for a given function $f$, there exists $R$ such that for any $x \in \{0,\ldots,L-1\}$, $f(x+1+R)-f(x+R) \geq f(x+1)-f(x)$, then $f$ is R-convex  
\end{Lemma}

\begin{proof}
Considering $y$ and $x$ in $\{0,\ldots,L-1\}$, with $y > x$, we have:
\begin{align}
f(y+R)-f(x+R)&=\sum_{k=x}^{y-1} [f(k+1+R)-f(k+R)]\\
&\geq \sum_{k=x}^{y-1} [f(k+1)-f(k)]\\
&=f(y)-f(x)
\end{align}
which concludes the proof.  
\end{proof}
\begin{mydef}
Let $g(q,s)$ be a real valued function defined on $X \times S$, with $S=\{0,1\}$, and $X=\{0,\ldots,L\}$. We say that $g$ is submodular if $g(q+1,1)-g(q+1,0) \leq g(q,1)-g(q,0)$ for all $q$ on $X$.
\end{mydef}
\begin{theorem}
$\underset{s}{\text{min}} TO(\cdot)(\cdot,s)$ conserves the $R$-convexity and increasing properties. In other words, if the input of the operator TO, i.e. a given $V(\cdot)$, is R-convex and increasing function in $q$ then $\underset{s}{\text{min}} TO(V)(\cdot,s)$ is R-convex and increasing function in $q$.  
\end{theorem}
\begin{proof}
Let us consider that the input of $TO(V)(\cdot,s)$), i.e. a given, $V(\cdot)$ is R-convex and increasing in $q$. %\cite{papadaki2002exploiting}.\\
For the increasing property of $\underset{s}{\text{min}} TO(V)(\cdot,s)$,  we have by definition that $C(\cdot,s)$ is increasing in $q$. We also have that $V(\cdot)$ is increasing in $q$ and the number of queue states is finite, then $\sum_{q'} Pr(q'|\cdot,s)V(q')$ is an increasing function in $q$ (see Puterman~\cite{puterman2014markov}). Since $\theta$ is a constant, $TO(V)(\cdot,s)$ is increasing in $q$ and therefore $\underset{s}{\text{min}} TO(V)(\cdot,s)$ is increasing in $q$.\\
For R-convexity, we should first prove the following lemma. 
%In the remaining of this section, we drop the indices $k$ and $i$  to simplify the notation (e.g. the queue length of a given user is denoted by q). 
\begin{Lemma}\label{submodilarity}
If $V(\cdot)$ is R-convex and increasing in $q$, $C(q,s)$ and $\sum_{q'} Pr(q'|q,s)V(q')$ are submodular functions.
\end{Lemma}
\begin{proof}
\renewcommand{\qedsymbol}{$\blacksquare$}
The proof is given in appendix \ref{app:submodilarity}.
\end{proof}
This demonstrates that the function $TO(V)(\cdot,\cdot)$ is submodular since it is the sum of two submodular functions. Let us now show that $\underset{s}{\text{min}} TO(V)(\cdot,s)$ is R-convex. For that, we consider the function $\Delta TO(V)(q)=TO(V)(q,1)-TO(V)(q,0)$ which is decreasing in $q$ since $TO(V)(\cdot,\cdot)$ is submodular. Therefore, there exists $r \in \mathbb{R}\cup \{+ \infty\}$ such that for $q \leq r$, \ $\Delta TO(V)(q) \geq 0$ and for $q\geq r$, \  $\Delta TO(V)(q)  \leq 0$. In the remainder of the proof, we consider all possible cases of $q$ and $r$.\\
If $q+R+1,q+R,q,q+1 \geq r$:\\
\begin{align}
\underset{s}{\text{min}} TO(V)(q+1+R,s)-\underset{s}{\text{min}} TO(V)(q+R,s)=&TO(V)(q+1+R,1)-TO(V)(q+R,1)&\\=&TO(V)(q+1,0)-TO(V)(q,0)\\
\geq& TO(V)(q+1,1)-TO(V)(q,1)&\\=&\underset{s}{\text{min}} TO(V)(q+1,s)-\underset{s}{\text{min}} TO(V)(q,s)
\end{align}
where the inequality is due to the sub-modularity of $TO(V)(\cdot,\cdot)$. \\ If $q \leq r \leq q+1,q+R,q+1+R$:\\
\begin{align}
\underset{s}{\text{min}} TO(V)(q+1+R,s)-\underset{s}{\text{min}} TO(V)(q+R,s)=&TO(V)(q+1+R,1)-TO(V)(q+R,1)\\
=&TO(V)(q+1,0)-TO(V)(q,0)\\
\geq& TO(V)(q+1,1)-TO(V)(q,0)\\
=&\underset{s}{\text{min}} TO(V)(q+1,s)-\underset{s}{\text{min}} TO(V)(q,s)
\end{align}
if $q,q+1 \leq r \leq q+R,q+R+1$:\\
\begin{align}
\underset{s}{\text{min}} TO(V)(q+1+R,s)-\underset{s}{\text{min}} TO(V)(q+R,s)=&TO(V)(q+1+R,1)-TO(V)(q+R,1)\\
=&TO(V)(q+1,0)-TO(V)(q,0)\\
=&\underset{s}{\text{min}} TO(V)(q+1,s)-\underset{s}{\text{min}} TO(V)(q,s)
\end{align}
if $q,q+1,q+R \leq r \leq q+R+1$:\\
\begin{align}
\underset{s}{\text{min}} TO(V)(q+1+R,s)-\underset{s}{\text{min}} TO(V)(q+R,s)=&TO(V)(q+1+R,1)-TO(V)(q+R,0)\\
\geq& TO(V)(q+1+R,1)-TO(V)(q+R,1)\\
=&TO(V)(q+1,0)-TO(V)(q,0)\\
=&\underset{s}{\text{min}} TO(V)(q+1,s)-\underset{s}{\text{min}} TO(V)(q,s)
\end{align}
If $q,q+1,q+R;q+R+1 \leq r$:\\
\begin{align}
\underset{s}{\text{min}} TO(V)(q+1+R,s)-\underset{s}{\text{min}} TO(V)(q+R,s)=&TO(V)(q+1+R,0)-TO(V)(q+R,0)\\
\geq& TO(V)(q+1+R,1)-TO(V)(q+R,1)\\
=&TO(V)(q+1,0)-TO(V)(q,0)\\
=&\underset{s}{\text{min}} TO(V)(q+1,s)-\underset{s}{\text{min}} TO(V)(q,s)
\end{align}
Using lemma \ref{telescopique}, $\underset{s}{\text{min}} TO(V)(\cdot,s)$ is R-convex in $q$, i.e., we can conclude the R-convexity conservation.\\

\end{proof}

\begin{remark}
Theorem 1 means that if the value function $V_t$ is increasing and R-convex, then the value function $V_{t+1}$ in equation ~\eqref{value_iteration}, which is computed with the operator $TO$, is increasing and R-convex.\\
Thus, as $V_0$ is increasing and R-convex, all $V_t$ are increasing and R-convex and therefore we can conclude that the value function $V$ will be also R-convex and increasing in $q$.
\end{remark}
\begin{corollary}\label{col:threshold_policy}
The optimal policy $\phi^*$ of each one-dimensional relaxed subproblem is a threshold-based policy. 
\end{corollary}
\begin{proof}
As explained in Definition $1$, it is sufficient to prove that the optimal policy $\phi^*$ is monotone in $q$. \\ We consider $q_1 \leq q_2$. According to Remark $1$, $V(.)$ is increasing and R-convex, then using lemma \ref{submodilarity}, $TO(V)$ is submodular. Therefore, we have:
\begin{equation}
(TO(V))(q_1,1)-(TO(V))(q_1,0) \geq (TO(V))(q_2,1)-(TO(V))(q_2,0)
\end{equation}
If $\phi^*(q_2)=\underset{s}{\text{argmin}} (TO(V))(q_2,s)=0$\\
Hence,
\begin{equation}
(TO(V))(q_1,1)-(TO(V))(q_1,0) \geq (TO(V))(q_2,1)-(TO(V))(q_2,0)
\end{equation}
Given that $(TO(V))(q_2,1)-(TO(V))(q_2,0) \geq 0$, then:
\begin{equation}
(TO(V))(q_1,1)-(TO(V))(q_1,0) \geq 0
\end{equation}
Which leads to: 
\begin{equation}
\underset{s}{\text{argmin}} (TO(V))(q_1,s)=0
\end{equation}i.e. 
\begin{equation}
\phi^*(q_1) \leq \phi^*(q_2) 
\end{equation}
If $\phi^*(q_2)=\underset{s}{\text{argmin}} (TO(V))(q_2,s)=1$ , obviously we have that:
\begin{equation}
\phi^*(q_1) \leq \phi^*(q_2)
\end{equation}
Therefore, we can conclude that the optimal solution is monotone and increasing in $q$, which implies that it is a threshold policy.
\end{proof}
\section{Stationary distribution}\label{sec:steady_state}
We have seen previously that the optimal solution of problem ~\eqref{eq:relaxed} is a threshold-based policy. Let us define $n_k$ as the threshold for users in class $k$, i.e. if the queue state of user $i$ in class $k$ is $q_i^k$ such that $q_i^k \leq n_k$ then the user will not be scheduled, and else, the user will be selected for transmission. The objective of this section is to derive the stationary distribution of the users' states. This will be useful in the subsequent section in the derivation of a closed form expression of the Whittle index values. We assume here that at each queue $i$ in class~$k$, packets arrive according to a discrete uniform distribution, that is, $\mathbb{P}(A_{i}^k(t)=x)=\rho_k$ for all $0\leq x\leq R_k-1$ and $0$ otherwise, where $\rho_k=1/R_k$.\\
For ease of notation, we again drop the indices $k$ and $i$  (e.g. we denote the threshold by $n$ and the queue length by $q$). We denote by $p^n(i,j)$ the transition probability from state $i$ to $j$, by $u$ the stationary distribution under the threshold policy $n$, and by $R$ the maximum rate ($\rho=1/R$). One can notice that $u$ verifies the full balance equation, i.e.:
\begin{equation} u(i)= \sum_{j=0}^L p^n(j,i)u(j) = \sum_{j=0}^n p^n(j,i)u(j)+\sum_{j=n+1}^L p^n(j,i)u(j)\end{equation} 
\begin{mydef}\label{pi}
We define $\pi_i$ as:
\begin{equation}
 \pi_i=\left\{
    \begin{array}{ll}
        \rho & if \ 0\leq i \leq R-1 \\
        0 & else
    \end{array}
\right.
\end{equation}
\end{mydef}
\begin{proposition}\label{prop:general_passage}
The expressions of $p^n(j,i)$ are given by:\\
if $0 \leq i < L$ and $j \leq n$\\
\begin{equation}
p^n(j,i) = \pi_{i-j}=\left\{
    \begin{array}{ll}
        \rho & if \ 0\leq i-j \leq R-1 \\
        0 & else
    \end{array}
\right.
\end{equation}
if $0 \leq i < L$ and $n<j\leq L$\\
\begin{equation}
p^n(j,i) = \pi_{i-(j-R)^+}=\left\{
    \begin{array}{ll}
        \rho & if \ 0\leq i-(j-R)^+ \leq R-1 \\
        0 & else
    \end{array}
\right.
\end{equation}
if $i=L$ and $j \leq n$\\
%\begin{equation}
%p^n(j,L)=\left\{
%    \begin{array}{ll}
%        (R-L+j)\rho & if \ 0\leq L-j \leq R-1 \\
%        0 & else
%    \end{array}
%\right.
%\end{equation}

%%this is more clear than the previous one
\begin{equation}
p^n(j,L)=(R-L+j)\pi_{L-j}=\left\{
    \begin{array}{ll}
        (R-L+j)\rho & if \ 0\leq L-j \leq R-1 \\
        0 & else
    \end{array}
\right.
\end{equation}

if $i=L$ and $n<j\leq L$\\
\begin{equation}
p^n(j,L) = (R-L+(j-R)^+)\pi_{L-(j-R)^+}=\left\{
    \begin{array}{ll}
        (R-L+(j-R)^+)\rho & if \ 0\leq L-(j-R)^+ \leq R-1 \\
        0 & else
    \end{array}
\right.
\end{equation}
\end{proposition}
\begin{proof}
See appendix \ref{app:general_passage}.       
\end{proof}

\begin{proposition}\label{eq:stat.dist}

The expressions of the stationary distribution is:
\begin{itemize}
\item $ L < R$:\\
1) $-1 \leq n \leq L-1$:
\begin{equation}
u(i)=\left\{
    \begin{array}{ll}
        \rho_k(1-\rho_k)^{n-i} & if  \ 0 \leq i \leq n\\
        \rho_k & if  \ n+1 \leq i \leq L-1 \\
        (1-\rho_k)^{n+1}-(L-n-1)\rho_k & if \ i=L
    \end{array}
\right.
\end{equation} 
2) $n=L$:\\
\begin{equation}
u(i)=\left\{
    \begin{array}{ll}
        0 & if  \ 0 \leq i \leq L-1\\
        1 & if  \ i=L
    \end{array}
\right.
\end{equation}

\item $R \leq L < 2R$:\\
1)$-1 \leq n \leq L-R-1$:
\begin{equation}
u(i)=\left\{
    \begin{array}{ll}
        \rho - (n-i)\rho^2 & if  \ 0 \leq i \leq n\\
        \rho & if  \ n+1 \leq i \leq R-1 \\
        \rho-(i-n)\rho^2 & if \ R \leq i \leq n+R
    \end{array}
\right.
\end{equation}
2)$L-R \leq n < R$:
\begin{equation}
u(i)=\left\{
    \begin{array}{ll}
        \rho-\rho^2(n-i) & if  \ 0 \leq i \leq L-R-1\\
        (1-\rho)^{n-i}\rho & if  \ L-R \leq i \leq n\\
         \rho & if \ n+1 \leq i \leq R-1\\
         \rho-\rho^2(i-n) & if  \ R \leq i \leq L-1\\
         (1-\rho)^{n-L+R+1}-\rho(L-1-n) & if \ i=L
    \end{array}
\right.
\end{equation}
3) $R \leq n \leq L-1$:
\begin{equation}
u(i)=\left\{
    \begin{array}{ll}
        \rho-\rho^2(n-i) & if  \ n-R+1 \leq i \leq L-R-1\\
        (1-\rho)^{n-i}\rho & if  \ L-R \leq i \leq n\\
         \rho - \rho^2(i-n) & if \ n+1 \leq i \leq L-1\\
         (1-\rho)^{n-L+R+1}-\rho(L-1-n) & if \ i=L
    \end{array}
\right.
\end{equation}
4) $n=L$
\begin{equation}
u(i)=\left\{
    \begin{array}{ll}
        0 & if  \ 0 \leq i \leq L-1\\
        1 & if  \ i=L
    \end{array}
\right.
\end{equation}
  
\item $L \geq 2R$ 

1) $-1 \leq n < R$:
\begin{equation}
u(i)=\left\{
    \begin{array}{ll}
        \rho - (n-i)\rho^2 & if  \ 0 \leq i \leq n\\
        \rho & if  \ n+1 \leq i \leq R-1 \\
        \rho-(i-n)\rho^2 & if \ R \leq i \leq n+R
    \end{array}
\right.
\end{equation}

2) $R \leq n < L-R$:
\begin{equation}
u(i)=\left\{
    \begin{array}{ll}
        \rho - (n-i)\rho^2 & if  \ n-R+1 \leq i \leq n\\
        \rho-(i-n)\rho^2 & if  \ n \leq i \leq n+R-1 
    \end{array}
\right.
\end{equation}

%3) $L-R \leq n < L$:\\ 
%\begin{flalign*}
%n-R+1 \leq i \leq L-R-1: u(i)&=\rho^2(R-n+i) \\
%L-R \leq i \leq n:u(i)&=(1-\rho)^{n-i}\rho \\
%n+1 \leq i \leq L-1: u(i)&=\rho - \rho^2(i-n)\\ 
%i=L: u(i)&=\rho^2 \sum_0^{n-L+R-1} (R-L-k+n)(1-\rho)^k\\ 
%&=(1-\rho)^{n-L+R+1}-\rho(L-1-n)
%\end{flalign*}
3) $L-R \leq n < L$:
\begin{equation}
u(i)=\left\{
    \begin{array}{ll}
        \rho-\rho^2(n-i) & if  \ n-R+1 \leq i \leq L-R-1\\
        (1-\rho)^{n-i}\rho & if  \ L-R \leq i \leq n\\
         \rho - \rho^2(i-n) & if \ n+1 \leq i \leq L-1\\
         (1-\rho)^{n-L+R+1}-\rho(L-1-n) & if \ i=L
    \end{array}
\right.
\end{equation}

4) $n=L$:\\ 
\begin{equation}
u(i)=\left\{
    \begin{array}{ll}
        0 & if  \ 0 \leq i \leq L-1\\
        1 & if  \ i=L
    \end{array}
\right.
\end{equation}
\end{itemize}
\end{proposition}

\begin{proof}
See appendix \ref{app:eq:stat.dist}.
\end{proof}

\section{Whittle's index}\label{sec:WI} 
In this section, we provide the derivation of the Whittle indices, which are values that depend on the queue state of the user and its maximum rate. Although this derivation is made using the relaxed problem, it allows us to develop a heuristic for the original problem. It is worth mentioning that the Whittle's index at given state, say $n$, represents the Lagrange multiplier for which the optimal decision of the individual dual relaxed problem at this state is indifferent (passive and active decision are both optimal). However, the Whittle index is well defined only if the property of indexability is satisfied. 
This property requires to establish that as the Lagrange multiplier (or equivalently the subsidy for passivity W) increases, 
the collection of states in which the optimal action is passive increases.
In this section, we work on a given class $k$, and we consider its maximum transmission rate is $R$ with $\rho=1/R$. All the obtained results here can be applied for any class. %We drop the user and class index for ease of notation. \\ 
We start the derivation by first reformulating the dual of the relaxed problem using the stationary distribution derived in the previous section. Since the solution of the dual of the relaxed problem~\eqref{eq:relaxed} (given a constant $W$) is a threshold-based policy, we can reformulate the problem as follows:
\begin{equation}\label{eq:new_prob_form_lagr_function_steady_state}
 \underset{n \in [0,L]}{\text{min}} \mathbb{E} [aq^n+Ws^n]=\underset{n \in [0,L]}{\text{min}} \{\sum_{q=0}^L au^n(q)q \ - \ W \sum_{q=0}^n u^n(q)\} 
\end{equation}
with $n$ and $u^n$ being the threshold and the stationary distribution under the threshold policy $n$.\\
The new formulation of the problem turns out to be useful to derive the Whittle indices since, for any $W$, we can find the minimizer of the expression in equation~\eqref{eq:new_prob_form_lagr_function_steady_state}.\\
We first give the expression of the mean cost in equation~\eqref{eq:new_prob_form_lagr_function_steady_state} given threshold $n$ (for all possible values of $n$ and $L$).\\

\begin{itemize}

\item $L < R$:\\
if $-1 \leq n \leq L-1$: \begin{equation} \sum_{q=0}^L au^n(q)q=a[(L+R)(1-\rho)^{n+1}+n-R+1+\frac{(L-1-n)(n-L)}{2R}]\end{equation} 
if $n=L$: \begin{equation} \sum_{q=0}^L au^n(q)q=aL\end{equation} 
\item $R \leq L < 2R$:\\
if $-1 \leq n \leq L-R-1$:  \begin{equation} \sum_{q=0}^L au^n(q)q=a[\frac{R-1}{2} + \frac{n(n+1)}{2R}]\end{equation} 
if $L-R \leq n \leq R-1$: \begin{equation} \sum_{q=0}^L au^n(q)q=2aR(1-\rho)^{n-L+R+1}-a[\frac{n(n+1)}{2R}+\frac{R-1}{2}+\frac{L}{R}(L-2n-1)]\end{equation}     
if $R \leq n \leq L-1$:  \begin{equation} \sum_{q=0}^L au^n(q)q=a[n+1-R+ 2R(1-\rho)^{n-L+R+1}+\rho(L-1-n)(n-L)]\end{equation} 
if $n=L$:  \begin{equation} \sum_{q=0}^L au^n(q)q=aL\end{equation} 
\item $L \geq 2R$:\\
if $-1 \leq n \leq R-1$:  \begin{equation} \sum_{q=0}^L au^n(q)q=a[\frac{R-1}{2} + \frac{n(n+1)}{2R}]\end{equation} 
if $R \leq n \leq L-R$: \begin{equation} \sum_{q=0}^L au^n(q)q=an\end{equation}     
if $L-R+1 \leq n \leq L-1$:  
\begin{align}
\sum_{q=0}^L au^n(q)q= a[n+1-R+ 2R(1-\rho)^{n-L+R+1}+\rho(L-1-n)(n-L)]
\end{align} 
if $n=L$:  \begin{equation} \sum_{q=0}^L au^n(q)q=aL\end{equation} 
\end{itemize}
Second, we provide the expression of the passive decision's average time in equation~\eqref{eq:new_prob_form_lagr_function_steady_state} given a threshold $n$:\\

\begin{itemize}
\item $L < R$:\\
if $-1 \leq n \leq L-1$: \begin{equation} \sum_{q=0}^n u^n(q)=1-(1-\rho)^{n+1}\end{equation} 
if $n=L$: \begin{equation} \sum_{q=0}^n u^n(q)=1\end{equation} 
\item $R \leq L < 2R$:\\
if $-1 \leq n \leq L-R-1$:  \begin{equation} \sum_{q=0}^n u^n(q)=(1-\frac{n}{2R})(\frac{n+1}{R})\end{equation} 
if $L-R \leq n \leq R-1$: \begin{equation} \sum_{q=0}^n u^n(q)=L\frac{\rho^2}{2}(L-1-2n)(L-n) + \frac{1+\rho}{2}+\rho n-(1-\rho)^{n-L+R+1}\end{equation}     
if $R \leq n \leq L-1$:  \begin{equation} \sum_{q=0}^n u^n(q)= \frac{\rho^2}{2}(L-1-n)(L-n) + 1 -(1-\rho)^{n-L+R+1}\end{equation} 
if $n=L$: \begin{equation} \sum_{q=0}^n u^n(q)=1\end{equation}  
\item $L \geq 2R$:\\
if $-1 \leq n \leq R-1$:  \begin{equation} \sum_{q=0}^n u^n(q)=(1-\frac{n}{2R})(\frac{n+1}{R})
\end{equation} 
if $R \leq n \leq L-R$: \begin{equation} \sum_{q=0}^n u^n(q)=\frac{1}{2}+\frac{1}{2R}
\end{equation}     
if $L-R+1 \leq n \leq L-1$:  
\begin{equation}
\sum_{q=0}^n u^n(q)= \frac{\rho^2}{2}(L-1-n)(L-n)+ 1 -(1-\rho)^{n-L+R+1}
\end{equation} 
if $n=L$: \begin{equation} \sum_{q=0}^n u^n(q)=1
\end{equation}  
\end{itemize}
\subsection{Computation of the Whittle index values}
We first formalize the indexability and the Whittle's index in the following definitions.
\begin{mydef}
Considering problem~\eqref{eq:new_prob_form_lagr_function_steady_state} for a given $W$, we define $D(W)$ as the set of states in which the optimal action (with respect to the optimal solution of Problem~\eqref{eq:new_prob_form_lagr_function_steady_state}) is the passive one. In other words, $n \in D(W)$ if and only if the optimal action at state $n$ is the passive one.
\end{mydef}
$D(W)$ is well defined as the optimal solution of Problem~\eqref{eq:new_prob_form_lagr_function_steady_state} is a stationary policy, more precisely, a threshold based policy.
\begin{mydef}
A class is indexable if the set of states in which the passive action is the optimal action increases in $W$, that is, $W' < W \Rightarrow D(W') \subseteq D(W)$.
When the class is indexable, the Whittle's index in state $n$ is defined as: \begin{equation} W(n)=\min \{W |n \in D(W)\}\end{equation} 
\end{mydef}
In the literature, several works have been conducted to find the Whittle index values. For example, an interesting iterative algorithm has been provided in  \cite{larranaga2015dynamic}. Even though the context of our work here is different from the one considered in \cite{larranaga2015dynamic}, we will prove in the sequel  that the proposed  algorithm  in \cite{larranaga2015dynamic} can be adapted to our case up to some modifications (e.g. in our case we have a maximum buffer state L, etc.). In addition, further analysis will be provided here to derive a closed form expression of the Whittle index values. We will first provide this modified algorithm and then prove that it allows the computation of the Whittle's index values for our problem.
\\
%
%Algorithm 1:\\
%Step 0:\\
%\begin{equation} W_0=\underset{n \in \mathbb{N}}{\text{inf}}\frac{\sum_{q=0}^L au^n(q)q-\sum_{q=0}^{L} au^{-1}(q)q}{\sum_{q=0}^n u^n(q)}\end{equation} 
%We call $n_0$ the largest minimizer of this expression and then we define $W(k)=W_0$ for all $k \leq n_0$.\\
%Step $j$: We compute:\\
%\begin{equation} W_j=\underset{n \in \mathbb{N}\setminus \{n: \ \sum_{q=0}^n u^n(q)=\sum_{q=0}^{n_{j-1}} u^{n_{j-1}}(q)\} \cup \{0,\cdots,n_{j-1}\}}{\text{inf}}\frac{\sum_{q=0}^L au^n(q)q-\sum_{q=0}^{L} au^{n_{j-1}}(q)}{\sum_{q=0}^n u^n(q)-\sum_{q=0}^{n_{j-1}} u^{n_{j-1}}(q)}\end{equation} 
%We denote $n_j$ the largest minimizer, and we define $W(k)=W_j$, for all $n_{j-1}<k \leq n_j$, we go to step $j+1$. We stop when $n_j=L$. \\
%The Whittle index of state $k$ is given by $W(k)$.
\begin{algorithm}
\caption{Whittle Index Computation}\label{euclid}
\begin{algorithmic}[1]
\State \textbf{Init.} Let $j$ be initialized to $0$
\State \textbf{Find} $W_0=\underset{n \in \mathbb{N}}{\text{inf}}\frac{\sum_{q=0}^L au^n(q)q-\sum_{q=0}^{L} au^{-1}(q)q}{\sum_{q=0}^n u^n(q)}$
\State \textbf{Define} $n_0$ as the largest minimizer of the above expression
\State \textbf{Let} $W(k)=W_0$ for all $k \leq n_0$
\While {$n_j\neq L$}
\State $j=j+1$
\State \textbf{Define} $M_j$ the set $\{n:\sum_{q=0}^n u^n(q)=\sum_{q=0}^{n_{j-1}} u^{n_{j-1}}(q)\}\cup \{0,\cdots,n_{j-1}\}$
\State \textbf{Find} $W_j=\underset{n \in \mathbb{N}\setminus M_j} {\text{inf}}\frac{\sum_{q=0}^L au^n(q)q-\sum_{q=0}^{L} au^{n_{j-1}}(q)}{\sum_{q=0}^n u^n(q)-\sum_{q=0}^{n_{j-1}} u^{n_{j-1}}(q)}$
\State \textbf{Define} $n_j$ as the largest minimizer of the above expression
\State \textbf{Let} $W(k)=W_j$ for all $n_{j-1}<k \leq n_j$
\EndWhile
\State \textbf{Output} The Whittle index of state $k$ which is given by $W(k)$
\end{algorithmic}
\end{algorithm}
\begin{proposition}\label{prop_wi}
Assuming that the optimal solution is a threshold policy, and that $\sum_{q=0}^n u^n(q)$ is increasing, then the class is indexable. Moreover, if $\sum_{q=0}^L au^n(q)q$ is increasing in $n$ and for all $i$ and $j$ such that $i<j$ \ $\sum_{q=0}^i u^i(q)=\sum_{q=0}^j u^j(q) \Longrightarrow \sum_{q=0}^L au^i(q)q < \sum_{q=0}^L au^j(q)q$, then the Whittle’s index values are computed by applying Algorithm 1.\\
\end{proposition}
\begin{proof}
For the proof, see appendix \ref{app:prop_wi}.        
\end{proof}

\begin{remark}
In order to simplify the notation in the sequel, we denote $\sum_{q=0}^L au^n(q)q$ by $a_n$ and $\sum_{q=0}^{n} u^{n}(q)$ by $b_n$. 
\end{remark}
In order to apply Algorithm 1 that allows to obtain the Whittle's index for each state in our case, we need to prove that the conditions given in Proposition \ref{prop_wi} are satisfied.
%We distinguish between the tree cases of $L$ as we have mentioned in \ref{sec:steady_state}
We focus only on the third case of $L$ ($L \geq 2R$) since it is more realistic as the maximum buffer length $L$ is often much higher than the transmission rate $R_k$. Nevertheless, the analysis in this paper can be easily extended to the case where $L <2R$.
To that end, we will be limited to give only the Whittle index expressions when $L<2R$ as well as a concise proof in the end of this section.
\begin{theorem} \label{indexability}
For each $k$, the class-k is indexable.
\end{theorem}
\begin{proof}
According to Proposition \ref{prop_wi}, we just need to prove that $\sum_{q=0}^n u^n(q)$ is increasing $n$. The proof is based on the two following two lemmas.
\begin{Lemma}\label{temps_pass_str_croiss_less_R}
$\sum_{q=0}^n u^n(q)$ is strictly increasing in $[-1,R-1]$
\end{Lemma}
 \begin{proof}
\renewcommand{\qedsymbol}{$\blacksquare$}
See appendix \ref{app:temps_pass_str_croiss_less_R}.
\end{proof}
\begin{Lemma} \label{temps_pass_str_croiss}
$\sum_{q=0}^n u^n(q)$ is strictly increasing in $n \in [L-R+1,L-1]$
\end{Lemma}
\begin{proof}
\renewcommand{\qedsymbol}{$\blacksquare$}
See appendix \ref{app:temps_pass_str_croiss}
\end{proof}
We have that for any $n \in [R,L-R]$:
\begin{equation}
\sum_{q=0}^{n} u^{n}(q)=\sum_{q=0}^{R-1} u^{R-1}(q) =\sum_{q=0}^{L-R+1} u^{L-R+1}(q)=\frac{1}{2}+\frac{1}{2R}
\end{equation}
Therefore: 
\begin{equation}
\sum_{q=0}^{R-1} u^{R-1}(q) \leq \sum_{q=0}^{n} u^{n}(q) \leq \sum_{q=0}^{L-R+1} u^{L-R+1}(q)
\end{equation}
Moreover:
\begin{equation}\sum_{q=0}^{L} u^{L}(q)=1 > 1- (1-\rho)^R=\sum_{q=0}^{L-1} u^{L-1}(q)
\end{equation}\\
Consequently, by combining Lemma \ref{temps_pass_str_croiss_less_R} and Lemma \ref{temps_pass_str_croiss}, we can conclude the indexability of the class as $\sum_{q=0}^{n} u^{n}(q)$ is shown to be increasing in $[0,L]$.
\end{proof}
We prove the two others conditions of Proposition \ref{prop_wi} which are the increasing property of $\sum_{q=0}^L au^n(q)q$ in $n$, and that for all $i$ and $j$ such that $i<j$ $\sum_{q=0}^i u^i(q)=\sum_{q=0}^j u^j(q) \Longrightarrow \sum_{q=0}^L au^i(q)q < \sum_{q=0}^L au^j(q)q$. From the expression of $a_n$ when $n \in [-1,R-1]$, $a_n$ is clearly increasing in $n$. For $n \in [R,L-R-1]$, $a_n$ is strictly increasing and $a_{R-1}=a(R-1) < aR=a_R$, which implies that $a_n$ is increasing in $[-1,L-R-1]$. For $n \in [L-R,L-1]$, we provide the following lemma  
\begin{Lemma}\label{average_cost_strict_increasing_n_greater_L_moins_R}
$\sum_{q=0}^L au^n(q)q$ is strictly increasing in $[L-R,L-1]$.
\end{Lemma}
\begin{proof}
See appendix \ref{app:average_cost_strict_increasing_n_greater_L_moins_R}.  
\end{proof} 
We have that $\sum_{q=0}^L au^{L-R}(q)q=a(L-R) > a (L-R-1)=\sum_{q=0}^L au^{L-R-1}(q)q$, and $\sum_{q=0}^L au^{L-1}(q)q=a L-aR(1-2(1-\rho)^R)< aL=\sum_{q=0}^L au^{L}(q)q$ (because $1-2(1-\rho)^R \geq 1-2 \exp(-1) \geq 0$), then we can conclude that $a_n$ is increasing in $[-1,L]$.  
\\
For the second condition (for all $i$ and $j$ such that $i<j$ $\sum_{q=0}^i u^i(q)=\sum_{q=0}^j u^j(q) \Longrightarrow \sum_{q=0}^L au^i(q)q < \sum_{q=0}^L au^j(q)q$), the only case when $\sum_{q=0}^i u^i(q)$ is equal to $\sum_{q=0}^j u^j(q)$ is when $i$ and $j$ are in the set $[R-1,L-R+1]$. In this set, we have shown that $\sum_{q=0}^L au^n(q)q$ is strictly increasing, then for $i<j$ and $(i,j) \in [R-1,L-R+1]^2$, $\sum_{q=0}^L au^i(q)q < \sum_{q=0}^L au^j(q)q$, hence the two conditions are satisfied. 
\\\\ 
As the indexability is satisfied and the two conditions of Proposition \ref{prop_wi} are verified, then we can apply Algorithm 1 to get the Whittle's index for each state. However, the complexity of this algorithm is $L^2$, where $L$ is the maximum buffer length which could be large in practice. In order to overcome this complexity issue, we will provide further analysis and derive simple expressions of the Whittle indices.\\ 
We first proceed by laying out the following definitions and lemmas.
\begin{mydef}
For any given increasing threshold policy $n$, we define $y^n$ as a function of the subsidy $W$, such that $y^n(W)=\sum_{q=0}^L au^n(q)q \ - \ W \sum_{q=0}^n u^n(q)=a_n - W b_n$.  
\end{mydef}
\begin{Lemma}\label{intersection_point_Whittle_index}
For any state $(i,j) \in [-1,L]^2$, the intersection point's abscess between $y^i(W)$ and $y^j(W)$ denoted by $x_{i,j}$ is:
\begin{equation} \frac{\sum_{q=0}^L au^i(q)q-\sum_{q=0}^{L} au^{j}(q)}{\sum_{q=0}^i u^i(q)-\sum_{q=0}^{j} u^{j}(q)}
\end{equation}     
\end{Lemma}
\begin{proof}
See appendix \ref{app:intersection_point_Whittle_index}. 
\end{proof}
\begin{mydef}
We define for $0 \leq n \leq R$, $w_n=x_{n,n-1}=\frac{\sum_{q=0}^L au^n(q)q-\sum_{q=0}^{L} au^{n-1}(q)}{\sum_{q=0}^n u^n(q)-\sum_{q=0}^{n-1} u^{n-1}(q)}=\frac{a_n-a_{n-1}}{b_n-b_{n-1}}=\frac{aRn}{R-n}$
(by replacing $a_n$ and $b_n$ by their expressions when $0 \leq n \leq R$).
\end{mydef}
\begin{mydef}
We define a function $f$, such that for each $n \in [0,R]$, $f(n)=w_n [\sum_{q=0}^L u^L(q)-\sum_{q=0}^n u^n(q)]+\sum_{q=0}^L au^n(q)q=w_n [1-(1-\frac{n}{2R})\frac{n+1}{R}]+a(\frac{R-1}{2}+\frac{n(n+1)}{2R})$, for $n=R$, $f(R)=+\infty$, and for $n=-1$, $f(-1)=0$. In other words, $f(n)/a$ can be interpreted as the value of $L$ such that $w_n=x_{L,n}$.  
\end{mydef}
\begin{Lemma}\label{f_strictly_increasing}
$f$ is strictly increasing in $n$, for $n \in [0,R]$.
\end{Lemma}
\begin{proof}
See appendix \ref{app:f_strictly_increasing}.
\end{proof}

\begin{Lemma}\label{lem:d_existence}
Assuming that $L \geq 2R$, then there exists an integer $d \in [0,R-1]$ such that $\frac{f(d)}{a} < L \leq \frac{f(d+1)}{a}$
\end{Lemma}
\begin{proof}
We have $f(0)/a=\frac{R-1}{2}$, and $f(R)/a=+\infty$. Hence, as $f(.)$ is strictly increasing in $n$, and $f(0)/a=\frac{R-1}{2} < 2R \leq L \leq f(R)/a=+\infty$, there exists one and only one $d \in [0,R-1]$ that satisfies $\frac{f(d)}{a} < L \leq \frac{f(d+1)}{a}$.
That completes the proof. 
\end{proof}

Therefore, according to the definition of $f$, $d$ satisfies $x_{d,d-1} \leq x_{L,d}$ and $x_{L,d+1} \leq x_{d+1,d}$.
\iffalse
\begin{Lemma}\label{e}
We consider $e=\frac{5}{2}R- \frac{\sqrt{13R^2-4R}}{2}$.
Assuming that $L \geq 2R$, then for all $n < \min(e,R)$, $L > \frac{f(n)}{a}$, i.e. $w_n < x_{L,n}$.
\end{Lemma}
\begin{proof}
See appendix \ref{app:e}.      
\end{proof}
\fi
\begin{theorem}\label{w.i}
The Whittle's index expressions are:\\
for $0 \leq n \leq d$: $W(n)=w_n=x_{n,n-1}=\frac{aRn}{R-n}$\\
for $d < n \leq L$: $W(n)=x_{L,d}=\frac{a[L-(\frac{R-1}{2}) + \frac{d(d+1)}{2R}]}{1-(1-\frac{d}{2R})(\frac{d+1}{R})}$
\end{theorem}
\begin{proof}    
To prove Theorem \ref{w.i}, according to Proposition \ref{prop_wi}, we have to prove that, from $0 \leq j \leq d$, the largest minimizer at step $j$ is $j$ and at step $d+1$ is $L$. In other words, for all $0 \leq j \leq d$, we have that $\frac{a_j-a_{j-1}}{b_j-b_{j-1}}<\frac{a_n-a_{j-1}}{b_n-b_{j-1}}$ for all $n > n_{j-1}+1=j$ such that $b_n \neq b_{j-1}$ and $\frac{a_L-a_{d}}{b_L-b_{d}} \leq \frac{a_n-a_{d}}{b_n-b_{d}}$ for all $n \geq n_{d}+1=d+1$ such that $b_n \neq b_d$ , with $n_j$ being the largest minimizer at step $j$.\\
To that extent, it turns out to be relevant to demonstrate that $x_{L,R} \leq x_{n,R}$ for $L-R+1 < n \leq L-1$. For the detailed proof, see Appendix \ref{app:w.i}. 
\end{proof}

\subsection{$L<R$}
The indexability property can be easily set up by observing that $\sum_{i=0}^n u^n(i)$ is strictly increasing in $n$. Furthermore, we have that $\sum_{i=0}^L au^n(i)i$ is increasing  in $n$. Thereby we can apply the algorithm 1 to compute the Whittle index expressions. According to \cite[Corollary~2.1]{larranaga2015dynamic}, if $x_{n,n-1}$ is increasing in $n$, then the Whittle index of state $n$ is $x_{n,n-1}$. Effectively, for $L<R$, $x_{n,n-1}$ is increasing in $n$ and we have the following theorem.
\begin{theorem}
Denoting $1-L^2\rho/2-\rho/2L+L\rho-1/\rho$ by $b$.\\
The Whittle index of state $n$:\\
For $n \in [0,L-1]$: $W(n)=x_{n,n-1}=\frac{-\rho(L+R)(1-\rho)^n+(-2n+1)\rho/2+1+L\rho-\rho/2}{\rho(1-\rho)^n}$\\
For $n=L$: $W(L)=x_{L,L-1}=\frac{L-[(L+R)(1-\rho)^L-(L-1)^2\rho/2+(L-1)(1+L\rho-\rho/2)+b]}{(1-\rho)^L}$
\end{theorem}    
\subsubsection{$R \leq L < 2R$}
Regarding the case where $R \leq L < 2R$, the class is indexable since $\sum_{i=0}^n u^n(i)$ is strictly increasing in $n$. Similar to the other cases, the algorithm 1 can be applied to obtain the expression of the Whittle index for different states. Following the same methodology in appendix \ref{app:w.i}, we obtain the Whittle index expression as follows: 
\begin{theorem}
It exists $d$ such that $x_{d,d-1} \leq x_{L,d} \leq x_{d+1,d}$ and $d < L-R$, where the Whittle index expressions are given by:
for $0 \leq n \leq d$: $W(n)=w_n=x_{n,n-1}=\frac{aRn}{R-n}$\\
for $d < n \leq L$: $W(n)=x_{L,d}=\frac{a[L-(\frac{R-1}{2}) + \frac{d(d+1)}{2R}]}{1-(1-\frac{d}{2R}(\frac{d+1}{R})}$
\end{theorem}

\subsection{Whittle index policy for the original problem}
We now consider the original optimization problem~\eqref{eq:constraint_original} and propose a simple Whittle index policy. This policy consists of simply allocating the channels to the $M$ users that have the highest Whittle index at time $t$, denoted by $WI$, and computed using the simple expressions in Theorem \ref{w.i}.\\

\begin{figure}
\centering
\includegraphics[width=0.5\textwidth]{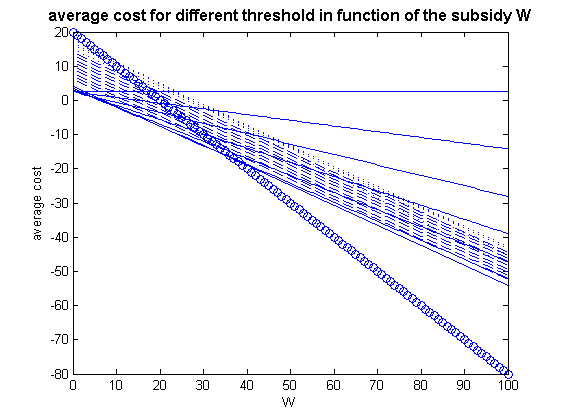}
\caption{Illustration of the function $a_n-Wb_n$ for different value of $n$}
\label{Whittle-figure}
\end{figure}
In Figure \ref{Whittle-figure}, we consider $L>2R$. The straight lines are for $n \leq R-1$, the dashed ones are for $R \leq n \leq L-R$, the doted ones are for $L-R+1 \leq n \leq L-1$, and the line with rounds is for $n=L$. As one can see, the slope of this latter line is very high if we compare it with the other curves. This means that all the intersection points between the round line and straight lines are surely smaller than all the intersection points between the doted and the straight ones, which confirms our Whittle index expressions. From now on, we consider that $L>2R$, furthermore, we suppose the following assumptions.

\begin{assumption}\label{assump:buffer_size_2}
The buffer length $L$ satisfies:
\begin{equation}L > \underset{(i,j) \in [1,K]^2}{\text{max}} \{\frac{a_j}{a_i}\} \underset{k \in [1,K]}{\text{max}}\{\frac{(R_k-1)^2}{2}\}\end{equation}
\end{assumption}
\begin{assumption}\label{assump:condition_on_alpha}
The proportion of queues scheduled at each time, $\alpha=M/N$ satisfies:
\begin{equation}\alpha \geq \frac{1}{2}-\sum_{k=1}^K \frac{\gamma_k}{2R_k}\end{equation}
\end{assumption} 
We justify in the next sections the reasons behind introducing these two assumptions.
\section{Further analysis of the optimal solution of the relaxed problem}\label{sec:relaxed}   
In this section, we provide further analysis and give the structure of the optimal solution for the relaxed problem, which will be useful for the proof of optimality of the Whittle's Index policy.
As we have seen in section \ref{sec:threshold policy}, for any given $W$, the optimal solution for the dual relaxed problem~\eqref{eq:relaxed} is a threshold-based policy for each user. By using the Whittle index expressions defined in  section \ref{sec:WI}, we will provide a derivation of the optimal threshold for each class as function of the Lagrange parameter $W$. 
In this section, we denote by $W^k_i$ the Whittle index at state $i$ in class $k$ (the user and class indices cannot be dropped here as in the previous sections). We denote by $l=(l_1,l_2,\cdots ,l_K)$ the vector which represents the set of thresholds for each class $k$.
As $f(R_k-1)/a=\frac{(R_k-1)^2}{2}$, then considering the assumption \ref{assump:buffer_size_2}, we have that for all $k$, $L > \underset{(i,j) \in [1,K]}{\text{max}} \{\frac{a_j}{a_i}\} \underset{k}{\text{max}} \frac{f(R_k-1)}{a} \geq \underset{k}{\text{max}} \frac{f(R_k-1)}{a}$. That means for each class $k$, the integer $d_k$ (which depends on the maximum rate $R_k$), defined in Lemma \ref{lem:d_existence}, is equal to $R_k-1$. This allows us to obtain a general expression of the Whittle index for all class $k$.  
%and that ($W \leq {W^k_L}$)\footnote{In fact, if there exists $k$ such that $W > W^k_L$, then for all queues belonging to this class and for all states the optimal decision is a passive action. This is not realistic if we take into account the fairness, since all these queues can never transmit.}.
We denote by $u^n_k$, the stationary distribution for class $k$ under threshold policy $n$.
\begin{proposition}\label{prop:optimal_solution_for_dual_relaxed_problem} 
For a given $W$, the optimal threshold vector $l=(l_1(W),l_2(W),\cdots ,l_K(W))$ for the dual problem satisfies:\\
For each $k$:
\begin{equation}
l_k(W)= \underset{i}{\text{max}} \{\arg \underset{i}{\text{max}} \{W^k_i | W^k_i \leq W \}\}
\end{equation} 
or 
\begin{equation}
l_k(W)= \underset{i}{\text{max}} \{\arg \underset{i}{\text{max}} \{W^k_i | W^k_i < W \}\}
\end{equation}
In other words, $l_k$ is the biggest index among the ones that give the biggest Whittle index less than $W$, or strictly less than $W$. We note that the solution can also be a linear combination between the threshold policies  $\underset{i}{\text{max}} \{\arg \underset{i}{\text{max}} \{W^k_i | W^k_i \leq W \}\}$ and $\underset{i}{\text{max}} \{\arg \underset{i}{\text{max}} \{W^k_i | W^k_i < W \}\}$. 

\end{proposition}
\begin{proof}
See appendix \ref{app:optimal_solution_for_dual_relaxed_problem}.     
\end{proof}
Now, we give the structure of the optimal solution of the constrained relaxed problem.
\begin{proposition}\label{prop:optimal_solution_for_constrained_relaxed_problem}
The solution of the constrained relaxed problem is of type threshold policy $l(W^*)$, with $l$ being the function vector defined in Proposition \ref{prop:optimal_solution_for_dual_relaxed_problem} and  $W^*$ satisfies $\alpha =\sum_{k=1}^K \gamma_k \sum_{i=l_k(W^*)+1}^L u_k^{l_k(W^*)}(i)$.    
\end{proposition}
\begin{proof}
See appendix \ref{app:optimal_solution_for_constrained_relaxed_problem}.
\end{proof}      
However, $W^*$ that satisfies the above constraint may not exist since $\alpha$ is a real number that can take any value in $[0,1]$, and $\sum_{k=1}^K \gamma_k \sum_{i=l_{k+1}(W)}^L u^{l_k(W)}_k(i)$ is discrete, since the vector $l(W)$ can only take discrete values in $[0,L]^K$. To deal with this issue, we   
use the fact that for some values of $W$, the optimal solution of the dual problem can be a linear combination or more precisely a randomized policy between two threshold policies for a given class as it has been mentioned in Proposition \ref{prop:optimal_solution_for_dual_relaxed_problem}. To that extent, our task is to find among these values of $W$, the one for which there exists a randomized parameter $\theta$ such that the constraint is satisfied with equality.       
To that end, we introduce this following proposition.
\begin{proposition}\label{prop:optimal_sol_relax_prob_charchterisation}
Under assumption \ref{assump:buffer_size_2} and \ref{assump:condition_on_alpha}, there exists a class $m$, state $p$, and a randomization parameter $\theta$ such that the optimal solution of the dual problem when the langrangian parameter $W=W^m_p$ is characterized by:
\begin{itemize}
\item For $k \neq m$, the optimal threshold is $l_k(W_p^m)= \underset{i}{\text{max}} \{\arg \underset{i}{\text{max}} \{W^k_i | W^k_i \leq W^m_p \}\}$
\item For $k=m$, the optimal solution is randomized policy between two threshold policies $l_m(W_p^m)= \underset{i}{\text{max}} \{\arg \underset{i}{\text{max}} \{W^m_i | W^m_i \leq W_p^m \}\}$ and $l_m(W_p^m)-1= \underset{i}{\text{max}} \{\arg \underset{i}{\text{max}} \{W^m_i | W^m_i < W_p^m \}\}$, where the factor of randomization $\theta$ is the probability of adopting the policy $l_k(W_p^m)$ and $1-\theta$ the probability of adopting the policy $l_k(W_p^m)-1$.
\item The constraint \eqref{eq:constraint_relaxed} is satisfied with equality, i.e. 
\begin{align*}
\alpha= \sum_{k\neq m} \sum_{i=l_k(W_p^m)+1}^L \gamma_k u^{l_k(W_p^m)}_k(i)+\sum_{i=l_m(W_p^m)+1}^L \gamma_m u^*_m(i)+(1-\theta) \gamma_m u^{l_m(W_p^m)-1}_m(l_m(W_p^m))
\end{align*}
Where $u^*_m=\theta u^{l_m(W_p^m)}+(1-\theta) u^{l_m(W_p^m)-1}$.
\item For all $k$, $l_k(W_p^m) < R_k$.
\end{itemize}
\end{proposition}
\begin{proof}
See appendix \ref{app:prop:optimal_sol_relax_prob_charchterisation}
\end{proof}
The solution of the dual problem described in Proposition \ref{prop:optimal_sol_relax_prob_charchterisation} satisfies the constraint \eqref{eq:constraint_relaxed} with equality, then according to Proposition \ref{prop:optimal_solution_for_constrained_relaxed_problem}, this solution is indeed the optimal solution of the constrained problem. 
%This structure of the optimal policy is very interesting and will be used in the proof of asymptotic optimality. 
In that regard, the optimal cost $C^{RP,N}$ is expressed as following: 
\begin{equation}
C^{RP,N}= \sum_{k\neq m} \sum_{i=0}^L N\gamma_k a_k u^{l_k(W_p^m)}_k(i)i+\sum_{i=0}^L N\gamma_m a_m u^*_m(i)i
\end{equation}
\section{Local optimality}\label{sec:optimality}
In this section, we will show that the performance of the Whittle's Index policy is asymptotically locally optimal. The asymptotic optimality means that for a large number of users $N$ and a large number of channels $M$ ($\alpha=\frac{M}{N}$ is a constant value), the Whittle's Index policy is optimal.
For that we will compare the average cost obtained by the Whittle's Index policy WI with the one obtained for the relaxed problem RP.
Explicitly, denoting by $C_T^N(\textbf{x})$ the average cost obtained over the time duration $0 \leq t \leq T$ under Whittle's Index policy conditioned on the initial state $\textbf{x}$ ,we show that $C^N_T(\textbf{x})$ tends to $C^{RP,N}$ when $N$ scales. The reason behind comparing $C^{RP,N}$ and $C^N_T(\textbf{x})$ is that $C^{RP,N}$ is a lower bound of all expected average cost obtained by any policy that resolves the original Problem \eqref{eq:constraint_original}. This means that it is sufficient to prove that $C^N_T(\textbf{x})$ converges to $C^{RP,N}$ when $T$ and $N$ scale in order to establish the asymptotic optimality of Whittle's Index policy. 
%In fact, if we denote by $C^{WI,N}$ the average cost of the Whittle's Index policy, $C^{OP,N}$ the optimal cost of the original problem with instantaneous inequality and $C^{RP,N}$ the optimal cost of the relaxed problem, the following always holds: $C^{WI,N} \geq C^{OP,N} \geq C^{RP,N}$.\\
%Hence, in order to show the asymptotic optimality, we just need to prove that for large $N$, $C^{WI,N}$ converges to $C^{RP,N}$ (this directly implies that $C^{WI,N}$ will converge to $C^{OP,N}$).\\   
For that, we will be in need of the optimal cost expression of the relaxed problem $C^{RP,N}$ derived in Section \ref{sec:relaxed}.\\ 
First, we denote by $Z_i^{k,N}$ the proportion of queues at state $i$ in class $k$ over all the queues of the system. In other words, it denotes the number of queues at state $i$ in class $k$ over the number of all users which is $N$. We have that $\textbf{Z}^N=(\textbf{Z}^{1,N},.....,\textbf{Z}^{K,N})$ with $\textbf{Z}^{k,N}=(Z_1^{k,N},......,Z^{k,N}_L)$ and $\sum_{i=0}^L Z_i^{k,N}=\gamma_k$ for each class $k$.\\
The expression of $C_T^N(\textbf{x})$ in function of $\textbf{Z}^N$ is $\frac{1}{T} \mathbb{E}\left[ \sum_{t=0}^{T-1}  \sum_{k=1}^{K} \sum_{i=1}^{L} a_kZ_i^{k,N}(t)iN \mid \textbf{Z}^N(0)=\textbf{x}\right]$, where $\textbf{Z}^N(t)$ evolves under Whittle's Index policy.
Denoting by $\textbf{z}^*$ the optimal proportion of the the relaxed problem, we say that the Whittle's Index policy is asymptotically locally optimal if there exists $\delta>0$ such that the initial proportion vector $\textbf{Z}^N(0)$ is within $\Omega_{\delta}(\textbf{z}^*)$ (i.e. $||\textbf{Z}^N(0)-\textbf{z}^*||<\delta$), then $C_T^{N}(\textbf{x})$ converges to $C^{RP,N}$ when $T$ and $N$ scale.\\ 
In order to prove that, we use the fluid limit technique that consists of analyzing the evolution of the expectation of $\textbf{Z}^N(t)$ under the Whittle's Index policy. For that, we define the vector $\textbf{z}(t)$ as follows:
\begin{equation}
\textbf{z}(t+1)-\textbf{z}(t)|_{\textbf{z}(t)=\textbf{z}}=\mathbb{E}\left[\textbf{Z}^N(t+1)-\textbf{Z}^N(t)|\textbf{Z}^N(t)=\textbf{z}\right]
\end{equation}
If we denote by $w_j^h$ the Whittle index for class $h$ at state $j$ and by $p_i^k(\textbf{z})$ the probability that a user is selected randomly among $z_i^k$ to transmit, one can easily show that \cite{weber1990index}:
\begin{equation}
p_i^k(\textbf{z})=\min \{ z_i^k,\max(0,\alpha-\sum_{w_j^h > w_i^k} z_j^h) \} /z_i^k
\end{equation}
We denote by $q_{i,j}^{k,0}$ and $q_{i,j}^{k,1}$ the probability to transition from state $i$ to state $j$ in a class $k$ queue if the queue is not scheduled or is scheduled for transmission respectively.

Then, the probability to transition from state $i$ to state $j$ in class $k$ is:
\begin{equation}
q_{i,j}^k(\textbf{z})=p_i^k(\textbf{z})q_{i,j}^{k,1}+(1-p_i^k(\textbf{z}))q_{i,j}^{k,0}
\end{equation}
Let $w^*$ be the Lagrangian parameter that gives the optimal solution of the relaxed problem. Then, according to Proposition \ref{prop:optimal_sol_relax_prob_charchterisation}, there exists a given class $m$ such that $w_{l_m}^m=w^*$ where the corresponding optimal solution of the relaxed problem is of type threshold policy for class $k \neq m$ denoted $l_k$, and a randomized policy between two threshold policies $l_m$ and $l_m-1$ for class $m$. Moreover, $l_k < R_{k}$ for all $k$.\\ 
We define $\jmath_{w^*}$ as the set of states such that at any system state $\textbf{z} \in \jmath_{w^*}$, if we use the Whittle's Index policy, all users with the Whittle index value higher than $w^*$ are scheduled, the
users with Whittle index value smaller than $w^*$ stay idle and
the users with index value $w^*$ are scheduled with a certain randomization. Specifically, $\jmath_{w^*}=\{ \textbf{z}: \ \sum_{w_i^k>w^*} z_i^k<\alpha, \sum_{w_i^k \geq w^*} z_i^k \geq \alpha \}$.\\
If we start with $\textbf{z}(0)$ in $\jmath_{w^*}$, then:
\begin{equation}\label{eq:relation_z_t+1_t}
z_i^k(t+1)-z_i^k(t)=\sum_{j \neq i} q_{j,i}^k(\textbf{z}(t))z_j^k(t)-\sum_{i\neq j}q_{i,j}^k(\textbf{z}(t))z_i^k(t)
\end{equation} 
%\begin{Lemma}
%for each $n \in [0,R-1]$, the Whittle index computed at state $n$ in classe which use $R$ as maximum rate $W_R(n)$, is decreasing in $R$.
%\end{Lemma}
%\begin{proof}
%\begin{equation}W_{R+1}(n)-W_R(n)=\frac{a(R+1)n}{R+1-n}-\frac{aRn}{R-n}=\frac{-n^2}{(R+1-n)(R-n)} \leq 0\end{equation}
%\end{proof} 
%That means for $R_k < R_m$, the optimal threshold policy $l_k$ will be less than $l_m$, and for $R_k> R_m$, the optimal threshold policy $l_k$ will be greater than $l_m$. 
Moreover, we have the following equality for all $k$ and  $t$: \begin{equation} \sum_{j=0}^L z_j^k(t)=\gamma_k\end{equation} 
and as $\textbf{z}(t) \in \jmath_{w^*}$, we can show the following: \\
1) $k \neq m$: \begin{equation} z_i^k(t+1)=\sum_{j=0}^{l_k-1}(q_{j,i}^{k,0}-q_{l_k,i}^{k,0})z_j^k(t) +\sum_{j=l_k+1}^{L}(q_{j,i}^{k,1}-q_{l_k,i}^{k,0})z_j^k(t)+ \gamma_k q_{l_k,i}^{k,0}\end{equation} 
2) $k=m$
\begin{align}
z_i^m(t+1)=&\sum_{j=0}^{l_m-1}(q_{j,i}^{m,0}-q_{l_m,i}^{m,0})z_j^m(t) +\sum_{j=l_m+1}^{L}(q_{j,i}^{m,1}-q_{l_m,i}^{m,1})z_j^m(t)+(1-\alpha) q_{l_m,i}^{m,0}+ \alpha q_{l_m,i}^{m,1} \nonumber \\
&-(\displaystyle\sum_{\substack{w_j^h>w_{l_m}^m \\ h \neq m, j \neq l_h \phantom{-}}} z_j^h(t)) q_{l_m,i}^{m,1} - (\displaystyle\sum_{\substack{w_j^h \leq w_{l_m}^m \\ h \neq m, j \neq l_h \phantom{-}}} z_j^h(t)) q_{l_m,i}^{m,0}+ (\sum_{\substack{h=1 \\ h \neq m}}^K \sum_{\substack{j=0 \\ j \neq l_h}}^L \mathds{1}_{\{w_{l_h}^h>w_{l_m}^m\}} z_j^h(t))q_{l_m,i}^{m,1} \nonumber\\
&+(\sum_{\substack{h=1 \\ h \neq m}}^K \sum_{\substack{j=0 \\ j \neq l_h}}^L \mathds{1}_{\{w_{l_h}^h \leq w_{l_m}^m\}} z_j^h(t))q_{l_m,i}^{m,0} -\sum_{\substack{h=1 \\ h \neq m}}^K \gamma_h (\mathds{1}_{\{w_{l_h}^h>w_{l_m}^m\}}q_{l_m,i}^{m,1}+\mathds{1}_{\{w_{l_h}^h \leq w_{l_m}^m\}}q_{l_m,i}^{m,0}) 
\end{align}

Let  $g_i^m=\sum_{\substack{h=1 \\ h \neq m}}^K \gamma_h (\mathds{1}_{\{w_{l_h}^h>w_{l_m}^m\}}q_{l_m,i}^{m,1}+\mathds{1}_{\{w_{l_h}^h \leq w_{l_m}^m\}}q_{l_m,i}^{m,0})$ $\forall$ $i \in [0,L]$, and  $C=(c^1,\cdots,c^K)$ such that $c^k= (\gamma_k q_{l_k,0}^{k,0},\cdots,\gamma_k q_{l_k,L}^{k,0})$ and $c^m=((1-\alpha) q_{l_m,0}^{m,0}+ \alpha q_{l_m,0}^{m,1}-g_0^m,\cdots,(1-\alpha) q_{l_m,L}^{m,0}+ \alpha q_{l_m,L}^{m,1}-g_L^m)$ for each $k \neq m$.

Then, by replacing in the equation above for all $k$ $z_{l_k}^k(t)$ with $\gamma_k-\sum_{j=0, j\neq {l_k}}^L z_j^k(t)$, we obtain the following linear relation in $\jmath_{w^*}$ between $\tilde{\textbf{z}}(t+1)$ and $\tilde{\textbf{z}}(t)$ where $\tilde{\textbf{z}}$ is the proportion vector in which the elements $z^k_{l_k}$ for different $k$ are eliminated. 
\begin{equation}
\tilde{\textbf{z}}(t+1)=\textbf{Q}\tilde{\textbf{z}}(t)+\textbf{C}
\end{equation} 
The expression of matrix $\textbf{Q}$ is given in Appendix \ref{app:spectral_radius_Q_less_1}. The vector solution of the relaxed problem, denoted by $\tilde{\textbf{z}}^{*}$, is the fixed point of the aforementioned linear equation. Moreover, as $\tilde{\textbf{z}}^* \in \jmath_{w^*}$, and if $\tilde{\textbf{z}}(0)=\tilde{\textbf{z}}^*+\textbf{e}$, then we obtain:
\begin{equation}
\tilde{\textbf{z}}(t)-\tilde{\textbf{z}}^*=\textbf{Q}^t\textbf{e}
\end{equation}
The analysis of the above linear system is therefore important to prove the local optimality. We first provide the following lemma. 
\begin{Lemma}\label{lem:lambda_inf_1}
If for all eigenvalues $\lambda$ of $\textbf{Q}$, $|\lambda|<1$, then there exists a neighborhood $\Omega_{\sigma}(\tilde{\textbf{z}}^*) \subseteq \jmath_{w^*}$ such that if $\tilde{\textbf{z}}(0) \in \Omega_{\sigma}(\tilde{\textbf{z}}^*)$, we have the following: \\
1) For all $t \geq 0$, $||\tilde{\textbf{z}}(t)-\tilde{\textbf{z}}^*||<\sigma$ ($\tilde{\textbf{z}}(t) \in \jmath_{w^*}$). \\
2) $\tilde{\textbf{z}}(t)$ converges to $\tilde{\textbf{z}}^*$.
\end{Lemma}
\begin{proof}
	The proof follows from the convergence of the linear system.
\end{proof}
\begin{proposition}\label{spectral_radius_Q_less_1}
For all eigenvalue $\lambda$ of $\textbf{Q}$, $|\lambda|<1$
\end{proposition}
\begin{proof}
See the proof in appendix \ref{app:spectral_radius_Q_less_1}.
\end{proof}
The aforementioned result, combined with Lemma 9, proves the convergence of the fluid limit system (i.e. $\tilde{\textbf{z}}(t+1)=\textbf{Q}\tilde{\textbf{z}}(t)+\textbf{C}$). Consequently, $\textbf{z}$ converges to the fixed point of Equation \eqref{eq:relation_z_t+1_t} $\textbf{z}^*$. However, the above result is not enough to prove the local optimality, as we have to show that the stochastic vector $\textbf{Z}^N(t)$ converges to $\textbf{z}^*$ in probability.  For that, we introduce the discrete-time version of Kurtz Theorem  applied to our problem (see \cite{kurtz1978strong}):
\begin{proposition} \label{prop:kurth_theorem}
 There exists a neighborhood $\Omega_\delta(\textbf{z}^*)$ of
$\textbf{z}^*$ such that if $\textbf{Z}^N(0)=\textbf{z}(0)=\textbf{x} \in \Omega_\delta(\textbf{z}^*)$, then for any 
$\mu>0$  and finite time horizon $T$, there exist positive
constants $C_1$ and $C_2$ such that
\begin{equation}
P_\textbf{x}(\underset{0 \leq t < T}{\text{sup}} ||\textbf{Z}^N(t)-\textbf{z}(t)|| \geq \mu) \leq C_1 exp(-NC_2)
\end{equation}
where $\delta < \sigma$, and $P_\textbf{x}$ denotes the probability conditioned
on the initial state $\textbf{Z}^N(0) = \textbf{z}(0)= \textbf{x}$. Furthermore, $C_1$
and $C_2$ are independent of $\textbf{x}$ and $N$.
\end{proposition}
According to the above proposition, the system state $\textbf{Z}^N(t)$
behaves very closely to the fluid approximation model
$\textbf{z}(t)$ when the number of users $N$ is large. Since we
have shown the convergence of $\textbf{z}(t)$ to within $\Omega_{\sigma}(\textbf{z}^*)$, we are ready to establish the local convergence of the system state $\textbf{Z}^N(t)$ to $\textbf{z}^*$.
\begin{Lemma}\label{lem:result_kurth_theorem}
If $\textbf{Z}^N(0) = \textbf{x} \in \Omega_{\delta}(\textbf{z}^*)$, then for any $\mu > 0$,
there exists a time $T_0$ such that for any $T > T_0$, there
exists positive constants $s_1$ and $s_2$ with,
\begin{equation}
P_\textbf{x}(\underset{T_0 \leq t < T}{\text{sup}} ||\textbf{Z}^N(t)-\textbf{z}^*|| \geq \mu) \leq s_1 exp(-Ns_2)
\end{equation}
\end{Lemma}
\begin{proof}
See appendix \ref{app:result_kurth_theorem}.
\end{proof}
Now we are ready to prove the asymptotic local optimality of the proposed scheduling policy.
\begin{proposition}\label{prop:local_optimality}
If the initial state is in the set $\Omega_\delta(\textbf{z}^*)$, then
\begin{equation}
\lim_{T \rightarrow \infty} \lim_{N \rightarrow \infty} \frac{C^N_T(\textbf{x})}{N}=\frac{C^{RP,N}}{N}
\end{equation}
\end{proposition}
\begin{proof}
See appendix \ref{app:local_optimality}
\end{proof}
\section{Global asymptotic optimality}\label{sec:globalOptimality}
In this section, we will prove that from any initial state $\textbf{x}$, the expected time average cost obtained with the Whittle's Index policy is optimal when $N$ is very large. In contrast to the method used to prove the local optimality, we work here with the steady state distribution of the stochastic process $\textbf{Z}^N(t)$. To ensure that such a stationary distribution exists, we need to show that there is at least one recurrent state. Since the states evolve according to a finite state Markov chain, we just need to prove that there exists a state reachable from any other states.  
\begin{Lemma}\label{lem:reachability_z_0}
The state $\textbf{z}(0)=(\textbf{z}^1(0),\cdots,\textbf{z}^K(0))$, defined for each class $k$ as $\textbf{z}^k(0)=(1,0,\cdots,0)$, is reachable from any initial state using the Whittle's Index policy.
\end{Lemma}
\begin{proof}
See appendix \ref{app:reachability_z_0}
\end{proof}
This lemma is stronger than proving the existence of a recurrent state. Indeed, this allows us to deduce that $\textbf{Z}^N(t)$ evolves in one recurrent aperiodic class, and that there exists a stationary distribution for $\textbf{Z}^N(t)$ denoted by $\textbf{Z}^N(\infty)$. We still need to check if for a fixed $N$, there exists at least one recurrent state within $\Omega_{\epsilon}(\textbf{z}^*)$, as otherwise $\Omega_{\epsilon}(\textbf{z}^*)$ will be a transient class. If such state exists, surely $\textbf{Z}^N(t)$ will evolve in one recurrent class that contains this recurrent state. For that, we demonstrate here that $\textbf{z}^*$ is reachable from any state for a fixed $N$. Since $\textbf{z}(0)$ is reachable from any state, we just need to find a path from $\textbf{z}(0)$ to $\textbf{z}^*$. First, we start by giving $\alpha$ in function of the optimal proportion $\textbf{z}^*$. Rewriting the expression of $\alpha$ given in Proposition \ref{prop:optimal_sol_relax_prob_charchterisation}, we get: 
\begin{equation} \alpha= \sum_{k\neq m} \sum_{i=l_k+1}^L \gamma_k u^{l_k}_k(i)+\sum_{i=l_m+1}^L \gamma_m u^*_m(i)+(1-\theta)\gamma_m u^{l_m-1}_m(l_m)\end{equation}
The relation between the optimal vector $\textbf{z}^*$  and the stationary distribution under the optimal threshold is as fellows:\\  
For  $k\neq m$ $z_h^{k,*}=\gamma_k u^{l_k}_k(h)$.\\
For $k=m$ $z^{m,*}_h=\gamma_m ((1-\theta) u^{l_m-1}_m(h)+(\theta u^{l_m}_m(h))=\gamma_m u^*_m(h)$.\\
When $h=l_m \leq R_m-1$, we have that $u^{l_m-1}_m(l_m)=\rho_m=1/R_m$, and $u^{l_m}_m(l_m)=\rho_m=1/R_m=u^{l_m-1}_m(l_m)$. Then: 
\begin{equation}z_{l_m}^{m,*}= \gamma_m [(1-\theta) \rho_m+ \theta_m \rho_m]=\gamma_m \rho_m\end{equation}
Hence:
\begin{equation}\gamma_m (1-\theta) u^{l_m-1}_m(l_m)=\gamma_m (1-\theta) \rho_m=(1-\theta)z_{l_m}^{m,*}\end{equation}
Therefore:
\begin{equation}\alpha = \sum_{k\neq m} \sum_{i=l_k+1}^L z_i^{k,*}+\sum_{i=l_m+1}^L z_i^{m,*} +(1-\theta) z_{l_m}^{m,*}\end{equation}
In addition, it will be useful for the subsequent analysis in this section also to derive the exact expression of $u^{l_k}_k(h)$, for all states $h$, by applying the results found in Section \ref{sec:steady_state} when the threshold $l_k$ is strictly less than $R_k$. 
For $k\neq m$, we have:
\begin{align}
0 \leq h \leq l_k-1: &u^{l_k}_k(h)=\rho_k-(l_k-h)\rho_k^2 \nonumber\\
l_k\leq h \leq R_k-1: &u^{l_k}_k(h)=\rho_k \nonumber\\
R_k\leq h \leq l_k+R_k-1: &u^{l_k}_k(h)=(l_k+R_k-h)\rho_k^2
\end{align}
if $k=m$:
\begin{align}
0\leq h \leq l_m-1: &u^*_m(h)=\rho_m- (l_m-1-h+\theta)\rho_m^2\nonumber\\
l_m\leq h\leq R_m-1: &u^*_m(h)=\rho_m\nonumber\\
R_m \leq h \leq l_m+R_m-1: &u^*_m(h)=(l_m+R_m-1-h+\theta)\rho_m^2
\end{align}
Now, We will find a path from state $\textbf{z}(0)$ to $\textbf{z}^*$ under the Whittle's Index policy.
\begin{proposition}\label{prop:reachability_z_star}
By applying the Whittle's Index policy, the steady state $\textbf{z}^*$ is reachable from the state $\textbf{z}(0)$. 
\end{proposition}
\begin{proof}
See appendix \ref{app:reachability_z_star}.
\end{proof}
From this proposition, the state $\textbf{z}^*$ is reachable from any state, which means that $\textbf{z}^*$ is a recurrent state. However, as we remark in the demonstration of Proposition \ref{prop:reachability_z_star}, the considered actions schedule a proportion of users (i.e. not an integer value). This is not feasible and unrealistic for some (small) values of $N$ since the queues are not splittable. In fact, for some values of $N$, the state $\textbf{z}^*$ may not exist. On the other hand, we can say that for enough large $N$ , and for any $\epsilon > 0$, there exists at least one recurrent state within the neighborhood $\Omega_{\epsilon}(\textbf{z}^*)$. This will ensure that there is a path to enter a neighborhood $\Omega_{\epsilon}(\textbf{z}^*)$ from any initial state. However, it is important to ensure that the time to enter $\Omega_{\epsilon}(\textbf{z}^*)$ should not scale up with $N$. For that, we give the following assumption which will be later justified via numerical studies in Section IX.
\begin{assumption}\label{assump:bounded_time}    
We assume that the expected time to enter a neighborhood of $\textbf{z}^*$ from any initial state $\textbf{x}$ does not depend on the number of queues $N$. In other words, for all $N$ the time to enter a neighborhood $\Omega_{\epsilon}(\textbf{z}^*)$ denoted by $\Gamma^N_{\textbf{x}}(\epsilon)$ is bounded by a constant $T_{b_{\epsilon}}$.
\end{assumption}   
Now we provide a useful lemma that allows us to demonstrate the global asymptotic optimality.
\begin{Lemma}\label{lem:prob_z_infty_near_neighborhood_z_star}
Under assumption \ref{assump:bounded_time}, and for any $\epsilon$, we have that:
\begin{equation}
 \underset{N \rightarrow + \infty}{\text{lim}} P(\textbf{Z}^N(\infty) \in \Omega_{\epsilon}(\textbf{z}^*))=1 
\end{equation}
\end{Lemma}   
\begin{proof}
See lemma $6$ in \cite{ouyang2016downlink}.
\end{proof}
Since we have found a stationary distribution of $\textbf{Z}^N(t)$ under the Whittle's Index policy, the expected average cost under Whittle's Index policy for a fixed $N$ can be written as follows: 
\begin{equation}
 \lim_{T \rightarrow \infty} \frac{C_T^{N}(\textbf{x})}{N}=\sum_{k=1}^{K} \sum_{i=0}^{L} a_k \mathbb{E}\left[Z_{i}^{k,N}(\infty)\right]iN\end{equation}
\begin{theorem}\label{prop:global_optimality}
Under assumption \ref{assump:bounded_time}, and for any initial state, we have that:
\begin{equation} 
\underset{N \rightarrow +\infty}{\text{lim}} \lim_{T \rightarrow \infty} \frac{C_T^{N}(\textbf{x})}{N}=\frac{C^{RP,N}}{N}
\end{equation}  
\end{theorem}
\begin{proof}
See appendix \ref{app:global_optimality}
\end{proof}
%We can show that $W_R(R)$ is decreasing function in $R$ for given $L$, then if it exist $k$ less strictly than $K$, such that $W*=W_{R_k}(R_k)$ then for any $l>k$, ($R_l > R_k$), $W^*>W_{R_l}=W_{L}$, then for this classe $l$ the optimal solution is to not transmit for all states, we will be in the passive mode for all time: this case doesn't interest us since there is no minimum fairness.
%when $W^*=W_{R_K}(R_K)=W_{R_K}(L)$,since this point represent the intersection between the two cost line when $R-1$ and $L$, and no more (all cost lines $R$ until $L-1$, don't intersect the cost line $R-1$ before $L$ intersect $L$, as we already prove), the optimal threshold can be either $R-1$ or $L$. then we can proceed then with the same steps however it will be some difference regarding the expression of 
\section{Numerical Results}\label{sec:numerics}
In this section, we provide numerical results that confirm the asymptotic optimality of the developed Whittle index policy. To that extent, we consider $2$ classes having a respective rate of $R_1=5$ and $R_2=10$. Moreover, we suppose that $\alpha=1/2$, $L=50$, $\gamma_1 =\gamma_2 = 1/2$, and $a_1=a_2=a=1$. We also consider two initial states $x$ and $y$ such that all the queues are equal to $0$ and $L$ respectively.
\subsection{Verification of Assumption \ref{assump:bounded_time}}
We plot in Figure \ref{fig:hitting_time}, the evolution of the time needed to enter a neighborhood $\Omega_{\epsilon}(\textbf{z}^*)$ (i.e. hitting time of $\Omega_{\epsilon}(\textbf{z}^*)$) with respect to $N$, given that $\epsilon$ is small enough.
\begin{figure}
%\centering
\includegraphics[width=0.4\textwidth]{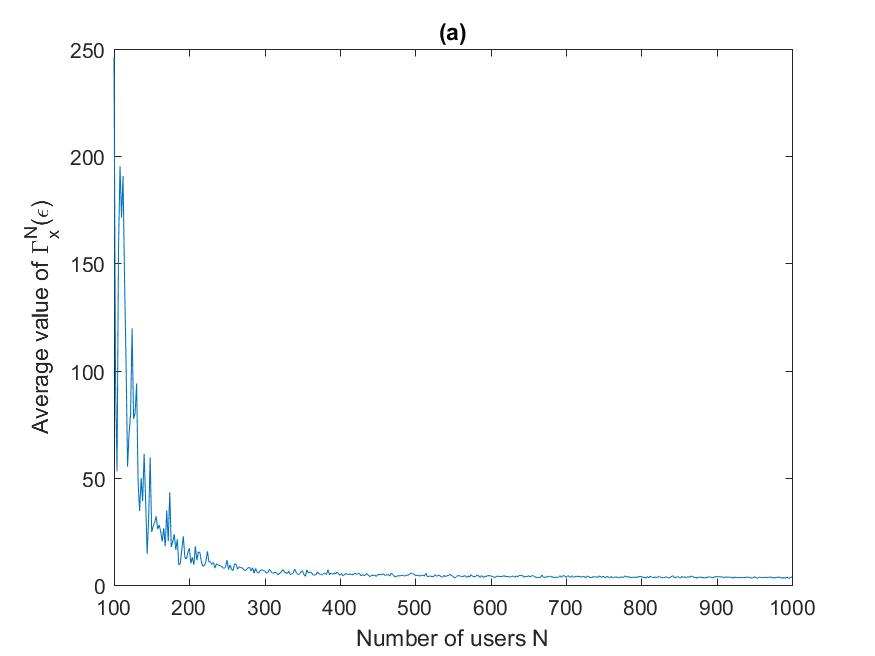}
\includegraphics[width=0.4\textwidth]{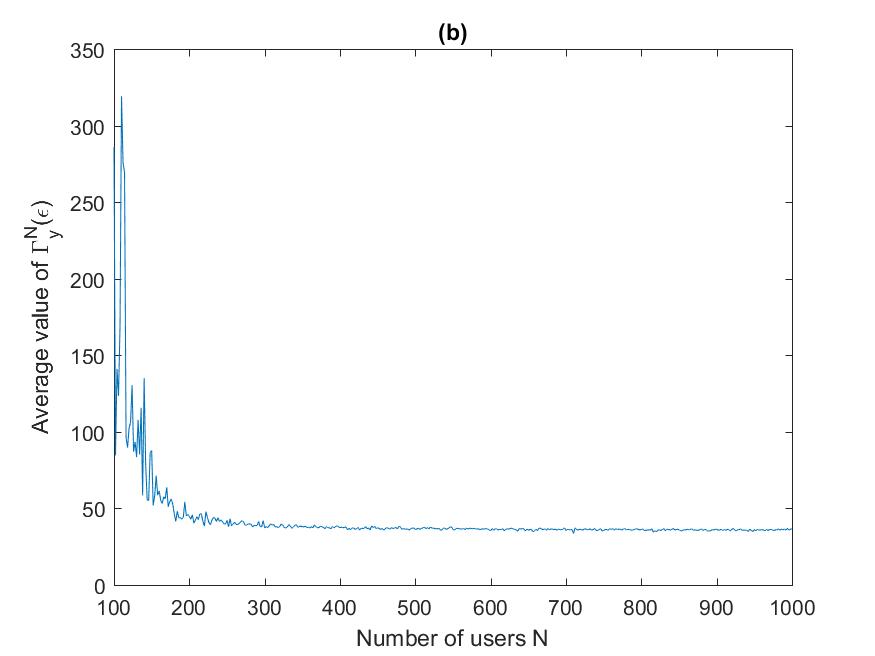}
\caption{Hitting Time of $\Omega_{\epsilon}(\textbf{z}^*)$ in function of $N$: (a) $\textbf{Z}^N(0)=\textbf{x}$, (b) $\textbf{Z}^N(0)=\textbf{y}$}
\label{fig:hitting_time}
\end{figure}
One can see that for large values of $N$, the hitting time can be considered as a constant and does not diverge, and this is true for both initial states $\textbf{x}$ and $\textbf{y}$. This implies that the hitting time is bounded for large values of $N$ which consolidates Assumption \ref{assump:bounded_time}.
\subsection{Performance of the Whittle's Index policy}
In this section, we compare the long run expected average cost per user under the Whittle's Index policy, i.e. \linebreak $\lim_{T \rightarrow \infty} C_T^N(\textbf{x})$,  with the one obtained by applying the Max-Weight policy $MW$. The latter schedules, at each time $t$, the $M$ weighted longest queues (equivalently the $M$ highest $a_kQ^k_i(t)$). We also compare the performance of these two policies with the optimal cost per user obtained  by using the optimal solution of the relaxed problem, i.e. $C^{RP}/N$. The results are plotted in Figures (\ref{fig:average_cost_evoluton}.a) and (\ref{fig:average_cost_evoluton}.b) respectively for the initial states $\textbf{x}$ and $\textbf{y}$ (defined above).
\begin{figure}
%\centering
\includegraphics[width=0.4\textwidth]{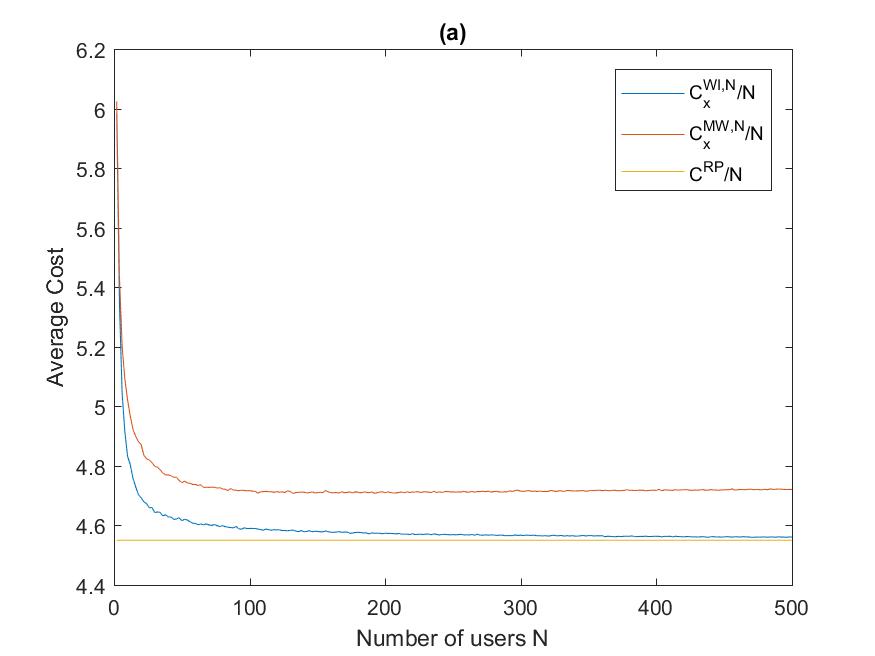}
\includegraphics[width=0.4\textwidth]{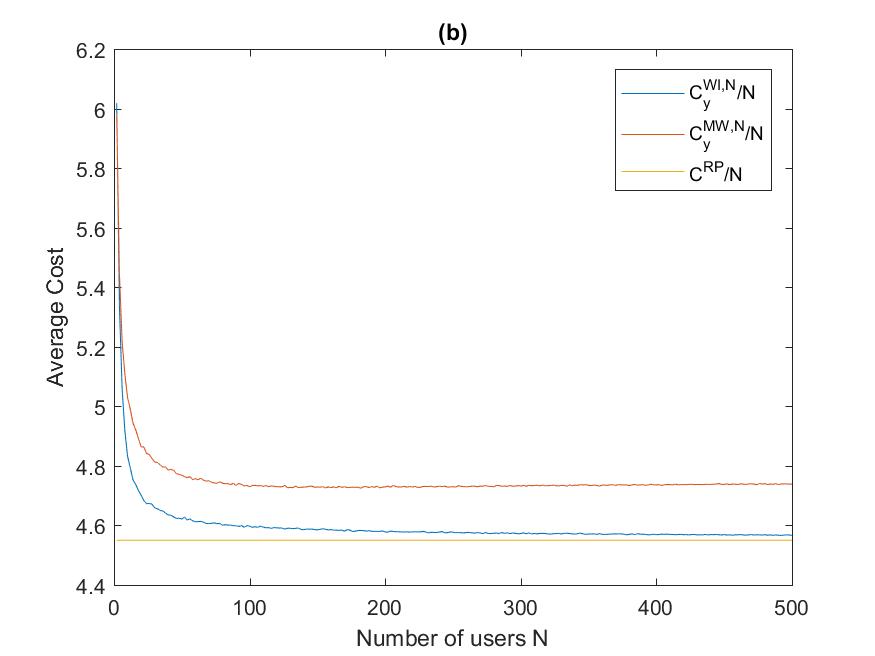}
\caption{Performance evaluation of Whittle's Index policy (a) $\textbf{Z}^N(0)=\textbf{x}$, (b) $\textbf{Z}^N(0)=\textbf{y}$}
\label{fig:average_cost_evoluton}
\end{figure}
One can see that for large $N$, regardless of the initial state, the cost incurred by adopting the Whittle's Index policy tends to the optimal cost of the relaxed problem, which proves that it asymptotically converges to the optimal solution of the original problem.
One can also remark that the optimal cost of the relaxed problem per user is constant and does not depend on $N$ (see section \ref{sec:relaxed}).
Lastly, we remark that the solution given by $MW$ is suboptimal and lacks behind our proposed scheduling scheme. \\
\subsection{Fairness among users}
In order to improve the fairness among the users in the network, one can use the developed Whittle index policy in this paper up to some modifications. For example, we introduce in this section the following a new policy $\Theta$ which works as follows: at each time slot $t$, we schedule the users with the highest $W_k(q_i^k(t))\overline{D_k}(q_i^k(t))$, where $q_i^k(t)$ is the queue state of user $i$ in class $k$, $W_k$ is the Whittle index of state $q_i^k(t)$ when the transmission rate is $R_k$ and $\overline{D_k}(q_i^k(t))=\frac{\sum_{u=1}^t a_k q_i^k(u)}{t}$. To evaluate numerically the performance of this policy, we consider the case of two classes of users. To that extent, we consider the following two costs $C_1(N)$ and $C_2(N)$ incurred respectively by users of class $1$ and users of class 2, specifically $C_1(N)=\lim_{T \rightarrow \infty} \frac{1}{T} \mathbb{E}\left[ \sum_{t=0}^{T-1} \sum_{i=1}^{{\gamma_1}N} a_1 q_i^1(t) \mid \textbf{x}\right]$  and $C_2(N)= \lim_{T \rightarrow \infty} \frac{1}{T} \mathbb{E}\left[ \sum_{t=0}^{T-1} \sum_{i=1}^{{\gamma_2}N} a_2 q_i^2(t) \mid \textbf{x}\right]$. We plot these quantities with respect to $N$ in Figure $5$. In Figure (\ref{fig:fairness_evaluation}.a), the costs are obtained by applying the new policy $\Theta$ while in Figure (\ref{fig:fairness_evaluation}.b) the standard Whittle index policy is applied.   
\begin{figure}
%\centering
\includegraphics[width=0.4\textwidth]{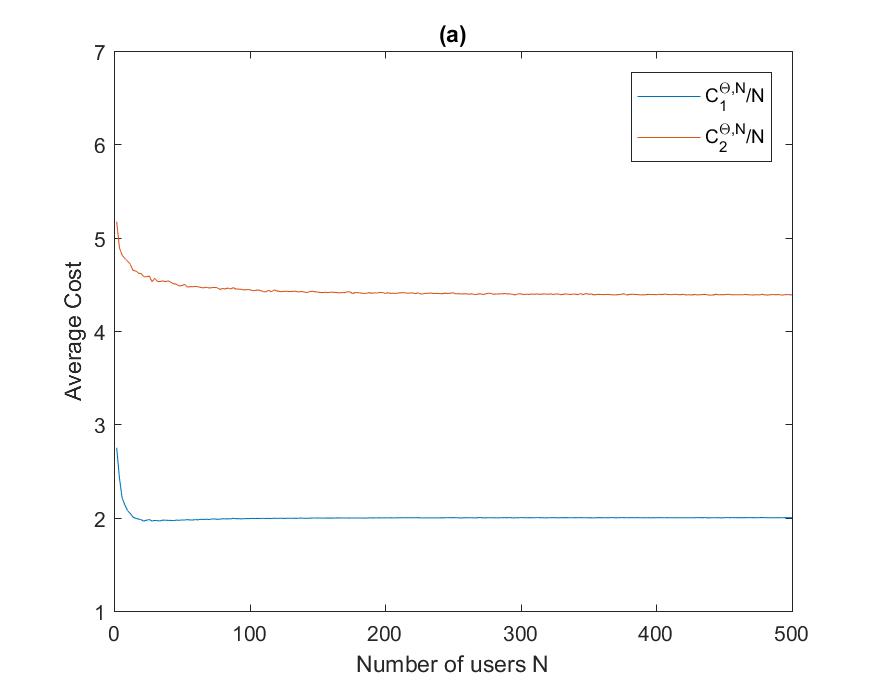}
\includegraphics[width=0.4\textwidth]{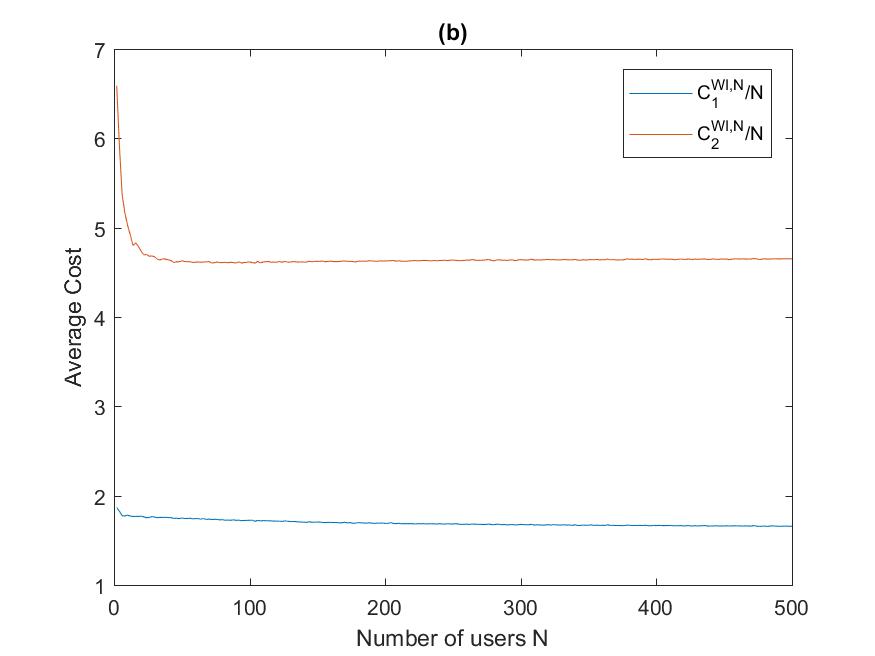}
\caption{Evaluation of $C_1$ and $C_2$ in function of $N$: (a) Policy $\Theta$, (b) Whittle's Index policy}
\label{fig:fairness_evaluation}
\end{figure}
We conclude that the new policy gives a better performance in terms of  fairness, since it reduces the gap between the costs of the two classes of users.

\subsection{Performance of Whittle Index when $L<R$}
To get more comprehensive results, we also evaluate the performance of Whittle index policy when $L<R$ by considering $L=10$, $R_1=20$ and $R_2=30$. We let $\alpha=1/2$ and $a_1=a_2=a=1$. To that end, we compare the long run expected average cost per user under the Whittle index policy, with the one obtained by applying the Max-Weight policy $MW$. We see in Figure \ref{fig:performance_L_less_R}, that Whittle index policy still asymptotically optimal even when $L < R$. Hence, we can presume that Whittle index policy is asymptotically optimal regardless of the value of $L$. This in fact has been analytically proved throughout this whole paper when $L \geq 2R$.
\begin{figure}
\centering
\includegraphics[width=0.4\textwidth]{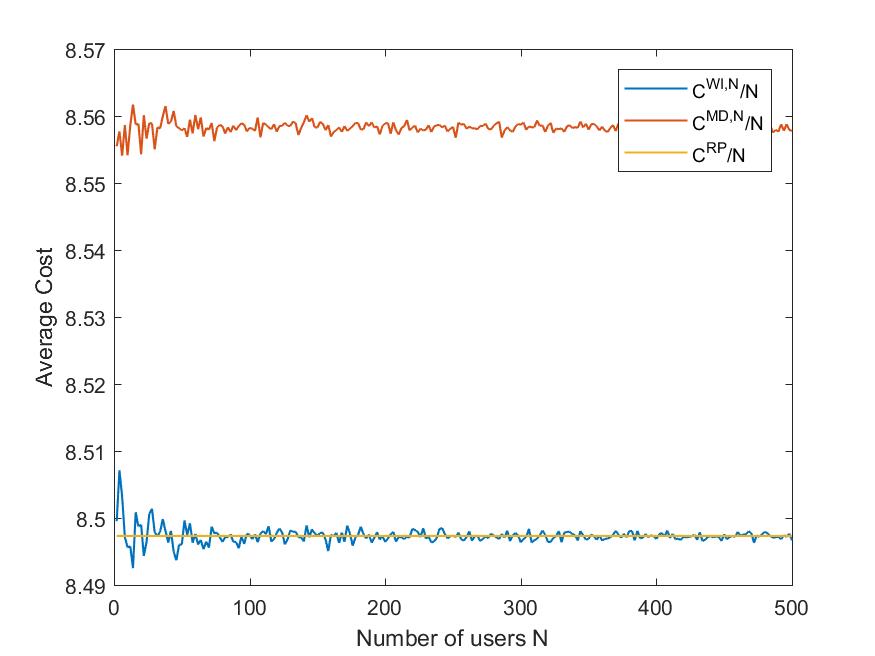}
\caption{Performance evaluation of Whittle's Index policy: $L<R$}
\label{fig:performance_L_less_R}
\end{figure}      
\section{Conclusion}
%***à changer***
In this paper, we have studied the problem of users and channels scheduling under bursty traffic arrivals. At each time slot, only $M$ channels can be allocated to the users knowing that a user can be allocated one channel at most.  We formulated a Lagrangian relaxation of the optimization problem and provided a characterization of the optimal solution of this relaxed problem. We then developed a simple Whittle index policy to allocate the channels to the users and proved its asymptotic local and global optimality when the numbers of users and channels are large enough. This result is of interest as the developed  Whittle Index Policy has a low complexity and is near optimal for large number of users. We then provided numerical results that corroborate our claims.

%centralized scheduling over uncorrelated channels. We consider
%the scenario where transmission rate is fixed, that means that channel gain is constant and known by the scheduler (the base station), and also the arrival packets follows an uniform distribution for each user. We analytically characterized
%the performance of Whittle’s Index Policy for uplink
%scheduling, and proved its local asymptotic
%optimality property as the number of users and the number of available uncorrelated channels scales.
%Specifically, provided that the initial system state is
%within a certain region, we established the local optimality
%of Whittle’s Index Policy by establishing the
%evolution of state proportion over the whole number of users with fluid approximation.
%Our results are important since Whittle’s Index Policy has low complexity ($O(N^2)$)
%and is near optimal for large number of users. %We can extend our work for infinite buffer size and for any arrival packet distribution.
\nocite{ansell2003Whittle}
\nocite{buyukkoc1985c}
\nocite{gittins2011multi}
\nocite{kurtz1978strong}
\nocite{larranaga2015dynamic}
\nocite{ouyang2016downlink}
\nocite{papadaki2002exploiting}
\nocite{puterman2014markov}
\nocite{ross2014introduction}
\nocite{ruan2017delay}
\nocite{weber1990index}
\nocite{Whittle1988restless}
\nocite{papadimitriou1999complexity}
\nocite{liu2010indexability}
\bibliographystyle{IEEEtran}
\bibliography{bil}

\begin{appendices} 

%\begin{align}
%\sum_{k=1}^{K} \sum_{i=1}^{{\gamma_k}N} \sum_{q_i^{'k}} Pr(q_i^{'k}|q_i^k,s_i^k)V_i^k(q_i^{'k})&=\sum_{k=1}^{K} \sum_{i=1}^{{\gamma_k}N} \sum_{\textbf{q}'} \sum_{q_i^{'k}}  Pr(\textbf{q}'|\textbf{q},\textbf{s},q_i^{'k}) Pr(q_i^{'k}|q_i^k,s_i^k)V_i^k(q_i^{'k})
%\end{align}
\renewcommand\thesection{\AlphAlph{\value{section}}}
%\section{fdsk \thesection}
%\thesection \thesection
%\section{fdsk \thesection}
%\section{fdsk}\thesection
%\section{fdsk}\thesection
%\section{fdsk}\thesection
%\section{fdsk}\thesection
%\section{fdsk}\thesection
%\section{fdsk}\thesection
%\section{fdsk}\thesection
%\section{fdsk}\thesection
%\section{fdsk}\thesection
%\section{fdsk}\thesection
%\section{fdsk}\thesection

\section{proof of Proposition \ref{prop:value_function_decomposition}}\label{app:value_function_decomposition}
We consider the Bellman Equation~\eqref{eq:individual}. By summing the RHS and the LHS of Equation~\eqref{eq:individual}, for all $k$ and $i$ we obtain:
\begin{align}
\sum_{k=1}^{K} \sum_{i=1}^{{\gamma_k}N} [V_i^k(q_i^k)+ \theta_i^k]=&\sum_{k=1}^{K} \sum_{i=1}^{{\gamma_k}N} \underset{s_i^k}{\text{min}} \{ C_k(q_i^k,s_i^k)+ \sum_{q_i^{'k}} Pr(q_i^{'k}|q_i^k,s_i^k)V_i^k(q_i^{'k})\} \\
=& \underset{\textbf{s}}{\text{min}} \{ \sum_{k=1}^{K} \sum_{i=1}^{{\gamma_k}N} [C_k(q_i^k,s_i^k)+ \sum_{q_i^{'k}} Pr(q_i^{'k}|q_i^k,s_i^k)V_i^k(q_i^{'k})]\},
\end{align}
where $\textbf{s}=(s_1^1,\ldots,s_{\gamma_1N}^1,\ldots, s_{1}^K,\ldots,s_{\gamma_kN}^K)$. 
We also have that:
\begin{align}
Pr(\textbf{q}'|\textbf{q},\textbf{s})=\sum_{q_i^{'k}} Pr(\textbf{q}'|\textbf{q} ,\textbf{s} ,q_i^{'k})Pr(q_i^{'k}|\textbf{q},\textbf{s})=\sum_{q_i^{'k}} Pr(\textbf{q}'|\textbf{q},\textbf{s},q_i^{'k})Pr(q_i^{'k}|q_k^i,s_i^k),
\end{align}
for all $\textbf{q}=(q_1^1,\ldots,q_{\gamma_1N}^1,\ldots,q_{1}^K,\ldots,q_{\gamma_KN}^K)$ and $\textbf{q}^{'}=(q_1^{'1},\ldots,q_{\gamma_1N}^{'1},\ldots,q_{1}^{'K},\ldots,q_{\gamma_KN}^{'K})$.
Since $Pr(q_i^k|\textbf{q},\textbf{s})$ only depends on the decision taken with respect to user $i$ in class $k$, we obtain:
\begin{align}
\sum_{k=1}^{K} \sum_{i=1}^{{\gamma_k}N} \sum_{q_i^{'k}} Pr(q_i^{'k}|q_i^k,s_i^k)V_i^k(q_i^{'k})&=\sum_{k=1}^{K} \sum_{i=1}^{{\gamma_k}N} \sum_{\textbf{q}'} \sum_{q_i^{'k}}  Pr(\textbf{q}'|\textbf{q},\textbf{s},q_i^{'k}) Pr(q_i^{'k}|q_i^k,s_i^k)V_i^k(q_i^{'k})\\
&=\sum_{\textbf{q}'} Pr(\textbf{q}'|\textbf{q},\textbf{s}) \sum_{k=1}^{K} \sum_{i=1}^{{\gamma_k}N} V_i^k(q_i^{'k})
\end{align}
From the previous equations we obtain:
\begin{align}
\sum_{k=1}^{K} \sum_{i=1}^{{\gamma_k}N} V_i^k(q_i^k)+ \sum_{k=1}^{K} \sum_{i=1}^{{\gamma_k}N}\theta_i^k=& \underset{\textbf{s}}{\text{min}} \ [\sum_{k=1}^{K} \sum_{i=1}^{{\gamma_k}N} C_k(q_i^k,s_i^k)+ \sum_{k=1}^{K} \sum_{i=1}^{{\gamma_k}N}\sum_{q_i^{'k}} Pr(q_i^{'k}|q_i^k,s_i^k)V_i^k(q_i^{'k})]\\
=& \underset{\textbf{s}}{\text{min}} \ [\sum_{k=1}^{K} \sum_{i=1}^{{\gamma_k}N} C(q_i^k,s_i^k)+ \sum_{\textbf{q}'} Pr(\textbf{q}'|\textbf{q},\textbf{s}) \sum_{k=1}^{K} \sum_{i=1}^{{\gamma_k}N} V_i^k(q_i^{'k})]
\end{align}
According to Theorem $2.1$ Chapter $2$, \cite{ross2014introduction}, it exists a unique function $V$ and a constant $\theta$ that resolve the equation \eqref{eq:Bellman}. Subsequently, since we have found a bounded function $\sum_{k=1}^{K} \sum_{i=1}^{{\gamma_k}N} V_i^k(q_i^k)$, and a constant $\sum_{k=1}^{K} \sum_{i=1}^{{\gamma_k}N}\theta_i^k$ that satisfy also the equation~\eqref{eq:Bellman}, then $V(\textbf{q})=\sum_{k=1}^{K} \sum_{i=1}^{{\gamma_k}N} V_i^k(q_i^k)$ and $\theta=\sum_{k=1}^{K} \sum_{i=1}^{{\gamma_k}N}\theta_i^k$. This is equivalent to finding for each user the decision that minimizes the right hand side of each individual Bellman equation. This concludes the proof.
\section{proof of Lemma \ref{submodilarity}}\label{app:submodilarity}
We first prove that $C(\cdot,\cdot)$ is submodular. That is,
$(C(q+1,1)-C(q+1,0))-(C(q,1)-C(q,0))=a(q+1)+W-a(q+1)-(aq+W-aq)=0 \leq 0$.
The latter is obtained by substituting the values of $C(q',s)$ for $s\in\{0,1\}$ and $q'\in\{q,q+1\}$.
In order to prove that $\sum_{q'} Pr(q'|q,s)V(q')$ is submodular, we  distinguish between two cases:\\
Case 1) $q < R$, then:
\begin{align}
\sum_{q'} Pr(q'|q+1,1)V(q')-\sum_{q'} Pr(q'|q+1,0)V(q')=&\sum_{q'=0} Pr(A=q')V(q')-\sum_{q'=q+1} Pr(A=q'-q-1)V(q') \nonumber\\
=&\sum_{q'=0} Pr(A=q')V(q')-\sum_{q'=q} Pr(A=q'-q)V(q'+1) \nonumber\\
\leq& \sum_{q'=0} Pr(A=q')V(q')-\sum_{q'=q} Pr(A=q'-q)V(q') \nonumber\\
=&\sum_{q'} Pr(q'|q,1)V(q')-\sum_{q'} Pr(q'|q,0)V(q')
\end{align}
The inequality follows from the fact that $V(\cdot)$ is increasing.
This concludes the proof for $q < R$.\\
Case 2) $q \geq R$, then:
\begin{align}\label{eq1}
\sum_{q'} Pr(q'|q+1,1)V(q')-\sum_{q'} Pr(q'|q+1,0)V(q')=&\sum_{q'} Pr(A=q'-q-1+R)V(q')-\sum_{q'} Pr(A=q'-q-1)V(q')\nonumber\\
=&\sum_{q'=q+1-R} Pr(A=q'-q-1+R)V(q')-\sum_{q'=q+1} Pr(A=q'-q-1)V(q')\nonumber\\
=&\sum_{q'=q} Pr(A=q'-q)V(q'-R+1)-\sum_{q'=q} Pr(A=q'-q)V(q'+1)
\end{align}
Moreover, we have: 
\begin{align}\label{eq2}
\sum_{q'} Pr(q'|q,1)V(q')-\sum_{q'} Pr(q'|q,0)V(q')=\sum_{q'=q} Pr(A=q'-q)V(q'-R)-\sum_{q'=q} Pr(A=q'-q)V(q').
\end{align}
Subtracting Equation~\eqref{eq1} and~\eqref{eq2} (i.e., \eqref{eq1}-\eqref{eq2}) we obtain
\begin{align}\label{eq3}
\sum_{q'=q} Pr(A=q'-q)[(V(q'-R+1)-V(q'-R))-(V(q'+1)-V(q'))] \leq 0,
\end{align}
which follows from the R-convexity of $V(\cdot)$. Therefore, $\sum_{q'} Pr(q'|q,s)V(q')$ is submodular.

\section{Proof of Proposition \ref{prop:general_passage}}\label{app:general_passage}
When $i <L$:\\
1) $j\leq n$:\\
Since $j \leq n$, the optimal decision is to stay idle, that means if $A$ denotes the number of arrival packets, in the next time slot the number of packets will be $i=j+A$ with $A\leq R-1$ and then $A=i-j$. Therefore, the probability to transition from state $j$ to $i$ is the probability that $A=i-j$, which is exactly $\pi_{i-j}$. \\
 2) $j>n$:\\
The optimal decision in this case is to transmit. However, at most $\min(R,j)$ can be transmitted. Taking into account the $A$ arrival packets, then the new state for the next time slot will be $i=j-\min(R,j)+A=(j-R)^++A$, which implies that $A=i-(j-R)^+$. This explains that the probability to transition from state $j$ to $i$ is the probability that $A$ is equal to $i-(j-R)^+$ which is equal to $\pi_{i-(j-R)^+}$. \\
When $i=L$:\\
1) $j \leq n$:\\
The optimal decision is a passive action. Then $A$ arrival packets are added to the $j$ packets present in the queue. For the next time slot, the number of packets is $j+A$. According to equation ~\eqref{eq: qeue_evolution}, since we cannot exceed the buffer length $L$, we reach the state $L$ if $j+A \geq L$. Since $A \leq R-1$, then the probability of this event or equivalently the probability to transition from state $j$ to state $L$ is $Pr(L-j \leq A \leq R-1)=\sum_{k=L-j}^{R-1} Pr(A=k)=(R-L+j)\pi_{L-j}$.\\
2) $j > n$:\\
The optimal decision is an active action, thus to reach the next state the arrival packet number $A$ must be in the set $[L-(j-R)^+,R-1]$. Then the probability to transition from $j$ to $L$ is $Pr(L-(j-R)^+ \leq A \leq R-1)=\sum_{k=L-(j-R)^+}^{R-1} Pr(A=k)=(R-L+(j-R)^+)\pi_{L-(j-R)^+}$. We can therefore conclude the results.

\section{Proof of Proposition \ref{eq:stat.dist}}\label{app:eq:stat.dist}
%In the whole proof, we deal only with the third case $L \geq 2R$. While for the remaining cases, the approach used to derive the stationary distribution is nearly similar up to some modifications.\\
We prove the four sub-cases separately when $L \geq 2R$:
\begin{enumerate}
\item First case: $-1 \leq n < R$:
\begin{equation} u(i)= \sum_{j=0}^n p^n(j,i)u(j)+ \sum_{j=n+1}^R p^n(j,i)u(j) + \sum_{j=R+1}^L p^n(j,i)u(j)\end{equation} 
We first provide the following lemma that follows from Proposition \ref{prop:general_passage}.
\begin{Lemma}\label{passage_1}
when $i<L$:\\
\begin{equation}
p^n(j,i)=\left\{
    \begin{array}{ll}
        \pi_{i-j} & if \ 0 \leq j \leq n \\
        \pi_{i} & if  \ n+1 \leq j \leq R-1 \\
        \pi_{i-(j-R)} & if \ R \leq j \leq L
    \end{array}
\right.
\end{equation}
when $i=L$:
\begin{equation}
p^n(j,i)=\left\{
    \begin{array}{ll}
        0 & if  \ 0 \leq j \leq n \\
        0 & if  \ n+1 \leq j \leq L \\
    \end{array}
\right.
\end{equation}
\end{Lemma}
Using Lemma \ref{passage_1}, we have:\\
if $i <L$
\begin{equation} u(i)=\sum_{j=0}^n \pi_{i-j}u(j)+ \sum_{j=n+1}^R \pi_{i}u(j) +\sum_{j=R+1}^L \pi_{i-(j-R)}u(j)\end{equation} 
By definition of $\pi$ given in definition \ref{pi}, then: 
\begin{equation} u(i)= \sum_{\max(i-R+1,0)}^{\min(i,n)} \rho u(j) \ + \sum_{n+1}^R \pi_i u(j) + \ \sum_{\max(i+1,R+1)}^{\min(i+R,L)} \rho u(j)
\end{equation} 
\\
In order to prove Proposition \ref{eq:stat.dist} for this case, according to Lemma \ref{passage_1} we will distinguish between five sub-cases:\\
a) $i=L$\\
b) $n+R+1 \leq i \leq L-1$\\
c) $n+1 \leq i \leq R-1$\\
d) $0 \leq i \leq n$\\
e) $R \leq i \leq n+R$
\\

a) Proof of $u(i)=0$ for $i=L$:\\
if $i=L$, since $\forall j$ $p^n(j,L)=0$, then 
\begin{equation}u(L)=0
\end{equation}
\\
\\
b) Proof of $u(i)=0$ for $n+R+1 \leq i \leq L-1$:\\ 
For this case, we prove by strong induction in decreasing order that $u(i)=0$ \\
In fact we have that $u(L)=0$, and for $n+R < i \leq L$, $\pi_i=0$ because $i > R-1$,  $\min(i,n)=n < i-R+1=\max(i-R+1,0)$, and $i+1 \geq R+1$, then:
\begin{equation} u(i)=\sum_{i+1}^{\min(i+R,L)} \rho u(j) \end{equation} 
we consider by induction that $\forall k \in [i,L]$,  $u(k)=0$.\\
So $u(i-1)=\sum_{i}^{\min(i-1+R,L)} \rho u(j) = 0$.\\
Hence we conclude the result.
\\
c) Proof of $u(i)=\rho$ for $n+1 \leq i \leq R-1$:\\
We have $\max(i-R+1,0)=0$, $\min(i,n)=n$, $\pi_i=\rho$ (since $0 \leq i \leq R-1$), $\max(i+1,R+1)=R+1$ and $\min(i+R,L)=i+R$ (recall that $i+R<2R\leq L$). This implies,
\begin{equation}
 u(i)=\sum_{0}^{n} \rho u(j)+ \sum_{n+1}^R \rho u(j) +\sum_{R+1}^{i+R} \rho u(j)
\end{equation} 
Now, we prove that $u(i)=\rho$\\
We have that:
\begin{align}
u(i)&=\sum_{0}^{n} \rho u(j)+ \sum_{n+1}^R \rho u(j) +\sum_{R+1}^{i+R} \rho u(j)=\rho [\sum_{0}^{n} u(j)+ \sum_{n+1}^R u(j) +\sum_{R+1}^{i+R} u(j)]\\
u(i)&=\sum_{j=0}^{i+R} \rho u(j)
\end{align}  
We have that $i+R > n+R$, then $u(p)=0$ for all $p \in [n+R+1,i+R]$. We can hence simplify the expression of $u(i)$ as follows:
\begin{equation} u(i)=\rho \sum_{j=0}^{n+R} u(j)\end{equation}   
 Since we have proved that when $j > n+R$ , $u(j)=0$ (sub-case (b)), then $\sum_{j=0}^{n+R} u(j)=1$ ($\sum_0^L u(j)=1$ because $u$ is probability distribution), i.e. $u(i)=\rho$.\\
This ends the proof of sub-case (c).
\\
We will provide a useful lemma which allows us to prove Proposition \ref{eq:stat.dist} for the cases (d) and (e). Before giving this lemma, we will give general expressions of $u(i)$ for these two cases.
\\
If $0 \leq i \leq n$:\\
$i \leq n < R$, which implies that $i-R+1 \leq 0$, $\max(i-R+1,0)=0$, $\min(i,n)=i$, $\pi_i=\rho$ since $0 \leq i \leq n<R $, $\max(i+1,R+1)=R+1$, and $i+R \leq n+R <2R \leq L$. Therefore $\min(i+R,L)=i+R$, which implies that: 
\begin{equation} u(i)=\sum_{0}^{i} \rho u(j)+ \sum_{n+1}^{R} \rho u(j) +\sum_{R+1}^{i+R} \rho u(j)
\end{equation}  
If $R \leq i \leq n+R$:\\
We have $\max(i-R+1,0)=i-R+1$, $\min(i,n)=n$ (due to $i \geq R > n$),  $\pi_i=0$ (since $i> R-1$) and $\max(i+1,R+1)=i+1$. Then:
\begin{equation} u(i)= \sum_{i-R+1}^{n} \rho u(j) + \ \sum_{i+1}^{\min(i+R,L)} \rho u(j)
\end{equation} 
\begin{Lemma}\label{ro_constante}
 for $0 \leq k \leq n$: \begin{equation}u(n+R-k)+u(n-k)=\rho\end{equation}\\
\end{Lemma}
\begin{proof}
\renewcommand{\qedsymbol}{$\blacksquare$}
See appendix \ref{app:ro_constante}
\end{proof}

d) Proof of $u(i)=\rho-\rho^2(n-i)$ for $0 \leq i \leq n$:\\
We start by proving by induction that for $k \in [0,n]$ $u(n-k)=\rho-\rho^2 k$,
we have for $0 \leq k \leq n$, \ $0 \leq n-k \leq n$, then:
\begin{equation} u(n-k)=\sum_{0}^{n-k} \rho u(j)+ \sum_{n+1}^{R} \rho u(j) +\sum_{R+1}^{n-k+R} \rho u(j)
\end{equation} 
For $k=0$,
\begin{align}
u(n-0)&=\rho [\sum_{0}^{n} u(j)+ \sum_{n+1}^R u(j) +\sum_{R+1}^{n+R} u(j)]\\
&=\sum_{j=0}^{n+R} \rho u(j)\\
&=\rho
\end{align}
\\
We suppose that the expression is true for some $k$, we prove it for $k+1$
\begin{align}
u(n-(k+1))&=\sum_{0}^{n-k-1} \rho u(j)+ \sum_{n+1}^{R} \rho u(j) +\sum_{R+1}^{(n-k-1+R)} \rho u(j)\\
&=\sum_{0}^{n-k} \rho u(j)+ \sum_{n+1}^{R} \rho u(j) +\sum_{R+1}^{(n-k+R)} \rho u(j)-\rho(u(n-k)+u(n-k+R))\\
&=u(n-k)-\rho(u(n-k)+u(n+R-k))\\
&=\rho-k\rho^2-\rho(u(n-k)+u(n+R-k))
\end{align}
Using Lemma \ref{ro_constante}, $u(n-k)+u(n+R-k)=\rho$, then:
\begin{align}
u(n-(k+1))&=\rho-k\rho^2-\rho(\rho)\\
&=\rho-k\rho^2-\rho^2\\
&=\rho-(k+1)\rho^2
\end{align}
Thus we conclude that for $k \in [0,n]$ $u(n-k)=\rho-k\rho^2$.\\
For $i \in [0,n]$, we replace $k \in [0,n]$ by $n-i$ ($n-i \in [0,n]$), we get: \begin{equation}u(i)=u(n-(n-i))=\rho-\rho^2(n-i)
\end{equation}  
e) Proof of $u(i)=\rho^2(n+R-i)$ for $R \leq i \leq n+R$:\\
For that we prove that for $k \in [0,n]$ $u(n+R-k)=\rho^2 k$.\\
From the above result in the case (d), we get $u(n-k)=\rho-k\rho^2$.\\
So, according to Lemma \ref{ro_constante}:
\begin{align}
u(n+R-k)&=\rho-u(n-k)\\
&=\rho-(\rho-\rho^2 k)\\
u(n+R-k)&=\rho^2 k
\end{align}
For $i \in [R,n+R]$, we replace  $k \in [0,n]$ by $n+R-i$ ($n+R-i \in [0,n]$), we get:\begin{equation}u(i)=u(n+R-(n+R-i))=\rho^2(n+R-i)\end{equation}

\item Second case: $R \leq n < L-R$: \\
\begin{equation}u(i)= \sum_{j=0}^n p^n(j,i)u(j)+ \sum_{j=n+1}^L p^n(j,i)u(j)
\end{equation}
\begin{Lemma}\label{passage_2}
when $i<L$:\\
\begin{equation}
p^n(j,i)=\left\{
    \begin{array}{ll}
        \pi_{i-j} & if \ 0 \leq j \leq n \\
        \pi_{i-(j-R)} & if \ n+1 \leq j \leq L
    \end{array}
\right.
\end{equation}
when $i=L$:
\begin{equation}
p^n(j,i)=\left\{
    \begin{array}{ll}
        0 & if  \ 0 \leq j \leq n \\
        0 & if  \ n+1 \leq j \leq L 
    \end{array}
\right.
\end{equation}
\end{Lemma}
The results of Lemma \ref{passage_2} come from Proposition \ref{prop:general_passage}.
Using Lemma \ref{passage_2}:\\
if $i < L$
\begin{equation} u(i)= \sum_{j=0}^n \pi_{i-j} u(j)+ \sum_{j=n+1}^L \pi_{i-(j-R)} u(j)\end{equation} 
By definition of $\pi$ given in definition \ref{pi}, then:
\begin{equation} u(i)=\sum_{\max(i+1-R,0)}^{\min(n,i)} \rho u(j) + \sum_{\max(n+1,i+1)}^{\min(L,i+R)} \rho u(j)\end{equation} 
According to Lemma \ref{passage_2}, we will distinguish between five sub-cases:\\
a) $i=L$\\
b) $0 \leq i \leq n-R$\\
c) $n+R+1 \leq i \leq L-1$\\
d) $n+1-R \leq i \leq n$\\
e) $n+1 \leq i \leq n+R$
\\
a) Proof of $u(i)=0$ for $i=L$:\\
if $i=L$, since $\forall j$ $p^n(j,L)=0$, then:
\begin{equation}u(L)=0
\end{equation}
b) Proof of $u(i)=0$ for $0 \leq i \leq n-R$:\\
We prove by induction that for all $0 \leq i < n+1-R$, $u(i)=0$.\\ 
In fact, if $0 \leq i < n+1-R$, then  $i < n-R < n$, $\min(n,i)=i$, and \ $\min(i+R,L) \leq i+R< n+1=\max(n+1,i+1)$.  Then:
\begin{equation} u(i)=\sum_{(i+1-R)^+}^i \rho u(j)\end{equation} 
for $i=0 \ u(0)=\rho u(0)$ i.e. $u(0)=0$ since $\rho<1$. \\
if $u(j)=0$ for all $j \leq i$, then:
\begin{align}
u(i+1)&=\sum_{(i+2-R)^+}^{i+1} \rho u(j)\\ 
&=\sum_{(i+2-R)^+}^{i} \rho u(j) +\rho u(i+1)\\
&=0+\rho u(i+1)\\
u(i+1)&=\rho u(i+1)
\end{align}
This implies that $u(i+1)=0$.
\\
c) Proof of $u(i)=0$ for $n+R+1 \leq i \leq L-1$:\\ 
If $i \geq n+R+1$ then $(i+1-R)^+=i+1-R > n=\min(n,i)$ and $\max(n+1,i+1)=i+1$. This implies that
 \begin{equation} u(i)= \sum_{i+1}^{\min(i+R,L)} \rho u(j)\end{equation} 
and we have $u(L)=0$.  \\
We now suppose that for all $k$ between $i$ and $L$: \ $u(k)=0$
then \begin{equation} u(i-1)= \sum_{i}^{\min(i-1+R,L)} \rho u(j)=0 \end{equation} 
We conclude the result.
\\
Next, we will provide a useful lemma which allows us to prove Proposition \ref{eq:stat.dist} for the cases (d) and (e). Before providing this lemma, we will give general expressions of $u(i)$ for these two cases.
\\ 
if $n+1-R \leq i \leq n$:\\
We have $\min(n,i)=i$, $\max(n+1,i+1)=n+1$, and $\min(L,i+R)=i+R$ (since $i+R \leq n+R < L-R+R=L$). Then: 
\begin{equation} u(i)=\sum_{(i+1-R)^+}^{i} \rho u(j) + \sum_{n+1}^{i+R} \rho u(j)\end{equation} 
We have $n-R+1 >0$ and $n-R+1 \geq i-R+1$, then $n-R+1 \geq (i+1-R)^+$. If $n-R+1 = (i+1-R)^+$, then we replace index $(i+1-R)^+$ by $n-R+1$ in the expression of $u(i)$. If $n-R+1 > (i+1-R)^+$, we know that for all $j$ less or equal to $n-R$, $u(j)=0$. Then, we can simplify the expression of $u(i)$ as follows:
\begin{align}
u(i)&=\sum_{n+1-R}^{i} \rho u(j) + \sum_{n+1}^{i+R} \rho u(j)\\
u(i)&=\sum_{n+1}^{i+R} \rho [u(j-R)+u(j)]
\end{align}
\\
if $n+1 \leq i \leq n+R$:\\
$i>R$, then $\max(i+1-R,0)=i+1-R$, $\min(n,i)=n$ and $\max(n+1,i+1)=i+1$.  Therefore:
 \begin{equation} u(i)=\sum_{i+1-R}^{n} \rho u(j) + \sum_{i+1}^{\min(L,i+R)} \rho u(j)\end{equation} 
 We have $i+R > n+R$, and $L > n+R$ because $n < L-R$, then  $\min(L,i+R) > n+R$. Therefore, given that $u(j)=0$ for all $j$ between $n+R+1$ and $\min(L,i+R)$, we can simplify the expression of $u(i)$ as follows:
\begin{align}
 u(i)&=\sum_{i+1-R}^{n} \rho u(j) + \sum_{i+1}^{n+R} \rho u(j)\\
 u(i)&=\sum_{i+1}^{n+R} \rho [u(j-R)+u(j)]
\end{align}
\\
\begin{Lemma}\label{ro_constante2}
 for $0 \leq k \leq R-1$, \begin{equation}u(n+R-k)+u(n-k)=\rho\end{equation}\\
 \end{Lemma}
\begin{proof}
\renewcommand{\qedsymbol}{$\blacksquare$} 
See appendix \ref{app:ro_constante2}
\end{proof}
Let us now prove the result for cases (d) and (e).
\\
d) Proof of $u(i)=\rho-(n-i)\rho^2$ for $n+1-R \leq i \leq n$:
\\
We prove by induction that, for $0\leq k \leq R-1$,  $u(n-k)=\rho - k\rho^2$ \\
For $k=0$: 
\begin{align}
u(n-0)&=\sum_{n+1}^{n+R} \rho [u(j-R)+u(j)]\\
&=\sum_{n+1-R}^{n} \rho u(j)+ \sum_{n+1}^{n+R} \rho u(j)\\
&=\rho \sum_{n+1-R}^{n+R} u(j)\\
u(n)&=\rho
\end{align}
\\
We suppose that the expression is true for some $k$, we prove it for $k+1$. 
\begin{align}
u(n-(k+1))&=\sum_{n+1}^{n-k-1+R} \rho (u(j-R)+u(j))\\
&=\sum_{n+1}^{n-k+R} \rho [u(j-R)+u(j)]-\rho[u(n-k)+u(n-k+R)]\\
&=u(n-k)-\rho[u(n-k)+u(n+R-k)]\\
&=\rho-k\rho^2-\rho[u(n-k)+u(n+R-k)]
\end{align}
Using Lemma \ref{ro_constante2}, $u(n-k)+u(n+R-k)=\rho$ , then 
\begin{align}
u(n-(k+1))&=\rho-k\rho^2-\rho(\rho)\\
&=\rho-k\rho^2-\rho^2\\
u(n-(k+1))&=\rho-(k+1)\rho^2
\end{align}
Thus we conclude that, for $k \in [0,R-1]$, $u(n-k)=\rho-k\rho^2$.\\
For $i \in [n+1-R,n]$, we replace $k \in [0,R-1]$ by $n-i$ ($n-i \in [0,R-1]$) and get:\begin{equation}u(i)=u(n-(n-i))=\rho-(n-i)\rho^2
\end{equation}
\\
e) Proof of $u(i)=\rho^2(n+R-i)$ for $n+1 \leq i \leq n+R$\\
We prove that, for $k \in [0,R-1]$, $u(n+R-k)=\rho^2 k$. 
From above, we have $u(n-k)=\rho - k\rho^2$,
and by using Lemma \ref{ro_constante2} we have:
\begin{align}
u(n+R-k)&=\rho-u(n-k)\\
&=\rho-(\rho-\rho^2 k)\\
u(n+R-k)&=\rho^2 k
\end{align}
For $i \in [n+1,n+R]$, by replacing $k \in [0,R-1]$ by $n+R-i$ ($n+R-i \in [0,n]$), we get:\begin{equation}u(i)=u(n+R-(n+R-i))=\rho^2(n+R-i)\end{equation}
This ends the proof of the second case.
\\
\\
\item Third case: $L-R \leq n<L$
\begin{equation}u(i)= \sum_{j=0}^n p^n(j,i)u(j)+ \sum_{j=n+1}^L p^n(j,i)u(j)\end{equation}
\begin{Lemma}\label{passage_3}
when $i<L$:\\
\begin{equation}
p^n(j,i)=\left\{
    \begin{array}{ll}
        \pi_{i-j} & if \ 0 \leq j \leq n \\
        \pi_{i-(j-R)} & if \ n+1 \leq j \leq L
    \end{array}
\right.
\end{equation}
when $i=L$:
\begin{equation}
p^n(j,L)=\left\{
    \begin{array}{ll}
        (R-L+j)\pi_{L-j} & if  \ 0 \leq j \leq n \\
        0 & if  \ n+1 \leq j \leq L \\
    \end{array}
\right.
\end{equation}
\end{Lemma}
This Lemma comes from Proposition \ref{prop:general_passage}.\\
So using Lemma \ref{passage_3}, and by definition of $\pi$:\\
if $i < L$:
\begin{equation} u(i)= \sum_{\max(i-R+1,0)}^{\min(i,n)} \rho u(j)+  \sum_{\max(n+1,i+1)}^{\min(L,i+R)} \rho u(j)\end{equation} 
if $i=L$:
\begin{equation} u(L)=\sum_{j=0}^n (R-L+j)\pi_{L-j} u(j)\end{equation} 
According to Lemma \ref{passage_3}, we will distinguish between five cases:\\
a) $0 \leq i \leq n-R$\\
b) $n+1 \leq i \leq L-1$\\
c) $n-R+1 \leq i \leq L-R-1$\\
d) $L-R \leq i \leq n$\\
e) $i=L$\\
\\
a) 	Proof of $u(i)=0$ for $0 \leq i \leq n-R$:\\
We prove by induction that, for $i \leq n-R$, $u(i)=0$. \\
Since $0 \leq i \leq n-R$, then $\min(i,n)=i$, $i+R \leq n < L$ and $\min(L,i+R)=i+R <n+1=\max(n+1,i+1)$.  Therefore:
\begin{equation} u(i)= \sum_{\max(i-R+1,0)}^{i} \rho u(j)\end{equation} 
for $i=0, u(0)=\rho u(0)=0$.\\ 
We consider that $u(j)=0$ for all $j$ between $0$ and $i$, we demonstrate that $u(i+1)=0$.
\begin{align}
u(i+1)&= \sum_{\max(i-R+2,0)}^{i+1} \rho u(j)\\ 
&= \sum_{\max(i-R+2,0)}^{i} \rho u(j)+\rho u(i+1)\\
&= 0+ \rho u(i+1)\\
u(i+1)&=\rho u(i+1)
\end{align}
This implies that:
\begin{equation} u(i+1)=\rho u(i+1)
\end{equation}  
Hence we prove that, for all $i \in [0,n-R]$, $u(i)=0$.
\\
 \\
We will provide a useful lemma which allows us to prove  Proposition \ref{eq:stat.dist} for  cases (b) and (c). Before giving this lemma, we will give general expressions of $u(i)$ for these two cases.\\  

if $n-R+1 \leq i \leq L-R-1$:\\
$i < L-R \leq n$, then $\min(i,n)=i$, $\max(n+1,i+1)=n+1$ and $\min(L,i+R)=i+R$. This implies that,
\begin{equation} u(i)= \sum_{(i-R+1)^+}^{i} \rho u(j)+  \sum_{n+1}^{i+R} \rho u(j)\end{equation} 
We have $n-R+1 >0$ and $n-R+1 > i-R+1$, which implies that  $n-R+1 > (i+1-R)^+$ and $n-R \geq (i+1-R)^+$. Since $u(j)=0$ for all  $j$ less or equal to $n-R$, we can simplify the expression of $u(i)$ as follows:
\begin{align}
    u(i)&= \sum_{n-R+1}^{i} \rho u(j)+  \sum_{n+1}^{i+R} \rho u(j)\\
    u(i)&= \sum_{n+1}^{i+R} \rho [u(j-R)+ u(j)]
\end{align}    
if $n+1 \leq i<L$:\\
We have $(i-R+1)^+=i-R+1$  (as $i\geq n+1 > R$), $\min(i,n)=n$,  $\max(n+1,i+1)=i+1$ and $\min(L,i+R)=L$  (due to $i+R > n+R \geq L-R+R = L$). Then: 
\begin{equation} u(i)= \sum_{i-R+1}^{n} \rho u(j)+  \sum_{i+1}^{L} \rho u(j)
\end{equation} 
 
\begin{Lemma}\label{ro}
 for $1 \leq k \leq L-n-1$, \begin{equation}u(n-R+k)+u(n+k)=\rho\end{equation}\\
\end{Lemma}

\begin{proof}
\renewcommand{\qedsymbol}{$\blacksquare$} 
See appendix \ref{app:ro}
\end{proof}

b) Proof of $u(i)=\rho-(i-n)\rho^2$ for $n+1 \leq i \leq L-1$:\\
We prove first that, for $1\leq k \leq L-n-1$, $u(n+k)=\rho-k\rho^2$.\\
In fact:\\
\begin{align}
u(n+k)&=\sum_{n+k-R+1}^{n} \rho u(j)+  \sum_{n+k+1}^{L} \rho u(j)\\
&= \rho- [\sum_{n-R+1}^{n+k-R} \rho u(j)+  \sum_{n+1}^{n+k} \rho u(j)]\\
&= \rho- [\sum_{1}^{k} \rho u(n-R+j)+  \sum_{1}^{k} \rho u(n+j)]\\
&= \rho- [\sum_{1}^{k} \rho [u(n-R+j)+u(n+j)]]
\end{align}
According to Lemma \ref{ro}, and given that $0 \leq k \leq L-n-1$, then for all $j \in [1,k]$, $u(n-R+j)+u(n+j)=\rho$, then: 
\begin{align}
u(n+k)&= \rho- [\sum_{1}^{k} \rho^2]\\
u(n+k)&= \rho- k\rho^2
\end{align}
Then for  $1\leq k \leq L-n-1$, $u(n+k)=\rho-k\rho^2$.\\
For $i \in [n+1,L-1]$, we replace $k \in [1,L-n-1]$ by $i-n$ ($i-n \in [1,L-n-1]$) and get:\begin{equation}u(i)=u(n+(i-n))=\rho-(i-n)\rho^2\end{equation}

c) Proof of $u(i)=\rho^2(R-n+i)$ for $n-R+1 \leq i \leq L-R-1$:\\
We need to prove that, for $k \in [1,L-n-1]$, $u(n-R+k)=\rho^2 k$\\ 
Given that $u(n+k)=\rho-\rho^2 k$ which is proved in case (d), and using Lemma \ref{ro}, then:
\begin{align}
u(n-R+k)&= \rho-u(n+k)\\
&=\rho-(\rho-\rho^2 k)\\
u(n-R+k)&=\rho^2 k
\end{align}
For $i \in [n-R+1,L-R-1]$, we replace  $k \in [1,L-n-1]$ by $R-n+i$ ($R-n+i \in [1,L-n-1]$) and get:\begin{equation}u(i)=u(n-R+(R-n+i))=\rho^2(R-n+i)\end{equation}
This ends the proof of case (c).

d) Proof of $u(i)=(1-\rho)^{n-i}\rho$ for $L-R \leq i \leq n$:\\
if $L-R \leq i \leq n$, $(i-R+1)^+=i-R+1$ because $i\geq L-R \geq R$, $\min(i,n)=i$, $\max(n+1,i+1)=n+1$ and $\min(L,i+R)=L$.  Then:
\begin{equation} u(i)= \sum_{i-R+1}^{i} \rho u(j)+  \sum_{n+1}^{L} \rho u(j)\end{equation}  
We have $n \geq i$, then $n-R+1 \geq i-R+1$. If $n-R+1 = i-R+1$, we replace $i-R+1$ by $n-R+1$ in the expression of $u(i)$. If $n-R+1 > i-R+1$, we know that, for all  $j$ less or equal to $n-R$, $u(j)=0$. We can then simplify the expression of $u(i)$ as follows: 
 \begin{equation} u(i)= \sum_{n-R+1}^{i} \rho u(j)+  \sum_{n+1}^{L} \rho u(j)\end{equation} 
In order to prove Proposition \ref{eq:stat.dist} for this case, we prove by induction that $u(n-k)=(1-\rho)^{k}\rho$ for $0\leq k \leq n-L+R$\\
For $k=0$:
\begin{align}
u(n)&=\sum_{n-R+1}^{n} \rho u(j)+  \sum_{n+1}^{L} \rho u(j)\\
&= \sum_{n-R+1}^{L} \rho u(j)\\ 
u(n)&= \rho  
\end{align}
We suppose it is true for $k$, we prove it for $k+1$:
\begin{align}
u(n-(k+1))&= \sum_{n-R+1}^{n-k-1} \rho u(j)+  \sum_{n+1}^{L} \rho u(j)\\
&= \sum_{n-R+1}^{n-k} \rho u(j)+  \sum_{n+1}^{L} \rho u(j)-\rho u(n-k)\\
&= u(n-k)-\rho u(n-k)\\
&= (1-\rho)^k\rho-\rho (1-\rho)^k\rho\\
&= (1-\rho)^k\rho(1-\rho)\\
u(n-(k+1))&= (1-\rho)^{k+1}\rho\\
\end{align}
Thus we conclude that, for $k \in [0,n-L+R]$, $u(n-k)=(1-\rho)^{k}\rho$.\\
For $i \in [L-R,n]$, we replace for $k \in [0,n-L+R]$ by $n-i$ ($n-i \in [0,n-L+R]$) and  get:\begin{equation}u(i)=u(n-(n-i))=(1-\rho)^{n-i}\rho\end{equation}
This proves the result.

e) Proof of $u(i)=(1-\rho)^{n-L+R+1}-\rho(L-1-n)$ for $i=L$:\\
\begin{align}
u(L)&=\sum_{j=0}^n (R-L+j)\pi_{L-j} u(j)\\
&=\sum_{j=L-R+1}^n (R-L+j)\rho u(j)
\end{align}
We replace $u(j)$ by its expression when $j \in [L-R+1,n]$ (it corresponds to the sub-case (d))
\begin{align}
u(L)&=\sum_{j=L-R+1}^n (R-L+j)[\rho (1-\rho)^{n-j}\rho]\\
u(L)&=\rho^2 \sum_{k=0}^{n-L+R-1} (R-L-k+n)(1-\rho)^{k}\\
u(L)&=(1-\rho)^{n-L+R+1}-\rho(L-1-n)
\end{align}
\item Fourth case: $n=L$\\
\begin{equation} u(i)= \sum_{j=0}^L p^L(j,i)u(j)
\end{equation} 
For $i \leq L-1$:\\
According to Proposition \ref{prop:general_passage}, we have:
\begin{equation} u(i)= \sum_{j=0}^L \pi_{i-j}u(j)
\end{equation} 
By definition of $\pi$, we get: 
\begin{equation} u(i)= \sum_{(i-R+1)^+}^{i} \rho u(j)
\end{equation} 
We prove by induction that for $0 \leq i <L$  $u(i)=0$\\
We have $u(0)=\rho u(0)=0$.\\
We suppose that $u(j)=0$ for all $0 \leq j \leq i$, then:
\begin{align}
u(i+1)&= \sum_{(i-R+2)^+}^{i+1} \rho u(j)\\
&= \sum_{(i-R+2)^+}^{i} \rho u(j)+\rho u(i+1)\\
&= 0+\rho u(i+1)\\
u(i+1)&= 0
\end{align}
Then, for all $i \in [0,L-1]$, $u(i)=0$.\\
Since $\sum_{j=0}^L u(j)=1$, we have  $u(L)=1-\sum_{j=0}^{L-1} u(j)=1-0=1$. \\

This ends the proof.
\end{enumerate}

\section{Proof of Lemma \ref{ro_constante}}\label{app:ro_constante}
\begin{align}
u(n-k)+u(n+R-k)&= \sum_{0}^{n-k} \rho u(j)+ \sum_{n+1}^{R} \rho u(j) +\sum_{R+1}^{n-k+R} \rho u(j)+ \sum_{n-k+1}^{n} \rho u(j) + \ \sum_{n+R-k+1}^{\min(n-k+2R,L)} \rho u(j))\\
u(n-k)+u(n+R-k)&=\rho \sum_{0}^{\min(2R+n-k,L)} u(j)
\end{align}
We know that $R>n$ and $n-k \geq 0$, which implies that $2R+n-k>n+R$ and $n+R < 2R \leq L$. and hence  $\min(2R+n-k,L)>n+R$.  Therefore, we get rid of all elements $u(j)$ such that $j \in [n+R+1,\min(2R+n-k,L)]$ since for all $j > n+R$, $u(j)=0$. Moreover $\sum_{0}^{n+R} u(j)=1$, consequently:
\begin{equation} u(k)+u(R+k)= \rho \sum_{0}^{n+R} u(j)=\rho\end{equation}

\section{Proof of Lemma \ref{ro_constante2}}\label{app:ro_constante2}
Since $n-R+1 \leq n-k \leq n$, and $n+1 \leq n+R-k \leq n+R$, then:\\
\begin{align}  
u(n-k)+u(n+R-k)&= \sum_{n+1-R}^{n-k} \rho u(j) + \sum_{n+1}^{n-k+R} \rho u(j) +\sum_{n-k+1}^{n} \rho u(j) + \sum_{n+R-k+1}^{n+R} \rho u(j)\\
&= \rho \sum_{n+1-R}^{n+R} u(j)
\end{align}
Given that $u(j)=0$ for $j \in [0,n-R] \cup [n+R+1,L]$, then  $\sum_{n+1-R}^{n+R} u(j)=1$. Consequently: 
\begin{equation} u(n-k)+u(n+R-k)= \rho\end{equation}  

\section{Proof of Lemma \ref{ro}}\label{app:ro}
Since $n-R+1 \leq n-R+k \leq L-R-1$, and $n+1 \leq n+k \leq L-1$, then:\\
\begin{align} 
u(n-R+k)+u(n+k)&= \sum_{n-R+1}^{n-R+k} \rho u(j)+  \sum_{n+1}^{n+k} \rho u(j)+ \sum_{n+k-R+1}^{n} \rho u(j)+  \sum_{n+k+1}^{L} \rho u(j)\\
u(n-R+k)+u(n+k)&=\rho \sum_{n-R+1}^{L} u(j)
\end{align} 
As we have demonstrated that $u(i)=0$ for $i \in [0,n-R]$, then $\sum_{n-R+1}^{L} u(j)=1$. Therefore, \\
\begin{equation} u(n-R+k)+u(n+k)= \rho\end{equation}

\section{Proof of Proposition \ref{prop_wi}}\label{app:prop_wi}
As mentioned previously in the paper, we denote $\sum_{q=0}^L au^n(q)q$ by $a_n$ and $\sum_{q=0}^{n} u^{n}(q)$ by $b_n$. Before proving the proposition, we give two useful lemmas.
\begin{Lemma}\label{inequalities}
Considering $a_{j-1}, a_j, a_{j+1}$ and $b_{j-1}, b_j, b_{j+1}$, such that $b_{j-1} < b_j < b_{j+1}$.\\

If $\frac{a_j-a_{j-1}}{b_{j}-b_{j-1}} \leq \frac{a_{j+1}-a_j}{b_{j+1}-b_j}$\\
Then: \begin{equation}\frac{a_{j}-a_{j-1}}{b_{j}-b_{j-1}} \leq \frac{a_{j+1}-a_{j-1}}{b_{j+1}-b_{j-1}} \leq \frac{a_{j+1}-a_j}{b_{j+1}-b_j}\end{equation}  

If $\frac{a_j-a_{j-1}}{b_{j}-b_{j-1}} \geq \frac{a_{j+1}-a_j}{b_{j+1}-b_j}$
Then: \begin{equation}\frac{a_{j}-a_{j-1}}{b_{j}-b_{j-1}} \geq \frac{a_{j+1}-a_{j-1}}{b_{j+1}-b_{j-1}} \geq \frac{a_{j+1}-a_j}{b_{j+1}-b_j}\end{equation}

If $\frac{a_{j}-a_{j-1}}{b_j-b_{j-1}} \leq \frac{a_{j+1}-a_{j-1}}{b_{j+1}-b_{j-1}}$ \\
Then: \begin{equation}\frac{a_{j}-a_{j-1}}{b_{j}-b_{j-1}} \leq \frac{a_{j+1}-a_{j-1}}{b_{j+1}-b_{j-1}} \leq \frac{a_{j+1}-a_j}{b_{j+1}-b_j}\end{equation}

If $\frac{a_j-a_{j-1}}{b_j-b_{j-1}} \geq \frac{a_{j+1}-a_{j-1}}{b_{j+1}-b_{j-1}}$
Then: \begin{equation}\frac{a_{j}-a_{j-1}}{b_{j}-b_{j-1}} \geq \frac{a_{j+1}-a_{j-1}}{b_{j+1}-b_{j-1}} \geq \frac{a_{j+1}-a_j}{b_{j+1}-b_j}\end{equation}

If $\frac{a_{j+1}-a_{j-1}}{b_{j+1}-b_{j-1}} \leq \frac{a_{j+1}-a_{j}}{b_{j+1}-b_{j}}$
Then: \begin{equation}\frac{a_{j}-a_{j-1}}{b_{j}-b_{j-1}} \leq \frac{a_{j+1}-a_{j-1}}{b_{j+1}-b_{j-1}} \leq \frac{a_{j+1}-a_j}{b_{j+1}-b_j}\end{equation} 

If $\frac{a_{j+1}-a_{j-1}}{b_{j+1}-b_{j-1}} \geq \frac{a_{j+1}-a_{j}}{b_{j+1}-b_{j}}$
Then: \begin{equation}\frac{a_{j}-a_{j-1}}{b_{j}-b_{j-1}} \geq \frac{a_{j+1}-a_{j-1}}{b_{j+1}-b_{j-1}} \geq \frac{a_{j+1}-a_j}{b_{j+1}-b_j}\end{equation} 

\end{Lemma}

\begin{proof}
\renewcommand{\qedsymbol}{$\blacksquare$}
See appendix \ref{app:inequalities}
\end{proof}

\begin{Lemma}\label{n_j_min}
The largest minimizer at step $j$ in algorithm 1 satisfies $n_j=\min\{k:b_k=b_{n_j}\}$
\end{Lemma}

\begin{proof}
\renewcommand{\qedsymbol}{$\blacksquare$}
See appendix \ref{app:n_j_min}.    
\end{proof}

We start by indexability:\\
 We consider $W_1 < W_2$ and prove that the optimal threshold $n_1$, when $W=W_1$,  is less than $n_2$ (when $W=W_2$). In fact if $n_1 \leq n_2$ and the threshold is $n_1$,  all states $[0,n_1]$, for which the optimal decision is passive action, are included in $[0,n_2]$. This implies the desired result $D(W_1) \subseteq D(W_2)$.\\
In order to prove that, we just need to prove that $b_{n_1} \leq b_{n_2}$ since $n_1 \leq n_2$ is equivalent to $b_{n_1} \leq b_{n_2}$ (due to increasiness of $b_n$).\\
We have according to equation~\eqref{eq:relaxed} and by definition of $n_1$ and $n_2$: \begin{equation}a_{n_1}-W_1 b_{n_1} \leq a_{n_2}-W_1 b_{n_2}\end{equation}
\begin{equation}a_{n_1}-W_2 b_{n_1} \geq a_{n_2}-W_2 b_{n_2}\end{equation}
This implies: \begin{equation}W_2 (b_{n_1}-b_{n_2}) \leq a_{n_1}-a_{n_2} \leq W_1 (b_{n_1}-b_{n_2})\end{equation}
Therefore: $(W_2-W_1) (b_{n_1}-b_{n_2}) \leq 0$. Since $W_2-W_1>0$,
hence: $b_{n_1} \leq b_{n_2}$,
then $n_1 \leq n_2$.\\
We conclude the indexability.
\\\\
For the Whittle's index expressions, we need to demonstrate that, for $k \in ]n_{j-1},n_j]$, $W_j=\min\{W,k \in D(W)\}$.\\
For that, we prove first that for $W < W_j$ then $k \notin D(W)$.\\
When $k > n_{j-1}$, $W < W_j$,
and $b_k \neq b_{n_{j-1}}$, then $W < W_j \leq \frac{a_k-a_{n_{j-1}}}{b_k-b_{n_{j-1}}}$, 
and $a_k-b_k W > a_{n_{j-1}}-b_{n_{j-1}} W$.\\
When $k > n_{j-1}$,  $W < W_j$ and $b_k = b_{n_{j-1}}$, then
given that $a_k > a_{n_{j-1}}$ we have
$a_k-b_k W > a_{n_{j-1}}-b_{n_{j-1}} W$\\ 
Hence we have proved that, for $W < W_j$ and $k > n_{j-1}$, $a_k-b_k W > a_{n_{j-1}}-b_{n_{j-1}} W$. That means at $W$ the optimal threshold is $n_{j-1}$ or even less. Therefore, for $k \in ]n_{j-1},n_j]$ where $k$ is necessary strictly higher  than the threshold,  the optimal action for $k$ is active action , i.e. $k \notin D(W)$.\\
There is still to prove that $k$ $\in$  $D(W_j)$.\\
For that, we prove that the threshold is at least $n_j$ when $W=W_j$. In other words, for all $k < n_j$, $a_k-b_k W_j \geq a_{n_j}-b_{n_j} W_j$.
We demonstrate this result by induction in $j$.\\
For $j=0$, we have for all $n$, $b_n > 0$, then $W_{0}$ is well defined.\\
$W_{0} \leq \frac{a_k-a_{-1}}{b_k}$ $\forall k \geq 0$.
Then for $0 \leq k < n_0$, according to Lemma \ref{n_j_min}, $b_k < b_{n_0}$.  Thus, by using Lemma \ref{inequalities} (fourth case), we can deduce that $\frac{a_{n_0}-a_{k}}{b_{n_0}-b_k} \leq W_0$. That means, for $k \in [-1,n_0[$, $\frac{a_{n_0}-a_{k}}{b_{n_0}-b_k} \leq W_0$, which implies that  $a_k-b_k W_0 \geq a_{n_0}-b_{n_0} W_0$.\\
We suppose at step $j$, $a_k-b_k W_j \geq a_{n_j}-b_{n_j} W_j$ i.e. $\frac{a_{n_j}-a_{k}}{b_{n_j}-b_k} \leq W_j$ for $k < n_j$ (this remains true since $b_k < b_{n_j}$ according to Lemma \ref{n_j_min}).\\
At $j+1$:\\
When $n_j \leq k < n_{j+1}$, then if $b_k \neq b_{n_j}$, $\frac{a_k-a_{n_j}}{b_k-b_{n_j}} \geq W_{j+1}$. Thus, by using Lemma \ref{inequalities} (fourth case), we get $\frac{a_{n_{j+1}}-a_{k}}{b_{n_{j+1}}-b_k} \leq W_{j+1}$ ($b_{n_j} < b_k < b_{n_{j+1}}$).
If $b_k=b_{n_j}$, $\frac{a_{n_{j+1}}-a_{k}}{b_{n_{j+1}}-b_k}=\frac{a_{n_{j+1}}-a_{k}}{b_{n_{j+1}}-b_{n_j}} \leq \frac{a_{n_{j+1}}-a_{n_j}}{b_{n_{j+1}}-b_{n_j}}=W_{j+1}$ since $a_k \geq a_{n_j}$.\\
When $k<n_j$, we have $\frac{a_{n_j}-a_{k}}{b_{n_j}-b_k} \leq W_j$ (induction assumption).  Using the definition of $n_j$ defined in Algorithm 1, we have $W_j < \frac{a_{n_{j+1}}-a_{n_{j-1}}}{b_{n_{j+1}}-b_{n_{j-1}}}$. Then according to Lemma \ref{inequalities} (third case), $W_j \leq W_{j+1}$( $b_{n_{j-1}} < b_{n_j} < b_{n_{j+1}}$).
Therefore $\frac{a_{n_j}-a_{k}}{b_{n_j}-b_k} \leq W_{j+1}$ and by using again Lemma \ref{inequalities} (first case), $\frac{a_{n_{j+1}}-a_{k}}{b_{n_{j+1}}-b_k} \leq W_{j+1}$. Therefore, for all $k \leq n_{j+1}$, $a_k-b_k W_{j+1} \geq a_{n_{j+1}}-b_{n_{j+1}} W_j$.\\
Thus, we have proved by induction that at any step $j$, for $k < n_j$, $a_k-b_k W_j \geq a_{n_j}-b_{n_j} W_j$.\\ 
Then when $W=W_j$, the threshold is at least $n_j$. This means that for $k \in ]n_{j-1},n_j]$, $k$ is less or equal than the threshold, which implies that the optimal decision at state $k$ is passive action, i.e. $k \in D(W_j)$.\\
As we have demonstrated that for $k \in ]n_{j-1},n_j]$ and $W < W_j$, $k \notin D(W)$ and $k \in D(W_j)$, then $W_j=\min\{W,k \in D(W)\}$. This concludes the proof.

\section{Proof of Lemma \ref{inequalities}}\label{app:inequalities}
We will just prove the first case. For the other cases, the proof is similar.\\
First case:$\frac{a_j-a_{j-1}}{b_{j}-b_{j-1}} \leq \frac{a_{j+1}-a_j}{b_{j+1}-b_j} \Longrightarrow \frac{a_{j}-a_{j-1}}{b_{j}-b_{j-1}} \leq \frac{a_{j+1}-a_{j-1}}{b_{j+1}-b_{j-1}} \leq \frac{a_{j+1}-a_j}{b_{j+1}-b_j}$:\\
For the LHS inequality:
\begin{align}
\frac{a_{j+1}-a_{j-1}}{b_{j+1}-b_{j-1}}& = \frac{a_{j+1}-a_j}{b_{j+1}-b_{j-1}}+\frac{a_{j}-a_{j-1}}{b_{j+1}-b_{j-1}}\\
&\geq \frac{(a_{j}-a_{j-1})(b_{j+1}-b_j)}{(b_{j}-b_{j-1})(b_{j+1}-b_{j-1})}+\frac{a_{j}-a_{j-1}}{b_{j+1}-b_{j-1}}
\end{align}
The inequality above comes from the fact that  $b_{j-1} < b_j < b_{j+1}$ and $\frac{a_j-a_{j-1}}{b_{j}-b_{j-1}} \leq \frac{a_{j+1}-a_j}{b_{j+1}-b_j}$\\
Then
\begin{align}
\frac{a_{j+1}-a_{j-1}}{b_{j+1}-b_{j-1}}& \geq \frac{a_{j}-a_{j-1}}{b_{j}-b_{j-1}} [\frac{b_{j+1}-b_j+b_j-b_{j-1}}{b_{j+1}-b_{j-1}}]\\
&= \frac{a_{j}-a_{j-1}}{b_{j}-b_{j-1}}
\end{align}
For the RHS inequality:
\begin{align}
\frac{a_{j+1}-a_{j-1}}{b_{j+1}-b_{j-1}}&= \frac{a_{j+1}-a_j}{b_{j+1}-b_{j-1}}+\frac{a_{j}-a_{j-1}}{b_{j+1}-b_{j-1}}\\
&\leq \frac{a_{j+1}-a_{j}}{b_{j+1}-b_{j-1}}+\frac{(a_{j+1}-a_{j})(b_j-b_{j-1})}{(b_{j+1}-b_j)(b_{j+1}-b_{j-1})}
\end{align}
where the above inequality comes from the fact that  $b_{j-1} < b_j < b_{j+1}$ and $\frac{a_j-a_{j-1}}{b_{j}-b_{j-1}} \leq \frac{a_{j+1}-a_j}{b_{j+1}-b_j}$\\
Then
\begin{align}
\frac{a_{j+1}-a_{j-1}}{b_{j+1}-b_{j-1}}& \leq \frac{a_{j+1}-a_{j}}{b_{j+1}-b_{j}} [\frac{b_{j+1}-b_{j}+b_{j}-b_{j-1}}{b_{j+1}-b_{j-1}}]\\
&= \frac{a_{j+1}-a_{j}}{b_{j+1}-b_{j}}
\end{align}

\section{Proof of Lemma \ref{n_j_min}}\label{app:n_j_min}
We consider $i$ such that $b_i=b_{n_j}$ and we prove that $n_j \leq i$:\\
By construction of $n_j$, $b_{n_{j-1}} \neq b_{n_j}$ and $n_{j-1} < n_j$. Hence, by increasiness of $b_k$, $b_{n_j} \geq b_{n_{j-1}}$.\\   
Therefore $b_i=b_{n_j}>b_{n_{j-1}}$, and  $i>n_{j-1}$. Consequently, according to definition of $n_j$:
\begin{equation}\frac{a_{n_j}-a_{n_{j-1}}}{b_{n_j}-b_{n_{j-1}}} \leq \frac{a_i-a_{n_{j-1}}}{b_{i}-b_{n_{j-1}}}\end{equation}
\begin{equation}\frac{a_{n_j}-a_{n_{j-1}}}{b_{n_j}-b_{n_{j-1}}} \leq \frac{a_i-a_{n_{j-1}}}{b_{n_j}-b_{n_{j-1}}}\end{equation}
This implies that $a_{n_j} \leq a_i$.\\
If $i<n_j$, as $b_i=b_{n_j}$, then  $a_i < a_{n_j}$ which contradicts with $a_{n_j} \leq a_i$.\\
Therefore $n_j \leq i$. This concludes the proof.

\section{Proof of Lemma \ref{temps_pass_str_croiss_less_R}}\label{app:temps_pass_str_croiss_less_R}
For $n \in [-1,R-2]$ \\
\begin{align}
\sum_{q=0}^{n+1} u^{n+1}(q)-\sum_{q=0}^n u^n(q)&=(1-\frac{n+1}{2R})(\frac{n+2}{R})-(1-\frac{n}{2R})(\frac{n+1}{R})\\
&=\frac{R-1-n}{R^2}\\
&>0
\end{align}

\section{Proof of Lemma \ref{temps_pass_str_croiss}}\label{app:temps_pass_str_croiss}
We introduce a useful Lemma:
\begin{Lemma}\label{inequality}
we have the inequality: for all $x \in ]0,1[$ \begin{equation}x+\ln(1-x)(1-x) > 0\end{equation}
\end{Lemma}
\begin{proof}
\renewcommand{\qedsymbol}{$\blacksquare$}
See appendix \ref{app:inequality}.
\end{proof}
We note that $R \geq 2$, then $\rho \in ]0,1[$.\\
We denote the function $h(n)=\sum_{q=0}^n u^n(q)=\frac{\rho^2}{2}(L-1-n)(L-n) + 1 -(1-\rho)^{n-L+R+1}$. We give the first derivative and the second derivative of $h$:
\begin{align}
h^{'}(n)&=\frac{\rho^2}{2}(-2L+1+2n)-\ln(1-\rho)(1-\rho)^{n-L+R+1}\\
h^{''}(n)&=\rho^2-(\ln(1-\rho))^2(1-\rho)^{n-L+R+1}
\end{align}
For $n \in [L-R+1,L-1]$, $(1-\rho)^{n-L+R+1}$ is decreasing in $n$, then \begin{equation}h^{''}(n) \geq \rho^2-(\ln(1-\rho))^2(1-\rho)^2\end{equation}
Using lemma \ref{inequality}, \begin{equation}\rho > -\ln(1-\rho)(1-\rho)\end{equation} then \begin{equation}\rho^2 > (\ln(1-\rho))^2(1-\rho)^2\end{equation}
Therefore
\begin{equation}h^{''}(n) \geq \rho^2-(\ln(1-\rho))^2(1-\rho)^2 > 0\end{equation} i.e. $h{'}$ is strictly increasing function in $n$. \\
We have $h{'}(L-R+1)=\frac{3\rho^2}{2}-\rho -\ln(1-\rho)(1-\rho)^2$. In order to prove the positivity of $h'$,
we introduce the function \begin{equation}r(x)=\frac{3x^2}{2}-x -\ln(1-x)(1-x)^2\end{equation}
$r'(x)=2(x+\ln(1-x)(1-x)) > 0$ (according to Lemma \ref{inequality}), which means $r$ is strictly increasing in $[0,1[$. Hence, for all $x \in ]0,1[$, $r(x) > r(0)=0$.\\
Then:\begin{equation}h'(L-R+1)=\frac{3\rho^2}{2}-\rho -\ln(1-\rho)(1-\rho)^2 > 0\end{equation} Since $h'$ is increasing function in $n$, then: \begin{equation}h'(n) \geq h'(L-R+1)> 0\end{equation}
Therefore $h$ is strictly increasing in $n$.
This concludes the proof.

\section{Proof of Lemma \ref{inequality}}\label{app:inequality}
We consider the function $v(x)= x+ \ln(1-x)(1-x)$ in $[0,1[$ \\
the first derivative:  $v'(x)= -\ln(1-x) > 0$ for all $x\in ]0,1[$, we have
$v(0)=0$, then for all $x \in ]0,1[ \ v(x) > v(0)=0$,
which concludes the result.

\section{Proof of Lemma \ref{average_cost_strict_increasing_n_greater_L_moins_R}}\label{app:average_cost_strict_increasing_n_greater_L_moins_R}
For $n \in [L-R,L-2]$, we have:
\begin{equation}\sum_{q=0}^L au^{n+1}(q)q-\sum_{q=0}^L au^{n}(q)q=1-2(1-\rho)^{n-L+1+R} +2L\rho-2n\rho-2\rho\end{equation}
If we denote the function $p$ as:
\begin{align}
p(n)&=1-2(1-\rho)^{n-L+1+R} +2L\rho-2n\rho-2\rho\\
p^{''}(n)&=-2(\ln(1-\rho))^2(1-\rho)^{n-L+1+R}
\end{align}
 
Hence, as $p^{''}(n) \leq 0$, $p$ is concave, that is $p$ is quasi-concave in $[L-R,L-1]$ , then:
\begin{align}
p(n)&\geq \min(p(L-R),p(L-1)]\\
p(L-R)&=1-2(1-\rho)+2-2\rho=1 >0\\
p(L-1)&=1-2(1-\rho)^R
\end{align}
As $(1-\rho)^R \leq \exp(-1)$ (with $ \exp $ the exponential function) for all $R \geq 2$, then:
\begin{equation}p(L-1) \geq 1-2 \exp(-1) > 0\end{equation}
Thus $p(n) > 0$ in $[L-R,L-1]$. %then as in $[L-R,L-2]$.\\
Hence, for $n \in [L-R,L-2]$  \begin{equation}\sum_{q=0}^L au^{n+1}(q)q-\sum_{q=0}^L au^{n}(q)q > 0\end{equation}

\section{Proof of Lemma \ref{intersection_point_Whittle_index}}\label{app:intersection_point_Whittle_index}
At $W=x_{i,j}$, $y^i(W)=y^j(W)$, i.e.:
\begin{align}
\sum_{q=0}^L au^i(q)q \ - \ W \sum_{q=0}^i u^i(q)&=\sum_{q=0}^L au^j(q)q \ - \ W \sum_{q=0}^j u^j(q)\\
\sum_{q=0}^L au^i(q)q \ -\sum_{q=0}^L au^i(q)q  &=W \sum_{q=0}^i u^i(q) - \ W \sum_{q=0}^j u^j(q)\\
\sum_{q=0}^L au^i(q)q \ -\sum_{q=0}^L au^i(q)q  &=W [\sum_{q=0}^i u^i(q) -\sum_{q=0}^j u^j(q)]
\end{align}
Hence
\begin{equation}W=\frac{\sum_{q=0}^L au^i(q)q-\sum_{q=0}^{L} au^{j}(q)}{\sum_{q=0}^i u^i(q)-\sum_{q=0}^{j} u^{j}(q)}\end{equation}

\section{Proof of Lemma \ref{f_strictly_increasing}}\label{app:f_strictly_increasing}
We start by giving a useful lemma.

\begin{Lemma}\label{increasiness_Whittle_index_less_R}
$w_n$ is strictly increasing in $n \in [0,R-1]$.
\end{Lemma}

\begin{IEEEproof}
\renewcommand{\qedsymbol}{$\blacksquare$}
for $n \in [0,R-2]$:
\begin{equation} w_{n+1}-w_n=\frac{aR^2}{(R-n)(R-n-1)} > 0.\end{equation}  
\end{IEEEproof}

Let us first consider the interval $[0,R-1]$.\\
We have:
\begin{equation} f'(n)=(w_n)^{'}[1-(1-\frac{n}{2R})\frac{n+1}{R}]+w_n [1-((1-\frac{n}{2R})\frac{n+1}{R})]^{'}+[a(\frac{R-1}{2}+\frac{n(n+1)}{2R})]^{'}\end{equation} 
First, we deal with the first term $(w_n)^{'}[1-(1-\frac{n}{2R})\frac{n+1}{R}]$:\\
According to Lemma \ref{increasiness_Whittle_index_less_R}, $(w_n)^{'}$ is positive since $w_n$ is increasing in $n$, and $1-\sum_{q=0}^n u^n(q) =1-(1-\frac{n}{2R})\frac{n+1}{R}$ is strictly positive since $\sum_{q=0}^n u^n(q) < 1$ for $n \leq R-1<L$. Then, $(w_n)^{'}[1-(1-\frac{n}{2R})\frac{n+1}{R}] \geq 0$, for $n \in [0,R-1]$.\\
For the second term, we have:
\begin{equation} w_n[1-((1-\frac{n}{2R})\frac{n+1}{R})]^{'}=a\frac{2n^2R-2R^2n+Rn}{(R-n)(2R^2)}\end{equation} 
For the third term $[a(\frac{R-1}{2}+\frac{n(n+1)}{2R})]^{'}=a \frac{2n+1}{2R}.$\\
Adding the second term to the third term, we get:
\begin{align}
w_n [1-((1-\frac{n}{2R})\frac{n+1}{R})]^{'}+[a(\frac{R-1}{2}+\frac{n(n+1)}{2R})]^{'}=&a\frac{2n^2R-2R^2n+Rn}{(R-n)(2R^2)}+a \frac{2n+1}{2R}\\ 
=&\frac{a}{2R(R-n)} > 0
\end{align} 
So $f$ is strictly increasing in $[0,R-1]$\\
For $n=-1$, $f(-1)=0 < f(0)=\frac{a(R-1)}{2}$, and $f(R-1)< +\infty$
then $f$ in strictly increasing in $[-1,R]$.\\

\section{Proof of Theorem \ref{w.i}}\label{app:w.i}
In order to prove the theorem, we introduce the following useful lemmas.
\begin{Lemma} \label{th.triangulaire}
For any numerical sequence: $-1 \leq i_{-1}<i_0<i_1<....<i_M \leq L$, such that for any $k \in [0,M-1]$, $b_{i_{k-1}} < b_{i_k} < b_{i_{k+1}}$ and 
\begin{equation}  \frac{a_{i_k}-a_{i_{k-1}}}{b_{i_k}-b_{i_{k-1}}} < \frac{a_{i_{k+1}}-a_{i_{k}}}{b_{i_{k+1}}-b_{i_k}}\end{equation} 
Then for any $k \in [0,M-1]$, we have for each $k<s \leq M$:
\begin{equation} \frac{a_{i_{s}}-a_{i_{k-1}}}{b_{i_s}-b_{i_{k-1}}} > \frac{a_{i_k}-a_{i_{k-1}}}{b_{i_k}-b_{i_{k-1}}} \end{equation}         
\end{Lemma}

\begin{IEEEproof}
\renewcommand{\qedsymbol}{$\blacksquare$} 
We fix certain $k \in [0,M-1]$, we prove  the result by induction:\\
for $s=k+1$
\begin{align}
\frac{a_{i_{k+1}}-a_{i_{k-1}}}{b_{i_{k+1}}-b_{i_{k-1}}}&=\frac{a_{i_{k+1}}-a_{i_{k-1}}-a_{i_k}+a_{i_{k}}}{b_{i_{k+1}}-b_{i_{k-1}}}\\ 
&= \frac{a_{i_{k+1}}-a_{i_{k}}}{b_{i_{k+1}}-b_{i_{k-1}}}+\frac{a_{i_k}-a_{i_{k-1}}}{b_{i_{k+1}}-b_{i_{k-1}}}\\
&> \frac{(a_{i_k}-a_{i_{k-1}})(b_{i_{k+1}}-b_{i_k})}{(b_{i_k}-b_{i_{k-1}})(b_{i_{k+1}}-b_{i_{k-1}})}+\frac{(a_{i_k}-a_{i_{k-1}})(b_{i_k}-b_{i_{k-1}})}{(b_{i_k}-b_{i_{k-1}})(b_{i_{k+1}}-b_{i_{k-1}})}
\end{align}
where the strict inequality comes from the lemma's assumptions. We then have:
\begin{align}
\frac{a_{i_{k+1}}-a_{i_{k-1}}}{b_{i_{k+1}}-b_{i_{k-1}}}& > \frac{a_{i_k}-a_{i_{k-1}}}{b_{i_k}-b_{i_{k-1}}} [\frac{b_{i_{k+1}}-b_{i_k}}{b_{i_{k+1}}-b_{i_{k-1}}}+\frac{b_{i_k}-b_{i_{k-1}}}{b_{i_{k+1}}-b_{i_{k-1}}}]\\
&=\frac{a_{i_k}-a_{i_{k-1}}}{b_{i_k}-b_{i_{k-1}}}
\end{align}
By induction, we consider that the above inequality is true for certain $s$ strictly higher than $k$. The inequality below is then verified for $s+1$:
\begin{align}
\frac{a_{i_{s+1}}-a_{i_{k-1}}}{b_{i_{s+1}}-b_{i_{k-1}}}&=\frac{a_{i_{s+1}}-a_{i_{k-1}}-a_{i_{s}}+a_{i_{s}}}{b_{i_{s+1}}-b_{i_{k-1}}}\\ 
&= \frac{a_{i_{s+1}}-a_{i_{s}}}{b_{i_{s+1}}-b_{i_{k-1}}}+\frac{a_{i_{s}}-a_{i_{k-1}}}{b_{i_{s+1}}-b_{i_{k-1}}}\\
&> \frac{(a_{i_k}-a_{i_{k-1}})(b_{i_{s+1}}-b_{i_s})}{(b_{i_k}-b_{i_{k-1}})(b_{i_{s+1}}-b_{i_{k-1}})}+\frac{(a_{i_k}-a_{i_{k-1}})(b_{i_s}-b_{i_{k-1}})}{(b_{i_k}-b_{i_{k-1}})(b_{i_{s+1}}-b_{i_{k-1}})}\\
&=\frac{a_{i_k}-a_{i_{k-1}}}{b_{i_k}-b_{i_{k-1}}} [\frac{b_{i_{s+1}}-b_{i_s}}{b_{i_{s+1}}-b_{i_{k-1}}}+ \frac{b_{i_s}-b_{i_{k-1}}}{b_{i_{s+1}}-b_{i_{k-1}}}]\\
&=\frac{a_{i_k}-a_{i_{k-1}}}{b_{i_k}-b_{i_{k-1}}}.
\end{align}
So the inequality is also true for $s+1$. This concludes the proof of the lemma.

\end{IEEEproof}
\begin{Lemma}\label{L_d_less_W_d_plus_1}
If  $L \leq \frac{f(d+1)}{a}$, then  $\frac{a_L-a_{d}}{b_L-b_{d}} \leq \frac{a_{d+1}-a_{d}}{b_{d+1}-b_{d}}$
\end{Lemma}
\begin{IEEEproof}
\renewcommand{\qedsymbol}{$\blacksquare$}
This lemma is an immediate application of Lemma \ref{inequalities}.\\
In fact when $L \leq \frac{f(d+1)}{a}$ it implies that $\frac{a_L-a_{d+1}}{b_L-b_{d+1}} \leq \frac{a_{d+1}-a_{d}}{b_{d+1}-b_{d}}$\\
Then according to the second case in Lemma \ref{inequalities}, we have directly:
\begin{equation}\frac{a_L-a_{d}}{b_L-b_{d}} \leq \frac{a_{d+1}-a_{d}}{b_{d+1}-b_{d}}\end{equation}
\end{IEEEproof}

\begin{Lemma}\label{x_L_R_less_x_n_R}
The intersection points $x_{L,R}$ and $x_{n,R}$ satisfy $x_{L,R} \leq x_{n,R}$, when $n \in [L-R+2,L-1]$.
\end{Lemma}

\begin{IEEEproof}
\renewcommand{\qedsymbol}{$\blacksquare$}
We have:
\begin{equation}x_{L,R}=\frac{2R(L-1)}{R-1} -2R\end{equation}
\begin{equation}x_{n,R}=\frac{n-R}{\frac{\rho^2}{2}(L-1-n)(L-n) + 1 -(1-\rho)^{n-L+R+1}-\frac{1}{2}-\frac{1}{2R}}-2R\end{equation}
\begin{equation}x_{n,R}-x_{L,R}= \frac{n-R}{\frac{\rho^2}{2}(L-1-n)(L-n) + 1 -(1-\rho)^{n-L+R+1}-\frac{1}{2}-\frac{1}{2R}}-\frac{2R(L-1)}{R-1}\end{equation}
\begin{equation} =\frac{(n-R)(R-1)-2R(L-1)(\frac{\rho^2}{2}(L-1-n)(L-n) + 1 -(1-\rho)^{n-L+R+1}-\frac{1}{2}-\frac{1}{2R})}{(R-1)(\frac{\rho^2}{2}(L-1-n)(L-n) + 1 -(1-\rho)^{n-L+R+1}-\frac{1}{2}-\frac{1}{2R})}\end{equation}
The denominator is greater than $0$ since $R > 1$, and $h(n)=\frac{\rho^2}{2}(L-1-n)(L-n) + 1 -(1-\rho)^{n-L+R+1} > h(L-R+1)=\frac{1}{2}+\frac{1}{2R}$ for $n \in [L-R+2,L-1]$ (using Lemma \ref{temps_pass_str_croiss}). \\
We consider the following function (which is equal to the numerator): \begin{equation}p(x)=(x-R)(R-1)-2R(L-1)(\frac{\rho^2}{2}(L-1-x)(L-x) + 1 -(1-\rho)^{x-L+R+1}-\frac{1}{2}-\frac{1}{2R})\end{equation}
The function $p$ is concave in the interval $[L-R+1,L-1]$ as $p{''}$ is negative. Then,  $p$ is quasi-concave in this interval and we have that  $\ p(x) \geq \min(p(L-R+1),p(L-1))$ for all $x \in [L-R+1,L-1]$, where 
\begin{equation}p(L-R+1)=(L-2R+1)(R-1) \geq 0\end{equation}
and \begin{equation}p(L-1)=2R(L-1)(1-\rho)^R+R-R^2 \geq 0\end{equation}
where the last inequality is due to the following analysis. First we use the fact that $(1-\rho)^R \geq 1/4$ for all $R \geq 2$ then \begin{equation}2R(L-1)(1-\rho)^R \geq \frac{2R(L-1)}{4}\end{equation}
\begin{equation}2R(L-1)(1-\rho)^R+R-R^2 \geq \frac{2R(L-1)}{4} +R-R^2\end{equation}
We have $L-1 \geq 2R-1$, then:
\begin{align}
2R(L-1)(1-\rho)^R+R-R^2 \geq& \frac{2R(2R-1)}{4}+R-R^2\\
\geq& R^2-\frac{R}{2}+R-R^2\\
\geq& \frac{R}{2} \geq 0
\end{align}
From all the analysis above, we conclude that for all $n \in [L-R+1,L-1]$ \ $p(n)\geq 0$. This is also true for $n \in [L-R+2,L-1]$.  
Hence, the numerator and denominator of $x_{n,R}-x_{L,R}$ are positive , which concludes the proof.
\end{IEEEproof}

\begin{Lemma}\label{L_d_less_n_d}
For any $d \in [0,R-1]$, $x_{L,d} \leq x_{n,d}$ for any $n \in [L-R+2,L-1]$.
\end{Lemma}

\begin{IEEEproof}
\renewcommand{\qedsymbol}{$\blacksquare$}

We start by proving that $\frac{a_L-a_{R-1}}{b_L-b_{R-1}} \leq \frac{a_n-a_{R-1}}{b_n-b_{R-1}}$.
We have: \begin{equation} \frac{a_L-a_{R-1}}{b_L-b_{R-1}}=\frac{a_L-a_{R}}{b_L-b_{R-1}}+\frac{a_R-a_{R-1}}{b_L-b_{R-1}} \end{equation} 
Since $b_R=b_{R-1}$(see the expression of average passive time when $n \in [R-1,L-R+1]$), then:
\begin{equation} \frac{a_L-a_{R-1}}{b_L-b_{R-1}}=\frac{a_L-a_{R}}{b_L-b_{R}}+\frac{a_R-a_{R-1}}{b_L-b_{R-1}} \end{equation} 
As we have already proved in Lemma \ref{x_L_R_less_x_n_R} that: $\frac{a_L-a_R}{b_L-b_R} \leq \frac{a_n-a_R}{b_n-b_R}$.
Hence:
\begin{equation} \frac{a_L-a_{R-1}}{b_L-b_{R-1}} \leq \frac{a_n-a_{R}}{b_n-b_{R-1}}+\frac{a_R-a_{R-1}}{b_L-b_{R-1}} \end{equation} 
Since $b_L > b_n$, hence:
\begin{equation} \frac{a_L-a_{R-1}}{b_L-b_{R-1}} \leq \frac{a_n-a_{R}}{b_n-b_{R-1}}+\frac{a_R-a_{R-1}}{b_n-b_{R-1}} \end{equation} 
\begin{equation} = \frac{a_n-a_{R-1}}{b_n-b_{R-1}} \end{equation} 
Thus: \begin{equation} \frac{a_L-a_{R-1}}{b_L-b_{R-1}} \leq \frac{a_n-a_{R-1}}{b_n-b_{R-1}} \end{equation} 

If $d=R-1$, the proof is direct result from the inequality above.\\ 
If $d < R-1$:\\
Given that $\frac{a_L-a_{R-1}}{b_L-b_{R-1}} \leq \frac{a_n-a_{R-1}}{b_n-b_{R-1}}$, then applying lemma \ref{inequalities} fourth case, we deduce: 
\begin{equation}\label{eq:inequality_1}
\frac{a_{L}-a_{n}}{b_{L}-b_{n}} \leq \frac{a_{L}-a_{R-1}}{b_{L}-b_{R-1}} \leq \frac{a_{n}-a_{R-1}}{b_{n}-b_{R-1}}
\end{equation}

Now we prove that: \begin{equation} \frac{a_{L}-a_{R-1}}{b_{L}-b_{R-1}} \leq \frac{a_{L}-a_{d}}{b_{L}-b_{d}} \end{equation}     
Given that $L \leq f(d+1)/a$: $\frac{a_{L}-a_{d+1}}{b_{L}-b_{d+1}} \leq \frac{a_{d+1}-a_{d}}{b_{d+1}-b_{d}}$.\\
Hence applying lemma \ref{L_d_less_W_d_plus_1}:  \begin{equation} \frac{a_L-a_{d}}{b_L-b_{d}}  \leq \frac{a_{d+1}-a_{d}}{b_{d+1}-b_{d}} \end{equation} 
According to Lemma \ref{th.triangulaire}, since $w_{d+1} < \cdots < w_{R-1}$, thus: \begin{equation} \frac{a_{d+1}-a_{d}}{b_{d+1}-b_{d}} \leq \frac{a_{R-1}-a_{d}}{b_{R-1}-b_{d}} \end{equation}  
Then: 
\begin{equation} \frac{a_L-a_{d}}{b_L-b_{d}} \leq \frac{a_{R-1}-a_{d}}{b_{R-1}-b_{d}} \end{equation} 
Given that $\frac{a_L-a_{d}}{b_L-b_{d}} \leq \frac{a_{R-1}-a_{d}}{b_{R-1}-b_{d}}$ and applying Lemma \ref{inequalities} (fourth case), then:

\begin{equation}\label{eq:inequality_2}
\frac{a_{L}-a_{R-1}}{b_{L}-b_{R-1}} \leq \frac{a_{L}-a_{d}}{b_{L}-b_{d}} \leq \frac{a_{R-1}-a_{d}}{b_{R-1}-b_{d}}
\end{equation}  

Combining \eqref{eq:inequality_1} and \eqref{eq:inequality_2}, we conclude:\\
\begin{equation} \frac{a_{L}-a_{n}}{b_{L}-b_{n}} \leq \frac{a_{L}-a_{R-1}}{b_{L}-b_{R-1}} \leq \frac{a_{L}-a_{d}}{b_{L}-b_{d}} \end{equation} 
\begin{equation} \frac{a_{L}-a_{n}}{b_{L}-b_{n}} \leq \frac{a_{L}-a_{d}}{b_{L}-b_{d}} \end{equation} 

Given this result and applying lemma\ref{inequalities} sixth case, we get our result:
\begin{equation} \frac{a_{n}-a_{d}}{b_{n}-b_{d}} \geq \frac{a_{L}-a_{d}}{b_{L}-b_{d}} \end{equation} 
Hence $x_{n,d} \geq x_{L,d}$.\\
This concludes the proof.
\end{IEEEproof}

Now, we can prove the proposition.\\
Referring to the algorithm 1 that allows us to obtain the Whittle indices, we denote by $j$  the step $j$ described in the algorithm.\\
For $0\leq j \leq d \leq R-1$ \\
We prove that for all $n \in [j+1,L]$, $\frac{a_n-a_{j-1}}{b_n-b_{j-1}} > \frac{a_j-a_{j-1}}{b_j-b_{j-1}}$\\
We study four cases: 
\\

1) $n \in [j+1,R-1]$:\\
Using lemma \ref{increasiness_Whittle_index_less_R}, $w_j < w_{j+1} < ....< w_{R-1}$,
therefore considering the set of element $\{j-1,j,j+1,...,R-1\}$, we can apply lemma \ref{th.triangulaire}, since $\frac{a_k-a_{k-1}}{b_k-b_{k-1}} < \frac{a_{k+1}-a_{k}}{b_{k+1}-b_{k}}$ for all $k \in [j,R-2]$. \\
So for all $n \in [j+1,R-1]$, $\frac{a_n-a_{j-1}}{b_n-b_{j-1}} > \frac{a_j-a_{j-1}}{b_j-b_{j-1}}$\\
\\

2) $n \in [R,L-R+1]$:\\
There are two cases:\\
a) $j=R-1$:\\
We have $b_n=b_{R-1}=b_j$, then $a_n > a_{j}$. Hence, $\frac{a_n-a_{j-1}}{b_n-b_{j-1}} > \frac{a_j-a_{j-1}}{b_j-b_{j-1}}$\\   
b) $j<R-1$:\\
$b_n=b_{R-1}$,and $a_n > a_{R-1}$, then $\frac{a_n-a_{R-2}}{b_n-b_{R-2}} > \frac{a_{R-1}-a_{R-2}}{b_{R-1}-b_{R-2}} $. 
Therefore, by considering the set $\{j-1,j,j+1,....,R-2,n\}$, we have $\frac{a_n-a_{R-2}}{b_n-b_{R-2}} > \frac{a_{R-2}-a_{R-3}}{b_{R-2}-b_{R-3}}=w_{R-2}> \cdots > w_j$.\\
Thus, we can apply  Lemma \ref{th.triangulaire} and get  $\frac{a_n-a_{j-1}}{b_n-b_{j-1}} > \frac{a_j-a_{j-1}}{b_j-b_{j-1}}$.\\
\\

3) $n \in [L-R+2,L-1]$:\\
Using Lemma \ref{L_d_less_n_d}, we have $\frac{a_n-a_{d}}{b_n-b_{d}} \geq \frac{a_L-a_{d}}{b_L-b_{d}}$.\\
Given that $\frac{f(d)}{a}<L$, that means $\frac{a_L-a_{d}}{b_L-b_{d}} > \frac{a_{d}-a_{d-1}}{b_{d}-b_{d-1}}$\\
So considering the set $\{j-1,j,...d, n\}$, we have $\frac{a_n-a_{d}}{b_n-b_{d}} >\frac{a_{d}-a_{d-1}}{b_{d}-b_{d-1}}=w_{d}>...>w_j$.\\
Then we can apply Lemma \ref{th.triangulaire} and obtain  $\frac{a_{n}-a_{j-1}}{b_{n}-b_{j-1}}>\frac{a_{j}-a_{j-1}}{b_{j}-b_{j-1}}$.\\
\\

4) $n=L$\\
We have $\frac{a_L-a_{d}}{b_L-b_{d}} > \frac{a_{d}-a_{d-1}}{b_{d}-b_{d-1}}=w_{d}> \cdots >w_j$.\\
Then, applying Lemma \ref{th.triangulaire}, $\frac{a_L-a_{j-1}}{b_L-b_{j-1}} > \frac{a_{j}-a_{j-1}}{b_{j}-b_{j-1}}$.\\

Therefore, the largest minimizer at step $j$ is $j$, and  $W(j)=w_j=\frac{a_{j}-a_{j-1}}{b_{j}-b_{j-1}}$\\ 
\\

At step $d+1$:\\
The largest minimizer at step $d$ was $d$, then in order to prove that the largest minimizer at this step is $L$, we should prove that for all $n>d$, we have: $\frac{a_L-a_{d}}{b_L-b_{d}} \leq \frac{a_{n}-a_{d}}{b_{n}-b_{d}}$. We distinguish again between three cases:\\
1) $n \in [d+1,R-1]$:\\
We know that $w_{d+1}<...<w_{R-1}$. Then, considering the set $\{d,d+1,...,R-1\}$ and according to Lemma \ref{th.triangulaire}, we get $\frac{a_{d+1}-a_{d}}{b_{d+1}-b_{d}} \leq \frac{a_{n}-a_{d}}{b_{n}-b_{d}}$  for all $n \in [d+1,R-1]$.\\
Since $\frac{a_L-a_{d}}{b_L-b_{d}} \leq \frac{a_{d+1}-a_{d}}{b_{d+1}-b_{d}}$ (according to Lemma \ref{L_d_less_W_d_plus_1}), then $\frac{a_L-a_{d}}{b_L-b_{d}} \leq \frac{a_{n}-a_{d}}{b_{n}-b_{d}}$ for all $n \in [d+1,R-1]$.\\
\\

2) $n \in [R,L-R+1]$:\\
a) $d=R-1$:\\
We have $b_n=b_{R-1}=b_d$. The case where  the passive decision average time $b_n$ is equal to $b_{n_{d}}=b_d$ is not included in the computation  of   Whittle indices (recall that $n_{d}$ is the largest minimizer at step $d$ which is $d$). This case can be hence skipped.\\ 
%This case can be skipped as in the computation  of   Whittle indices the case where  then we skip this case because we do not take into account in computation of Whittle indices the case when the passive decision average time $b_n$ is equal to $b_{n_{d}}=b_d$ where $n_{d}$ is the largest minimizer at step $d$ which is $d$.\\ 
b) $d=R-2$:\\
$b_n=b_{R-1}$,and $a_n > a_{R-1}$, then applying Lemma \ref{L_d_less_W_d_plus_1} we have $\frac{a_n-a_{R-2}}{b_n-b_{R-2}} > \frac{a_{R-1}-a_{R-2}}{b_{R-1}-b_{R-2}} \geq \frac{a_L-a_{d}}{b_L-b_{d}} $, and we conclude the result.\\  
c) $d<R-2$:\\
We have $\frac{a_n-a_{R-2}}{b_n-b_{R-2}} > \frac{a_{R-1}-a_{R-2}}{b_{R-1}-b_{R-2}} $.
Therefore, by considering the set $\{d,d+1,....,R-2,n\}$, we have $\frac{a_{d+1}-a_{d}}{b_{d+1}-b_{d}}=w_{d+1} < \cdots \leq w_{R-2}=\frac{a_{R-2}-a_{R-3}}{b_{R-2}-b_{R-3}} < \frac{a_{n}-a_{R-2}}{b_{n}-b_{R-2}}$\\
Combining Lemma \ref{th.triangulaire} and Lemma \ref{L_d_less_W_d_plus_1}, we get  $\frac{a_{L}-a_{d}}{b_{L}-b_{d}} \leq \frac{a_{d+1}-a_{d}}{b_{d+1}-b_{d}} < \frac{a_{n}-a_{d}}{b_{n}-b_{d}}$ for all $n \in [R,L-R+1]$\\
3) $n \in [L-R+2,L-1]$:\\
Applying Lemma \ref{L_d_less_n_d}, we have $\frac{a_{L}-a_{d}}{b_{L}-b_{d}} < \frac{a_{n}-a_{d}}{b_{n}-b_{d}}$.\\

Hence we proved that at step $d+1$, the largest minimizer is $L$. Therefore the Whittle's index for all state $i$ from $d+1$ until $L$ is $W(i)=x_{L,d}=\frac{a_{L}-a_{d}}{b_{L}-b_{d}}$.\\        
This concludes the proof of the proposition.

\section{Proof of Proposition \ref{prop:optimal_solution_for_dual_relaxed_problem}}\label{app:optimal_solution_for_dual_relaxed_problem}
In order to prove this proposition we distinguish between two types of classes:\\
1) Class $k$ in which $W$ is different from all $W^k_i$.\\
2) Class $k$ such that there exists a given state $j$ that satisfies $W^k_j=W$.\\ 
First type of classes:
For the class $k$ in which $W$ is different from all $W^k_i$, we prove that the optimal threshold verifies $l_k(W)=l_k=\underset{i}{\text{max}} \{\arg \underset{i}{\text{max}} \{W^k_i | W^k_i \leq W\}\}= \underset{i}{\text{max}} \{\arg \underset{i}{\text{max}} \{W^k_i) | W^k_i < W\}\}$. First we have $\underset{i}{\text{max}} \{\arg \underset{i}{\text{max}} \{W^k_i | W^k_i \leq W\}\}= \underset{i}{\text{max}} \{\arg \underset{i}{\text{max}} \{W^k_i) | W^k_i < W\}\}$ since $W_i^k$ is different from $W$ for all state $i$. For state $i$ less than $l_k$, given that $W^k_i$ is increasing in $i$, then $W^k_i\leq W^k_{l_k} < W$. Hence, due to the indexability of the class, $D(W_i^k) \subseteq D(W)$,  which implies that  the optimal decision at state $i$ is passive action. For the state $i$ strictly greater than $l_k$, by definition of $l_k$, $W^k_i$ must be strictly greater than $W$ since $l_k$ is the biggest integer among the states that give the biggest Whittle index less than $W$.
Then, according to the definition of Whittle index, $W < \min\{W, i \in D(W)\}$ that means $W  \not\in \{W, i \in D(W)\}$, therefore $i \not\in D(W)$.
Thus, the optimal decision at state $i>l_k$ is active decision. Hence $l_k=\underset{i}{\text{max}} \{\arg \underset{i}{\text{max}} \{W^k_i | W^k_i \leq W\}\}= \underset{i}{\text{max}} \{\arg \underset{i}{\text{max}} \{W^k_i) | W^k_i < W\}\}$ is effectively the optimal threshold $l_k(W)$.\\
For the second case, we start first by describing qualitatively the optimal threshold with respect to $W$. Then we prove the explicit expression: \\ 
Second type of classes:\\
For the class $k$ such that there exists $j$, $W^k_j=W$, we distinguish between two cases:\\
1) $j \leq R_k-1$:\\
We know that according to Theorem \ref{w.i} $W^k_j=w^k_j=x_{j,j-1}$ which is the point for which if $W=x_{j,j-1}$, we have
$\sum_{q=0}^L au^j(q)q \ - \ W \sum_{q=0}^j u^j(q)=\sum_{q=0}^L au^{j-1}(q)q \ - \ W \sum_{q=0}^{j-1} u^{j-1}(q)$. That means, according to equation~\eqref{eq:new_prob_form_lagr_function_steady_state}, for $W=x_{j,j-1}$, if $j$ is a minimizer of this equation ($j$ is the optimal threshold), then $j-1$ is also a minimizer of this equation.
Due to indexability, for all states less or equal than $j$ the optimal decision is to stay passive. Also, according to definition of Whittle index,  for all states strictly higher than $j$  the optimal decision is to be active. Then, $j$ could be the threshold, so as for $j-1$.\\
Hence, the optimal threshold can be either $j$ or $j-1$.\\
In fact, since $W^k_{0}< \cdots < W^k_{j-1} < W^k_j=W$, then  $j= \underset{i}{\text{max}} \{\arg \underset{i}{\text{max}} \{W^k_i | W^k_i \leq W\}\}$, and $j-1=\underset{i}{\text{max}} \{\arg \underset{i}{\text{max}} \{W^k_i | W^k_i < W\}\}$.\\
This proves the proposition for this case.\\
2) If $j \geq R_k$:\\
Then $W^k_j=W^k_L=W^k_{R_k}=W$, thus according to Theorem \ref{w.i}, $W=x_{L,R_k-1}$. That means, at $W$, the threshold policy can be either $L$ or $R_k-1$. $L$ is the biggest integer such that $W^k_j=W$, and $R_k-1$ is the biggest integer that verifies the strict inequality, explicitly $L= \underset{i}{\text{max}} \{\arg \underset{i}{\text{max}} \{W^k_i | W^k_i \leq W\}\}$ and $R_k-1=\underset{i}{\text{max}} \{\arg \underset{i}{\text{max}} \{W^k_i | W^k_i < W\}\}$.\\

\section{Proof of Proposition \ref{prop:optimal_solution_for_constrained_relaxed_problem}}\label{app:optimal_solution_for_constrained_relaxed_problem}
From optimization theory, it is known that the optimal solution of the dual problem is less or equal than the primal problem's solution when the constraint is satisfied, i.e:
\begin{equation}\label{dual_less_primal}
\underset{W}{\text{max}} \underset{\phi \in \Phi}{\text{min}} \ f(W,\phi)\leq  \underset{\phi \in \Phi}{\text{min}} \limsup\limits_{T\rightarrow\infty}\frac{1}{T} \mathbb{E}\left[ \sum_{t=0}^{T-1}  \sum_{k=1}^{K} \sum_{i=1}^{{\gamma_k}N} a_k q_i^k(t) \mid \textbf{q}(0), \phi \right] 
\end{equation}
As the optimal solution for fixed $W$ is a threshold policy, we use  the steady state form and the expression of the LHS of the inequality becomes:
\begin{equation}\label{eq:new_relaxed_prob_form}
\underset{W}{\text{max}}\underset{\phi}{\text{min}} f(W,\phi)=\underset{W}{\text{max}}\{\sum_{k=1}^K \sum_{i=1}^{\gamma_k N} [\underset{l_k \in [0,L]}{\text{min}} \{ \sum_{q=0}^L a_k u^{l_k}_k(q)q - W \sum_{q=0}^{l_k} u^{l_k}_k(q)\}]+ W (1-\alpha) N\} 
\end{equation} 
with $\phi$ the threshold policy that corresponds to $l(W)$ computed using Proposition \ref{prop:optimal_solution_for_dual_relaxed_problem} for fixed $W$.
For $W^*$ that satisfies the constraint with equality (i.e. $\alpha N=\sum_{k=1}^K \gamma_k N\sum_{i=l_{k+1}(W^*)}^L u_k^{l_k(W^*)}(i)$, which is in fact true for all $N$, and then we can get rid of $N$), we get exactly the objective function of the primal problem. Therefore,we get a threshold vector $l(W^*)$ that gives a solution for the primal problem less than the optimal solution for this problem according to inequality (\ref{dual_less_primal}). Then, surely this solution given by $l(W^*)$ is the optimal one for the constrained relaxed problem, since it satisfies the constraint and for all policy $\phi$ that satisfies the constraint and belong to $\Phi$, we have $f(W^*,l(W^*))=\sum_{k=1}^K \sum_{i=1}^{\gamma_k N} [ \sum_{q=0}^L a_k u^{l_k(W^*)}_k(q)q]= \limsup\limits_{T\rightarrow\infty}\frac{1}{T} \mathbb{E}\left[ \sum_{t=0}^{T-1}  \sum_{k=1}^{K} \sum_{i=1}^{{\gamma_k}N} a_k q_i^k(t) \mid \textbf{q}(0),l(W^*)\right]$ \linebreak $\leq \underset{\phi}{\text{min}} \limsup\limits_{T\rightarrow\infty}\frac{1}{T} \mathbb{E}\left[ \sum_{t=0}^{T-1}  \sum_{k=1}^{K} \sum_{i=1}^{{\gamma_k}N} a_kq_i^k(t) \mid \textbf{q}(0),\phi\right]$. \\  
We deduce that the solution of the relaxed problem is of type threshold-based policy $l(W^*)$ with $W^*$ satisfies $\alpha =\sum_{k=1}^K \gamma_k \sum_{i=l_{k+1}(W^*)}^L u_k^{l_k(W^*)}(i)$.

\section{Proof of Proposition \ref{prop:optimal_sol_relax_prob_charchterisation}}\label{app:prop:optimal_sol_relax_prob_charchterisation}
We first introduce the following useful lemma.
\begin{Lemma}\label{decreasing_active_time}
For each class $k$, $\sum_{n+1}^L u^{n}_k(i)$ is strictly decreasing in $n$, when $n \in [-1,R_k-1] \cup L$.
\end{Lemma}
\begin{proof}
\renewcommand{\qedsymbol}{$\blacksquare$}
We have $\sum_{0}^n u^{n}_k(i)$ is strictly increasing in this set (see Lemma \ref{temps_pass_str_croiss_less_R} and the fact that $1=\sum_{0}^L u^L_k(i)>\sum_{0}^{R_k-1} u^{R_k-1}_k(i)$ ), then $\sum_{n+1}^L u^{n}_k(i)=1-\sum_{0}^n u^{n}_k(i)$, is strictly decreasing in $n$.\\
\end{proof}
We define the following order relation in $\mathbb{R}^K$ such that for any two vectors $l^1$ and $l^2$, $l^1 \leq l^2 \Longleftrightarrow$ for each element of vector of index $k$, we have $l^1_k \leq l^2_k$. Recall that according to Proposition \ref{prop:optimal_solution_for_dual_relaxed_problem}, we can directly deduce that for $W_1 \leq W_2$ $l(W_1) \leq l(W_2)$ and for all $W$ and class $k$, $l_k(W)$ can be either less than $R_k-1$ or equal to $L$.\\
Without loss of generality, when $W \in \mathbb{R}^+$, the corresponding set of threshold vectors $l(W)$ is perfectly ordered. Then, by applying Lemma \ref{decreasing_active_time}, $\sum_{k=1}^K \gamma_k \sum_{i=l_k(W)+1}^L u^{l_k(W)}_k(i)$ is strictly decreasing in $l(W)$, and take discrete values from $1$ to $0$. According to Proposition \ref{prop:optimal_solution_for_dual_relaxed_problem}, we have for each class $k$ and state $i$, if $W=W_i^k$ then there is two possible optimal thresholds vectors $l^1(W)$ and $l^2(W)$ with $l^1(W) < l^2(W)$. Hence we can deduce that there exists a class $m$ and state $p$ such that $\sum_{k=1}^K \gamma_k \sum_{i=l_k^1(W_p^m)+1}^L u^{l_k^1(W_p^m)}_k(i) \geq \alpha$ and $\sum_{k=1}^K \gamma_k \sum_{i=l_k^2(W_p^m)+1}^L u^{l_k^2(W_p^m)}_k(i) \leq \alpha$. We find the relation between $l^1(W_p^m)$ and $l^2(W_p^m)$.   
%Then, there exists $W_1$ and $W_2$ such that $\sum_{k=1}^K \gamma_k \sum_{i=l_k(W_1)+1}^L u^{l_k(W_1)}_k(i) > \alpha \geq \sum_{k=1}^K \gamma_k \sum_{i=l_k(W_2)+1}^L u^{l_k(W_2)}_k(i)$.

Before that, we prove that $l_k(W_p^m)$ is less than $R_k-1$ for all class $k$ under assumption \ref{assump:condition_on_alpha} and \ref{assump:buffer_size_2}. For that we need to check if there exist $W_0$ such that $l(W_0)=(R_1,\cdots,R_K)$.
In fact, according to assumption \ref{assump:buffer_size_2}, we can deduce that $W^{k'}_L$ is strictly greater than $W^k_i$ for $i \in [0,R_k-1]$ for all $k$ and $k'$ (we check that by replacing the expression of $W^{k'}_L$ and $W^k_{R_k-1}$ given in Theorem 3). Hence there exists a given $W$ such that $W^k_i < W$ for $i \in [0,R_k-1]$ and $W^{k'}_L>W$ for all $k$ and $k'$. Then for a such $W$ denoted $W_0$ the optimal threshold for each class $k$ is $l_k(W_0)=R_k-1$.\\

According to the expression of the average passive time given in section \ref{sec:WI}, $\sum_{k=1}^K \gamma_k \sum_{i=0}^{R_k-1} u^{R_k-1}_k(i)=\frac{1}{2}+\sum_{k=1}^K \frac{\gamma_k}{2R_k}$, therefore $\sum_{k=1}^K \gamma_k \sum_{i=R_k}^L u^{R_k-1}_k(i)=\frac{1}{2}-\sum_{k=1}^K \frac{\gamma_k}{2R_k}$. Hence, considering the assumption \ref{assump:condition_on_alpha}, $\alpha \geq \sum_{k=1}^K \gamma_k \sum_{i=R_k}^L u^{R_k-1}_k(i)$.\\
As $\sum_{k=1}^K \gamma_k \sum_{i=l_k^1(W_p^m)+1}^L u^{l_k^1(W_p^m)}_k(i) \geq \alpha$, then $l^1(W_p^m) \leq (R_1-1,\cdots,R_K-1)=l(W_0)$. Given that the thresholds vector are increasing in $W$, $W_p^m \leq W_0$, hence $l^2(W_p^m) \leq l(W_0)=(R_1-1,\cdots,R_K-1)$\\
Therefore, $l_m^2(W_p^m) \leq R_k-1$, then according to Proposition \ref{prop:optimal_solution_for_dual_relaxed_problem}, when $W=W_p^m$, $l_m(W_p^m)=l_m^2(W_p^m)$ and $l_m^1(W_p^m)=l_m^2(W_p^m)-1=l_m(W_p^m)-1$ can be both the optimal thresholds for class $m$. As for the other classes, $l_k^1(W_p^m)=l_k^2(W_p^m)=l_k(W_p^m)$.

If we force $W^*$ to be equal to $W_p^m$, the optimal threshold vector can be either $l^1(W_p^m)$ or $l^2(W_p^m)$, then we can introduce some randomization between the two policies. In other words, we use the threshold policy $l^1(W_p^m)$ with probability $\theta$ and $l^2(W_p^m)$ with probability $1-\theta$. The new stationary distribution for the class $m$ is then a linear combination of these two threshold policies $l_m(W_p^m)$ and $l_m(W_p^m)-1$:
$u^*_m=\theta u^{l_m(W_p^m)}_m + (1-\theta) u^{l_m(W_p^m)-1}_m$.\\
Hence, in a state strictly less than $l_m(W_p^m)$, the queues will not transmit, whereas in a state strictly greater than $l_m(W_p^m)$, they will transmit with probability one. If the queues are in state $l_m(W_p^m)$, they will transmit with probability $\frac{(1-\theta) u^{l_m(W_p^m)-1}_m(l_m(W_p^m))}{\theta u^{l_m(W_p^m)}_m(l_m(W_p^m)) + (1-\theta) u^{l_m(W_p^m)-1}_m(l_m(W_p^m))}$. Since the probability to be in this state $l_m(W_p^m)$ is $u^*_m(l_m(W_p^m))$,  the proportion of time that the queues will be in active mode is: 
\begin{equation}
\alpha= \sum_{k\neq m} \sum_{i=l_k(W_p^m)+1}^L \gamma_k u^{l_k(W_p^m)}_k(i)+\sum_{i=l_m(W_p^m)+1}^L \gamma_m u^*_m(i)\nonumber+(1-\theta) \gamma_m u^{l_m(W_p^m)-1}_m(l_m(W_p^m))
\end{equation}
When $\theta=0$, the threshold policy is $l_m(W_p^m)-1$ and the total average time in active mode is higher than $\alpha$. 
When $\theta=1$, the threshold policy is $l_m(W_p^m)$ and the total average time in active mode is less than $\alpha$.\\
Given that $\sum_{k\neq m} \sum_{i=l_k(W_p^m)+1}^L \gamma_k u^{l_k(W_p^m)}_k(i)+\sum_{i=l_m(W_p^m)+1}^L \gamma_m u^*_m(i)+(1-\theta) \gamma_m u^{l_m(W_p^m)-1}_m(l_m(W_p^m))$ is continuous in $\theta$, then there exists at least one $\theta$ which verifies the equality.
Hence, for $W^*=W_p^m$, we get a threshold policy for all classes except for class $m$ where the optimal solution is a linear combination of two threshold policies. Moreover for a given randomized parameter $\theta$, the constraint \eqref{eq:constraint_relaxed} is satisfied with equality. \\

\section{Proof of Proposition \ref{spectral_radius_Q_less_1}}\label{app:spectral_radius_Q_less_1}
We derive the eigenvalues of $Q$.\\
The matrix $Q$ is of the form: \\       
\begin{equation}\left[\begin{array}{ccccccc}
Q_1&0&\cdots &\cdots &\cdots &\cdots &0 \\
0&Q_2&\cdots &\cdots &\cdots &\cdots &0 \\
\vdots & &\ddots & & & & \\
A_1&A_2&\cdots &Q_m&\cdots&A_{K-1}&A_K\\
\vdots&  & &\ddots & &\vdots&\\
0&0&\cdots&\cdots&\cdots&Q_{K-1}&0\\
0&0&\cdots&\cdots&\cdots&0&Q_K\\
\end{array}\right]
\end{equation}

The characteristic polynomial of $Q$ is the product of 
the characteristic polynomial of each matrix $Q_k$: \begin{equation} \chi_{Q}(\lambda)=\prod_{k=1}^K \chi_{Q_k}(\lambda) \end{equation} 
%1)The case $k>m$:
%%%% debut macro %%%%
\newenvironment{changemargin}[2]{\begin{list}{}{%
\setlength{\topsep}{0pt}%
\setlength{\leftmargin}{0pt}%
\setlength{\rightmargin}{0pt}%
\setlength{\listparindent}{\parindent}%
\setlength{\itemindent}{\parindent}%
\setlength{\parsep}{0pt plus 1pt}%
\addtolength{\leftmargin}{#1}%
\addtolength{\rightmargin}{#2}%
}\item }{\end{list}}

1)The case $k \neq m$:
\[Q_k=\]
\begin{changemargin}{-2cm}{0 cm}
\[
\begin{blockarray}{ccccccccccccccccc}
 &0&1&\cdots&l-1&l+1&l+2&\cdots&R-1&R&R+1&\cdots&l+R-1&l+R&l+R+1&\cdots&L \\
\begin{block}{c(cccccccccccccccc)}
  0&\rho_k& 0 &\cdots& 0 & \rho_k &\cdots& \cdots & \rho_k&\rho_k&0&\cdots&0&0&\cdots&\cdots&0 \\
  \vdots &\vdots& \ddots& & \vdots & \vdots & \rho_k &\rho_k& \vdots&\vdots&\ddots&&\vdots&\vdots&0&0&       \vdots \\
  l-2&\vdots&&\ddots&0&\vdots&&&\vdots&\vdots&&\ddots&0&\vdots&&&\vdots\\  
  l-1 & \rho_k & \cdots &\cdots& \rho_k & \rho_k & \cdots&\cdots&\rho_k&\rho_k&\cdots&\cdots&\rho_k&0&\cdots&\cdots&0 \\
  l+1 & 0 & \cdots &\cdots& 0 &0&\cdots&\cdots& 0 &0&\cdots&\cdots&0&0&-\rho_k&\cdots&-\rho_k \\
  \vdots & \vdots & 0 & 0&\vdots & \vdots & 0&0&\vdots&\vdots&0&0&\vdots&\vdots&\ddots&-\rho_k&-\rho_k \\
  R-2&\vdots&&&\vdots&\vdots&&&\vdots&\vdots&&&\vdots&\vdots&\ddots&-\rho_k&-\rho_k\\
  R-1&0&\cdots&\cdots&0&0&\cdots&\cdots&0&0&\cdots&\cdots&0&0&\cdots&-\rho_k&-\rho_k\\  
  R& -\rho_k & 0 &\cdots& 0 & -\rho_k &\cdots& \cdots&-\rho_k& -\rho_k&0&\cdots&0&0&\cdots&-\rho_k&-\rho_k \\  
  \vdots& \vdots & \ddots && \vdots & \vdots &-\rho_k& -\rho_k&\vdots&\vdots&\ddots&&\vdots&\vdots&0&\ddots&\vdots \\
  &\vdots&&\ddots&0&\vdots&&&\vdots&\vdots&&\ddots&0&\vdots&&&-\rho_k\\  
  & -\rho_k &\cdots& \cdots & -\rho_k & -\rho_k &\cdots& \cdots&-\rho_k&-\rho_k&\cdots&\cdots&-\rho_k&0&\cdots&-\rho_k&-\rho_k\\
  l+R& 0 & \cdots&\cdots & 0 & 0 &\cdots& \cdots&0&0&\cdots&\cdots&0&0&\rho_k&\cdots&0 \\  
  \vdots& \vdots &0& 0 & \vdots & \vdots &0& 0&\vdots&\vdots&0&0&\vdots&\vdots&\ddots&\rho_k&0 \\
  L-1&\vdots&&&\vdots&\vdots&&&\vdots&\vdots&&&\vdots&\vdots&&\ddots&0\\  
  L& 0 &\cdots& \cdots & 0 & 0 & \cdots&\cdots&0&0&\cdots&\cdots&0&0&\cdots&\cdots&0 \\
\end{block}
\end{blockarray}
 \]
 \end{changemargin}
 After computations 	and some algebraic manipulations, we get,
 $\chi_{Q_k}(\lambda)=(-\lambda)^L$\\

2)The case $k=m$:
\[Q_m=\]
\begin{changemargin}{-2cm}{0 cm}
\[
\begin{blockarray}{ccccccccccccccccc}
 &0&1&\cdots&l-1&l+1&l+2&\cdots&R-1&R&R+1&\cdots&l+R-1&l+R&l+R+1&\cdots&L \\
\begin{block}{c(cccccccccccccccc)}
  0&\rho_m& 0 &\cdots& 0 & 0 &\cdots& \cdots & 0&0&-\rho_m&\cdots&-\rho_m&-\rho_m&\cdots&\cdots&-\rho_m \\
  \vdots &\vdots& \ddots& & \vdots & \vdots & 0 &0& \vdots&\vdots&\ddots&&\vdots&\vdots&-\rho_m&-\rho_m&       \vdots \\
  l-2&\vdots&&\ddots&0&\vdots&&&\vdots&\vdots&&\ddots&-\rho_m&\vdots&&&\vdots\\  
  l-1 & \rho_m & \cdots &\cdots& \rho_m & 0 & \cdots&\cdots&0&0&\cdots&\cdots&0&-\rho_m&\cdots&\cdots&-\rho_m \\
  l+1 & 0 & \cdots &\cdots& 0 &0&\cdots&\cdots& 0 &0&\cdots&\cdots&0&0&-\rho_m&\cdots&-\rho_m \\
  \vdots & \vdots & 0 & 0&\vdots & \vdots & 0&0&\vdots&\vdots&0&0&\vdots&\vdots&\ddots&-\rho_m&-\rho_m \\
  R-2&\vdots&&&\vdots&\vdots&&&\vdots&\vdots&&&\vdots&\vdots&\ddots&-\rho_m&-\rho_m\\
  R-1&0&\cdots&\cdots&0&0&\cdots&\cdots&0&0&\cdots&\cdots&0&0&\cdots&-\rho_m&-\rho_m\\  
  R& -\rho_m & 0 &\cdots& 0 & 0 &\cdots& \cdots&0& 0&\rho_m&\cdots&\rho_m&\rho_m&\cdots&0&0 \\  
  \vdots& \vdots & \ddots && \vdots & \vdots &0&0&\vdots&\vdots&\ddots&&\vdots&\vdots&\rho_m&\ddots&\vdots \\
  &\vdots&&\ddots&0&\vdots&&&\vdots&\vdots&&\ddots&\rho_m&\vdots&&&0\\  
  & -\rho_m &\cdots& \cdots & -\rho_m &0 &\cdots& \cdots&0&0&\cdots&\cdots&0&\rho_m&\cdots&0&0\\
  l+R& 0 & \cdots&\cdots & 0 & 0 &\cdots& \cdots&0&0&\cdots&\cdots&0&0&\rho_m&\cdots&0 \\  
  \vdots& \vdots &0& 0 & \vdots & \vdots &0& 0&\vdots&\vdots&0&0&\vdots&\vdots&\ddots&\rho&0 \\
  L-1&\vdots&&&\vdots&\vdots&&&\vdots&\vdots&&&\vdots&\vdots&&\ddots&0\\  
  L& 0 &\cdots& \cdots & 0 & 0 & \cdots&\cdots&0&0&\cdots&\cdots&0&0&\cdots&\cdots&0 \\
\end{block}
\end{blockarray}
 \]
 \end{changemargin}
 After computations 	and some algebraic manipulations, we get:
 $\chi_{Q_m}(\lambda)=(-\lambda)^{L-1}(l_m\rho_m-\lambda)$\\
 For $k \neq m$ $Q_k$ has only $0$ as eigen value.\\
 For $k=m$, $\chi_{Q_m}(\lambda)=0 \Leftrightarrow \lambda=0$ or $\lambda=l_m \rho_m$, hence $Q_m$ has two eigen values which are $0$ and $l_m \rho_m$. Given Assumption \ref{assump:condition_on_alpha}, the optimal threshold $l_k$ is less strictly than $R_k$ for all $k$. Accordingly $l_m \rho_m < R_m \rho_m=1$.\\
Consequently, in both cases, the norms of all eigen values of the obtained matrix  are strictly less than 1.

\section{Proof of Lemma \ref{lem:result_kurth_theorem}}\label{app:result_kurth_theorem}
We take $0<\epsilon<\mu$, $\textbf{Z}^N(t)$ converges to $\textbf{z}^*$, i.e. there exists $T_0$ such that for all $t \geq T_0$, $||\textbf{Z}^N(t)-\textbf{z}^*||\leq \epsilon$.
Hence: 
\begin{align}
P_x(\underset{T_0 \leq t < T}{\text{sup}} ||\textbf{Z}^N(t)-\textbf{z}^*|| \geq \mu)& \leq P_x(\underset{T_0 \leq t < T}{\text{sup}} ||\textbf{Z}^N(t)-\textbf{z}(t)||+||\textbf{Z}^N(t)-\textbf{z}^*|| \geq \mu)\\
& \leq P_x(\underset{T_0 \leq t < T}{\text{sup}} ||\textbf{Z}^N(t)-\textbf{z}(t)|| \geq \mu-\epsilon)\\
& \leq P_x(\underset{0 \leq t < T}{\text{sup}} ||\textbf{Z}^N(t)-\textbf{z}(t)|| \geq \mu-\epsilon)
\end{align}
Using Proposition \ref{prop:kurth_theorem}, there exists $s_1$ and $s_2$ such that:
\begin{equation} P_x(\underset{0 \leq t < T}{\text{sup}} ||\textbf{Z}^N(t)-\textbf{z}(t)|| \geq \mu-\epsilon) \leq s_1 exp(-Ns_2). \end{equation} 
Therefore:
\begin{equation} P_x(\underset{T_0 \leq t < T}{\text{sup}} ||\textbf{Z}^N(t)-\textbf{z}^*|| \geq \mu) \leq s_1 exp(-Ns_2).
 \end{equation}

\section{Proof of Proposition \ref{prop:local_optimality}}\label{app:local_optimality}
We recall that  $\textbf{Z}^N(t)$ represents the proportion vector at time $t$ under Whittle's Index policy.\\
%We have 
%\begin{align}
%C^{WI,N}-C^{RP,N}=&|\frac{1}{T} \mathbb{E}\left[ \sum_{t=0}^{T-1}  \sum_{k=1}^{K} \sum_{i=1}^{{\gamma_k}N} C^{WI,N}(q^k_i(t)) \mid \textbf{x}\right]\\
%&-\frac{1}{T} \mathbb{E}\left[ \sum_{t=0}^{T-1}  \sum_{k=1}^{K} \sum_{i=1}^{{\gamma_k}N} C^{RP,N}(q^k_i(t))\right]|\\
%=&|\frac{1}{T} \mathbb{E}\left[ \sum_{t=0}^{T-1}  \sum_{k=1}^{K} \sum_{i=1}^{{\gamma_k}N} C^{WI,N}(Q^k_i(t)) \mid \textbf{x}\right]\\
%&- C^{RP,N}|
%\end{align}
Replacing $C^{RP,N}$ by its expression given in section \ref{sec:relaxed} and knowing that $z^{k,*}_i=\gamma_k u^{l_k}_k(i)$ for $k \neq m$ and $z^{m,*}_i=\gamma_m u^{*}_m(i)=\theta\gamma_m u^{l_m}_m(i)+(1-\theta)\gamma_m u^{l_m-1}_m(i)$ (by definition of $\textbf{z}^*$), then the difference between $C^{WI,N}$ and $C^{RP,N}$ can be expressed as:
\begin{align}
C_T^{N}(\textbf{x})-C^{RP,N}=|\frac{1}{T} \mathbb{E}\left[ \sum_{t=0}^{T-1}  \sum_{k=1}^{K} \sum_{i=1}^{L} a_kZ_i^{k,N}(t)iN \mid \textbf{x}\right]-\frac{1}{T} \mathbb{E}\left[ \sum_{t=0}^{T-1} \sum_{k=1}^{K} \sum_{i=1}^{L} a_kz_i^{k,*}iN\right]|
\end{align}
We divide all by $N$
\begin{align}
\frac{C_T^{N}(\textbf{x})}{N}-\frac{C^{RP,N}}{N}=&|\frac{1}{T}  \sum_{t=0}^{T-1}  \sum_{k=1}^{K} \sum_{i=1}^{L} \mathbb{E}(a_kZ_i^{k,N}(t)i) - a_kz_i^{k,*}i |\nonumber\\
\leq& |\frac{1}{T}  \sum_{t=0}^{T_0-1}  \sum_{k=1}^{K} \sum_{i=1}^{L} \mathbb{E}(a_kZ_i^{k,N}(t)i) - a_kz_i^{k,*}i |+|\frac{1}{T}  \sum_{t=T_0}^{T-1}  \sum_{k=1}^{K} \sum_{i=1}^{L} \mathbb{E}(a_kZ_{i}^{k,N}(t)i) - a_kz_i^{k,*}i |\nonumber\\
\leq &\frac{T_0L(L+1)}{T}\sum_{k=1}^Ka_k\gamma_k+|\frac{1}{T}  \sum_{t=T_0}^{T-1}  \sum_{k=1}^{K} \sum_{i=1}^{L} \mathbb{E}(a_kZ_{i}^{k,N}(t)i) - a_kz_i^{k,*}i |
\end{align}
We have the function $f: \ z \rightarrow \sum_{k=1}^K\sum_{i=0}^L a_k z_i^k i$ is lipchitz and continuous, then for an arbitrary small $\epsilon$, there exists $\mu$ such that if $||\textbf{z}-\textbf{z}^*||<\mu$, then $|f(z)-f(\textbf{z}^*)|<\epsilon$.

We denote $Y_N$ the event {$\underset{T_0 \leq t < T}{\text{sup}} ||\textbf{Z}^N(t)-\textbf{z}^*|| \geq \mu$}, we proceed to bound the second term: 
\begin{align}
|\frac{1}{T}  \sum_{t=T_0}^{T-1}  \sum_{k=1}^{K} \sum_{i=1}^{L} E(a_kZ_{i}^{k,N}(t)i) - a_kz_i^{k,*}i |=&P_x(Y_N)\frac{1}{T}  \sum_{t=T_0}^{T-1}  \mathbb{E}\left[|\sum_{k=1}^{K} \sum_{i=1}^{L} (a_kZ_{i}^{k,N}(t)i) - a_kz_i^{k,*}i| | Y_N\right]\nonumber\\
&+(1-P_x(Y_N))\frac{1}{T} \mathbb{E}\left[| \sum_{t=T_0}^{T-1}  \sum_{k=1}^{K} \sum_{i=1}^{L} (a_kZ_{i}^{k,N}(t)i) - a_kz_i^{k,*}i||\overline{Y_N}\right]\nonumber\\
\leq& \frac{(T-T_0)L(L+1)}{T}\sum_{k=1}^Ka_k\gamma_k P_x(Y_N)+(1-P_x(Y_N))\epsilon.
\end{align}
where the above inequality comes from the fact that $|a_kZ_{i}^{k,N}(t)i - a_kz_i^{k,*}i| \leq 2 \gamma_k a_ki$.  According to Lemma \ref{lem:result_kurth_theorem}, we have $\lim_{N \rightarrow \infty}P_x(Y_N)=0$, then
\begin{align}
\lim_{N \rightarrow \infty} |\frac{1}{T} \mathbb{E}\left[ \sum_{t=0}^{T-1}  \sum_{k=1}^{K} \sum_{i=1}^{L} a_kZ_i^{k,N}(t)iN \mid \textbf{x}\right]-\frac{1}{T} \mathbb{E}\left[ \sum_{t=0}^{T-1}  \sum_{k=1}^{K} \sum_{i=1}^{L} a_kz_i^{k,*}iN\right]|\leq \frac{T_0L(L+1)}{T}\sum_{k=1}^Ka_k\gamma_k+\epsilon
\end{align} 
This inequality is true $\forall \epsilon > 0$, then:
\begin{align}
 \lim_{N \rightarrow \infty} |\frac{1}{T} \mathbb{E}\left[ \sum_{t=0}^{T-1}  \sum_{k=1}^{K} \sum_{i=1}^{L} a_kZ_i^{k,N}(t)iN \mid \textbf{x}\right]-\frac{1}{T} \mathbb{E}\left[ \sum_{t=0}^{T-1}  \sum_{k=1}^{K} \sum_{i=1}^{L} a_kz_i^{k,*}iN\right]| \leq \frac{T_0L(L+1)}{T}\sum_{k=1}^Ka_k\gamma_k
\end{align} 
Finally we have:
\begin{align}
 \lim_{T \rightarrow \infty} \lim_{N \rightarrow \infty}\frac{C_T^{N}(\textbf{x})}{N}-\frac{C^{RP,N}}{N} =0 
\end{align} 

\section{Proof of Lemma \ref{lem:reachability_z_0}}\label{app:reachability_z_0}
We consider any initial state $(\textbf{z}^1,\textbf{z}^2,\cdots,\textbf{z}^K)$, and we consider only the following possible event (that arises with strictly positive probability): whatever the transmission decision taken, there is no arrivals ($A(t)=0$ for large but finite number of time slots $T$).\\
This implies as long as the queue $q^k_i$ is scheduled, the number of its packets will decrease by $R_k$. However, in order to prove that the queues' state "0" is achieved,   we must ensure that the queues can be scheduled enough times in order to reach the smallest length $0$. This is shown in the sequel.
If we consider that from time $t$ to $T$, the queue $q^k_i$ is not scheduled and $q^k_i(t)>0$, then  among the other queues a proportion $\alpha$ will be scheduled at each time slot, until we reach a state when $\alpha$ queues will have a length  strictly less than $q^k_i(t)$ (this is feasible since the queues length does not increase with time, and they either stay unchanged or decrease depending on the action taken). Recall that the maximum length of each queue is $L$. So after a finite time but long enough, the queue $q^k_i$ should be scheduled. Hence all queues will be surely scheduled as long as their length is strictly higher than $0$. Therefore, after a finite long time all queues will be at state $0$. The aforementioned analysis is true if the event assumed at the beginning of the proof arises. Since L is finite, time $T$ is surely finite and this event definitely arises with strictly positive probability. This implies that state "0"  can be reached from any other state, which concludes the proof.            

\section{Proof of Proposition \ref{prop:reachability_z_star}}\label{app:reachability_z_star}
In order to find a path from $\textbf{z}(0)$ to $\textbf{z}^*$, we give two useful lemmas.
\begin{Lemma}\label{lem:Whittle_index_comparing_1_R}
$w^k_1 \leq w^{k^{'}}_{R_{k^{'}}-1}$ for all $R_k$ and $R_{k^{'}}$ greater or equal than $2$.
\end{Lemma}
\begin{IEEEproof}
\renewcommand{\qedsymbol}{$\blacksquare$}
\begin{align}
w^k_1 &= \frac{R_k}{R_k-1} \leq 2 \ \forall  R_k\geq 2 \\
w^k_{R_k-1}&=R_k(R_k-1)\geq 2 \  \forall R_k\geq 2
\end{align}
\end{IEEEproof}
\begin{Lemma}\label{lem:inequality_alpha}
\begin{equation}\sum_{k \neq m} \sum_{i=R_k}^{R_k+l_k-1} \gamma_k u^{l_k}_k(i)+\sum_{i=R_m}^{R_m+l_m-1} \gamma_m u^*_m(i) \leq 1-\alpha.\end{equation}
\end{Lemma}

\begin{IEEEproof}
\renewcommand{\qedsymbol}{$\blacksquare$}
In fact \begin{equation}1-\alpha=\sum_{k \neq m} \sum_{i=0}^{l_k} \gamma_k u^{l_k}_k(i)+\sum_{i=0}^{l_m-1} \gamma_m u^*_m(i)+\theta \gamma_m u^*_m(l_m).\end{equation}
For any $k \in [1,K]$ and for any threshold $n_k < R_k$, and by replacing $u^{n_k}_k$ by its expression given in section \ref{sec:steady_state}, we have: 
\begin{align}
\sum_{i=0}^{n_k} \gamma_k u^{n_k}_k(i)=\gamma_k \sum_{i=0}^{n_k} (\rho_k-(n_k-i)\rho_k^2)=\gamma_k(n_k+1)\rho_k\nonumber-\gamma_k\rho_k^2\frac{(n_k+1)n_k}{2} 
\end{align}
and
\begin{equation}\sum_{i=R_k}^{R_k+n_k-1} \gamma_k u^{n_k}_k(i)=\gamma_k \sum_{i=0}^{n_k} (n_k-i)\rho_k^2=\gamma_k\rho_k^2\frac{(n_k+1)n_k}{2}\end{equation}
we have:

\begin{equation}\gamma_k(n_k+1)\rho_k>\gamma_k\rho_k^2(n_k+1)n_k\end{equation}
Hence:
\begin{align}\label{inequality_for_path}
\sum_{i=0}^{n_k} \gamma_k u^{n_k}_k(i) \geq \sum_{i=R_k}^{R_k+n_k-1} \gamma_k u^{n_k}_k(i)
\end{align}
That means, for $k \neq m$:
\begin{equation}\sum_{k \neq m} \sum_{i=0}^{l_k} \gamma_k u^{l_k}_k(i) \geq \sum_{k \neq m} \sum_{i=R_k}^{R_k+l_k-1} \gamma_k u^{l_k}_k(i)\end{equation}

For $k=m:$
\begin{align}
\sum_{i=0}^{l_m-1} \gamma_m u^{*}_m(i) + \theta \gamma_m u^*_m(l_m)=& \gamma_m(1-\theta)\sum_{i=0}^{l_m-1}u^{l_m-1}_m(i)+\theta \gamma_m \sum_{i=0}^{l_m} u^{l_m}_m(i)\nonumber\\
\geq& \gamma_m(1-\theta)\sum_{i=R_m}^{R_m+l_m-2}u^{l_m-1}_m(i)+\theta \gamma_m \sum_{i=R_m}^{R_m+l_m-1}u^{l_m}_m(i)\nonumber\\
=& \gamma_m(1-\theta)\sum_{i=R_m}^{R_m+l_m-1}u^{l_m-1}_m(i)+\theta \gamma_m \sum_{i=R_m}^{R_m+l_m-1}u^{l_m}_m(i)\nonumber\\
=&\sum_{i=R_m}^{R_m+l-1}\gamma_m u^*_m(i)
\end{align}
The inequality comes from (\ref{inequality_for_path}).\\
Then $\sum_{k \neq m} \sum_{i=R_k}^{R_k+l_k-1} \gamma_k u^{l_k}_k(i)+\sum_{i=R_m}^{R_m+l_m-1} \gamma_m u^*_m(i)$ is less than $\sum_{k \neq m} \sum_{i=0}^{l_k} \gamma_k u^{l_k}_k(i)+\sum_{i=0}^{l_m-1} \gamma_m u^*_m(i)+\theta \gamma_m u^*_m(l_m)=1-\alpha$

\end{IEEEproof}
In the remaining of the proof, we will consider separately the cases $\alpha \leq \frac{1}{2}$ and $\alpha > \frac{1}{2}$. \\
If $\alpha \leq \frac{1}{2}$, the proof of the desired result consists of 3 steps. \\

Step 1:\\
We start by state $\textbf{z}(0)$,
for all $k \neq m$, we will exactly schedule all proportions: $z^{k,*}_{l_k+1}.......z^{k,*}_{L}$ , and for $k=m$, we schedule all proportions $z^{m,*}_{l_k+1},.....z^{m,*}_{L}$ plus the proportion $(1-\theta)z^{m,*}_{l_m}$.  The sum of these tree proportions is $\alpha$.  We denote these sets of queues by group A. We consider that, after scheduling, all these proportions will be at state $R_k-1$ (depending on each class). For the rest of proportions which is equal to $1-\alpha$, only $\alpha$ proportion will be at state $1$ (we call this group B). The rest which equals to $1-2\alpha$ (group C) will be at state $0$.
The queue state proportions vector for class $k \neq m$ after this step is: 
\begin{align}
z^k= (z^k_0=\beta_k,z^k_1=\alpha_k,0,0,\cdots, z^k_{R_k-1}=\sum_{i=l_k+1}^L z_i^{k,*},0,\cdots,0)
\end{align}
The queue state proportions vector for class $k=m$: 
\begin{equation}z^m=(\beta_m,\alpha_m,0,0,\cdots,\sum_{i=l_m+1}^L z_i^{m,*}+(1-\theta)z_{l_m}^{m,*},0,\cdots,0)\end{equation} with $\sum \alpha_k=\alpha$ and $\sum \beta_k=1-2\alpha$.\\
Step 2:\\
Using the Whittle's Index policy, according to Lemma \ref{lem:Whittle_index_comparing_1_R}, group A is scheduled again. After scheduling, we consider that group B which is at state $1$ goes to state $R_k$ ($R_k-1$ packets are the arrivals at each class-k queue). For group C, the queues stay at state $0$ (no arrivals).\\ 
But for the $\alpha$ proportion scheduled (group A), we have for each $k$: \\
1) when $k \neq m$:\\
a) For each state $h$ from $l_k+1$ until $R_k-1$: exactly $z_h^{k,*}$ goes to state $h$ (this is feasible since if a queue at state $R_k-1$ is scheduled, it can go to any other state strictly less than $R_k$)\\
b) For each state $h$ from $R_k$ until $R_k+l_k-1$, we will have exactly $z_h^{k,*}$ proportion of queues that go to state $h-(R_k-1)$, which is  strictly less than $R_k$.\\
2) When $k=m$\\
a) for each state from $l_m+1$ until $R_m+l_m-1$, the same analysis done for  $k \neq m$ holds.\\
b) for $h=l_m$, $(1-\theta)z_{l_m}^{m,*}$ will be at state $l_m$.\\
Hence after this step the new queue state proportion vector for class $k \neq m$ is:
\begin{align}
(\beta_k,z^k_1=z_{R_k}^{k,*},\cdots,z^k_{l_k}=z^{k,*}_{R_k+l_k-1},z^k_{l_k+1}=z^{k,*}_{l_k+1},\cdots,z^k_{R_k-1}=z^{k,*}_{R_k-1},\alpha_k,0,\cdots,0)
\end{align}
The queue state proportion vector for class $k=m$ is: 
\begin{align}
(\beta_m,z^m_1=z^{m,*}_{R_m},\cdots,z^m_{l_m}=z^{m,*}_{R_m+l_m-1}+(1-\theta)z^{m,*}_{l_m},z^m_{l_m+1}=z^{m,*}_{l_m+1},\cdots,z^m_{R_m-1}=z^{m,*}_{R_m-1},\alpha_m,0,\cdots,0)
\end{align}\\ 
Step 3: Under assumption \ref{assump:buffer_size_2}, we have $w_L^{k}=w_R^{k} \geq w_n^{k^{'}}$ for all $k$ and $k^{'}$ and for $0\leq n\leq R_{k^{'}}-1$.\\
That means, we will schedule all the $\alpha$ queues at state $R_k$ (i.e. group B), and we can therefore go to any state  less than $R_k-1$.\\
For the remaining  $1-2\alpha$ queues that are in state $0$ (i.e. group C), after applying a passive action (no transmission), their states will change to any state  less than or equal to $R_k-1$.\\
For group A ($\alpha$ proportion of queues), we have for each $k$:\\
1) For each state from $l_k+1$ until $R_k-1$; they stay at same state ($0$ arrivals).\\
2) For $h$ from $R_k$ until $R_k+l_k-1$, the proportion $z_h^{k,*}$ goes from state $h-(R_k-1)$ to $h$ after that $R_k-1$ packets arrive.\\
3) For $k=m$ and $h=l_m$: $(1-\theta)z_{l_m}^{m,*}$ proportion stays at same state (0 arrivals).\\
So after this step: we will reach the optimal $\textbf{z}^*$ of the relaxed problem:
The queue state proportion vector for class $k \neq m$ is:
\begin{equation}z^{k,*}=(z_0^{k,*},z_1^{k,*},....,z^{k,*}_{l_k+R_k-1},0,......0)\end{equation}
The queue state proportion vector for class $m$ is:
\begin{equation}z^{m,*}=(z_0^{k,*},z^{m,*}_{1},......,z^{m,*}_{l_m+R_m-1},0,......0)\end{equation}
This implies that we have reached the optimal proportion $\textbf{z}^*$. \\

If $\alpha > \frac{1}{2}$:\\

Step 1: the same step as we did when $\alpha \leq \frac{1}{2}$, however all $1-\alpha$ queues (group B) that are not scheduled will be at state $1$ since $1-\alpha < \alpha$.
Hence the new queue state proportions vector after this step for $k \neq m$ is:
\begin{equation}z^k=(0,z^k_1=\beta_k,0,0,\cdots, z^k_{R_k-1}=\sum_{i=l_k+1}^L z_i^{k,*},0,\cdots,0)\end{equation}
For $k=m$:
\begin{equation}z^m=(0,z^m_1=\beta_m,0,0,\cdots,\sum_{i=l_m+1}^L z_i^{m,*}+(1-\theta)z_{l_m}^{m,*},0,\cdots,0)\end{equation}
with $\sum \beta_k=1-\alpha$ \\
Step 2: The group A is scheduled again, and the $1-\alpha$ proportion of queues at state $1$ (group B), which are not scheduled, will go to state $R_k$. 
For $l_k+1 \leq h \leq R_k-1$, $z_h^{k,*}$ will be at state $h$, after scheduling.\\
For $R_k \leq h \leq R_k+l_k-1$, $z_h^{k,*}$ will be at state $h-(R_k-1)$, and $(1-\theta)z_{l_m}^{m,*}$ will be at state $l_m$.\\
Hence after this step, the queue state proportion vector for class $k \neq m$ is:
\begin{align}
(0,z^k_1=z_{R_k}^{k,*},\cdots,z^k_{l_k}=z^{k,*}_{R_k+l_k-1},z^k_{l_k+1}=z^{k,*}_{l_k+1},\cdots,z^k_{R_k-1}=z^{k,*}_{R_k-1},\beta_k,0,\cdots,0)
\end{align} 
For $k=m$: \begin{align}
(0,z^m_1=z^{m,*}_{R_m},\cdots,z^m_{l_m}=z^{m,*}_{R_m+l_m-1}+(1-\theta)z^{m,*}_{l_m},z^m_{l_m+1}=z^{m,*}_{l_m+1},\cdots,z^m_{R_m-1}=z^{m,*}_{R_m-1},\beta_m,0,\cdots,0)
\end{align}
Step 3:\\
Using the Whittle's Index policy, we schedule $(1-\alpha)$ proportion of queues at state $R_k$ (group B), plus proportion among the group A. We divide the group A into two disjoint proportions $A_1$ and $A_2$, where  $A_2$ is defined as the set that contains all proportions $z^k_{1}$ till $z^k_{l_k}$  for each $k$ minus part from $z^m_{l_m}$ which is $(1-\theta)z^{m,*}_{l_m}$. Explicitly, replacing $z^k_{i}$ by its value at step 2, we have $A_2=\sum_{k \neq m} \sum_{i=R_k}^{R_k+l_k-1} z^{k,*}_i+\sum_{i=R_m}^{R_m+l_m-1} z^{m,*}_i$. Since we have proved that this sum is less than $1-\alpha$ according to Lemma \ref{lem:inequality_alpha}, then, we can be sure that the whole proportion is not scheduled at step 3. Since $B=1-\alpha < \alpha$ and $B+A_1=1-A_2 > \alpha$,  we  just need to schedule in addition to B, a proportion from $A_1$ called $A_{11}$. In fact we will choose the $A_{11}$ highest Whittle index's queues among $A_1$ such that $A_{11}+B= \alpha$. We note $A_{12}=A_1-A_{11}$. Hence in this step the proportion scheduled is $A_{11}+B$ and the proportion for which we take a passive decision is $A_{12}+A_2$. However we still need to prove that the Whittle index of proportions $A_2$ is less than that of  the $\alpha$ proportion scheduled states (i.e. B plus $A_{11}$). \\
For the group B at state $R_k$, $w^k_{R_k} \geq w^{k'}_n$ for all $k$ and $k'$ and $0\leq n \leq R_{k'}-1$, then the Whittle index of all other queues state belonging to either $A_1$ or $A_2$ are less than the one of queue state belonging to group B.\\
For $A_1$: their states are surely among the states $l_k+1,.....,R_k-1$ for all $k$, plus the state $l_m$. Hence, the Whittle index of any of these states is higher or equal than $w^*$, with $w^*$ is the optimal subsidy for the relaxed problem (following the definition of the optimal threshold vector $l$), that is also true for $A_{11}$.\\
For the proportion $A_2$, the whole proportion is at a state that has an index less or equal than $w^*$ that is less than the Whittle indices of proportion $A_{11}$. Hence, the Whittle indices of proportion $A_2$ is less that the Whittle indices of proportion $A_{11}$ and B. By definition of $A_{12}$, the Whittle indices of this proportion is less than the Whittle indices of the proportion $A_{11}$ and by consequence less than those of the proportion B. This confirms that the whole proportion $A_{12}+A_2$ is not scheduled.\\
For the proportion $(1-\alpha)$ (group B) at $R_k$, the group of queues can go to any state less than $R_k-1$ after scheduling. In fact, their states will go to all states less than $l_k$ for each $k$ according to the optimal proportion vector $\textbf{z}^*$, except for the state $l_m$ at class $m$ for which only $\theta z^{m,*}_{l_m}$  goes to state $l_m$. \\
For $A_{11}$: the queues in this group will stay at the same states. In fact, for each class $k$, the states of the queues are all less than $R_k$. Then by scheduling these queues, the departure will be equal to the queue length. On the other hand, by considering that the number of arrival packets is equal to the previous queue length, one can ensure  that the states of the queues in this group remain unchanged. \\%arrival   number of packets transmitted from the queues will be at on one hand, the scheduling can change their states to any state less than $R_k-1$. In fact, after scheduling, all $A_1$ will be at state $0$. But for the arrival packets at this step, \textbf{we add the same numbers packets as what exist before (these queues are not changed after this step).} \\
For  $A_{12}$: Not scheduling the queues in this group implies that they will stay at the same state considering the number of packet arrival is $0$.\\
For $A_2$: This group is not scheduled. The state of the queues in class $k$ will change by adding $R_k-1$ arrival packets to their previous length.\\
Consequently, after this step, the new queue state proportion vector:\\
for $k \neq m$: 
\begin{equation}z^{k,*}=(z_0^{k,*},z_1^{k,*},....,z^{k,*}_{l_k+R_k-1},0,......0)\end{equation}
for $k=m$:
\begin{equation}z^{m,*}=(z_0^{k,*},z^{m,*}_{1},......,z^{m,*}_{l_m+R_m-1},0,......0)\end{equation}
which means that we have reached the optimal proportion vector $\textbf{z}^*$. 
\section{Proof of Theorem \ref{prop:global_optimality}}\label{app:global_optimality}
\begin{equation} 
\lim_{T \rightarrow \infty} \frac{C_T^{N}(\textbf{x})}{N}-\frac{C^{RP,N}}{N} = \sum_{k=1}^{K} \sum_{i=0}^{L} a_k \mathbb{E}\left[Z_{i}^{k,N}(\infty) \right]i-\sum_{k=1}^K \sum_{i=0}^L a_kz_i^{k,*}i 
\end{equation} 
We have the function $f: \ z \rightarrow \sum_{k=1}^K\sum_{i=0}^L a_kz_i^k i$ is lipchitz and continuous, then for an arbitrary small $\epsilon$, there exists $\mu$ such that if $||\textbf{z}-\textbf{z}^*||<\mu$, then $|f(\textbf{z})-f(\textbf{z}^*)|<\epsilon$.

We denote $U_N$ the event {${\text{sup}} ||\textbf{Z}^N(\infty)-\textbf{z}^*|| \geq \mu$}, then :
\begin{align}
 |\sum_{k=1}^{K} \sum_{i=0}^{L} a_k \mathbb{E}\left[Z_{i}^{k,N}(\infty)\right]i-\sum_{k=1}^K \sum_{i=0}^L a_kz_i^{k,*}i|\leq& P(U_N) \mathbb{E}\left[|\sum_{k=1}^{K} \sum_{i=0}^{L} (a_kZ_{i}^{k,N}(\infty)i) - a_kz_i^{k,*}i|| U_N\right]\nonumber\\ 
&+(1-P(U_N)) \mathbb{E}\left[|\sum_{k=1}^{K} \sum_{i=0}^{L} (a_kZ_{i}^{k,N}(\infty)i) - a_kz_i^{k,*}i|| \overline{U_N}\right]\nonumber\\
\leq& L(L+1)\sum_{k=1}^K a_k\gamma_k P(U_N)+(1-P(U_N))\epsilon
\end{align}
According to Lemma \ref{lem:prob_z_infty_near_neighborhood_z_star}, we have $\lim_{N \rightarrow \infty}P(U_N)=0$, then:
\begin{equation} 
\lim_{N \rightarrow \infty} |\sum_{k=1}^{K} \sum_{i=0}^{L} a_k \mathbb{E}\left[Z_{i}^{k,N}(\infty)\right]i-\sum_{k=1}^K \sum_{i=0}^L a_kz_i^{k,*}i| \leq \epsilon 
\end{equation} 
This is true for any $\epsilon$. Finally we have:
\begin{equation} \lim_{N \rightarrow \infty} |\lim_{T \rightarrow \infty} \frac{C_T^{N}(\textbf{x})}{N}-\frac{C^{RP,N}}{N}|=0
 \end{equation} 
That completes the proof. 
\end{appendices}
\end{document}